%% file: main.tex
\documentclass[a4paper, twoside]{scrartcl}
\input{Packages}
\input{Macros}

\usepackage{etex}
\usepackage[a4paper, top=88pt, bottom=88pt, left=72pt, right=72pt, headsep=16pt, footskip=28pt]{geometry}

\usepackage{hyperref}
\hypersetup{
	colorlinks,
	linkcolor=blue,
	citecolor=blue,
	urlcolor=blue,
}

\newcommand*{\arXiv}[1]{\bgroup\color{blue}\href{https://arxiv.org/abs/#1}{arXiv:#1}\egroup}

\ExplSyntaxOn
\NewDocumentCommand{\doi}{m}
 {
  \group_begin:
  \tl_set:Nn \l_tmpa_tl {#1}

  \tl_replace_all:Nnn \l_tmpa_tl {\_} {\c_underscore_str}
  \regex_replace_once:nnN { \s+ \z } {} \l_tmpa_tl
  \regex_replace_once:nnN { [\.,] \z } {} \l_tmpa_tl

  \tl_replace_all:Nnn \l_tmpa_tl {<} {\c_percent_str 3C}
  \tl_replace_all:Nnn \l_tmpa_tl {>} {\c_percent_str 3E}
  \tl_replace_all:Nnn \l_tmpa_tl {(} {\c_percent_str 28}
  \tl_replace_all:Nnn \l_tmpa_tl {)} {\c_percent_str 29}
  \tl_replace_all:Nnn \l_tmpa_tl {:} {\c_percent_str 3A}
  \tl_replace_all:Nnn \l_tmpa_tl {;} {\c_percent_str 3B}

  \href{https\c_colon_str//doi.org/\tl_use:N \l_tmpa_tl}
       {\textcolor{blue}{doi:\texttt{\detokenize{#1}}}}
  \group_end:
 }
\ExplSyntaxOff

\newcommand*{\email}[1]{\bgroup\color{blue}\href{mailto:#1}{#1}\egroup}
\usepackage{enumitem, moreenum}
\setlist[enumerate]{nosep}
\setlist[itemize]{nosep}
\usepackage{mleftright} \mleftright
\renewcommand{\qedsymbol}{$\blacksquare$}
\renewenvironment{proof}[1][\proofname]{\noindent{\bfseries\sffamily #1.} }{\hfill\qedsymbol\medskip}

\usepackage[labelfont={sf,bf}]{caption}
\usepackage{scrlayer-scrpage, xhfill}
\automark[section]{section}
\setkomafont{pageheadfoot}{\normalcolor\sffamily}
\setkomafont{pagenumber}{\normalfont\normalsize\sffamily}
\clearpairofpagestyles
\let\oldtitle\title
\renewcommand{\title}[1]{\oldtitle{#1}\newcommand{\theshorttitle}{#1}}
\newcommand{\shorttitle}[1]{\renewcommand{\theshorttitle}{#1}}
\let\oldauthor\author
\renewcommand{\author}[1]{\oldauthor{#1}\newcommand{\theshortauthor}{#1}}
\newcommand{\shortauthor}[1]{\renewcommand{\theshortauthor}{#1}}
\cohead{\xrfill[0.525ex]{0.6pt}~\theshorttitle~\xrfill[0.525ex]{0.6pt}}
\cehead{\xrfill[0.525ex]{0.6pt}~\theshortauthor~\xrfill[0.525ex]{0.6pt}}
\newcommand{\theabstract}[1]{\par\bgroup\noindent\textbf{\textsf{Abstract.}} #1\egroup}
\newcommand{\thekeywords}[1]{\par\smallskip\bgroup\noindent\textbf{\textsf{Keywords.}}\newcommand{\and}{ $\bullet$ } #1\egroup}
\newcommand{\themsc}[1]{\par\smallskip\bgroup\noindent\textbf{\textsf{2020 Mathematics Subject Classification.}}\newcommand{\and}{ $\bullet$ } #1\egroup}

\newcommand*{\affilref}[1]{\ref{affiliation#1}}
\newcommand*{\affiliation}[3]{
	\footnotetext[#1]{\label{affiliation#2}#3}
}

\setlist{topsep=0.3ex, itemsep=0.3ex}

\title{Mixture-Weighted Ensemble Kalman Filter with Quasi-Monte Carlo Transport}
\shorttitle{Mixture-Weighted EnKF with QMC Transport}
\author{%
	Ilja~Klebanov\textsuperscript{\affilref{FUB}}%
	\and%
	Claudia~Schillings\textsuperscript{\affilref{FUB}}%
	\and%
	Dana~Wrischnig\textsuperscript{\affilref{FUB}}%
}
\shortauthor{I.~Klebanov, C.~Schillings, and D.~Wrischnig}
\date{\today}

\begin{document}
\maketitle
\affiliation{1}{FUB}{Freie Universit{\"a}t Berlin, Arnimallee 6, 14195 Berlin, Germany 
\newline
(\email{ilja.klebanov@fu-berlin.de}, \email{c.schillings@fu-berlin.de}, \email{dana.wrischnig@fu-berlin.de})}

\begin{abstract}\small
\theabstract{\input{./chunk-abstract.tex}}
\thekeywords{\input{./chunk-keywords.tex}}
\themsc{\input{./chunk-msc.tex}}
\end{abstract}

\section{Introduction}

Sequential Bayesian filtering aims to approximate the filtering distribution \(p(x_t| Y_t)\) for a latent state sequence \((x_t)_{t \in \bN_{0}}\) in $\mathbb{R}^d$ generated by a stochastic state-evolution model, given observations \(Y_t=(y_1,\dots,y_t)\).
Particle methods approximate this distribution by a finite \emph{ensemble} represented as a (possibly weighted) empirical measure,
\[
p(x_t | Y_t) \approx \sum_{i=1}^N w_t^{(i)}\,\delta_{x_t^{(i)}},
\qquad \sum_{i=1}^N w_t^{(i)}=1,
\]
where \(\delta_x\) denotes the Dirac measure at \(x\).

Most sequential filters operate by iterating a two-stage cycle---\emph{prediction} followed by \emph{analysis} (or \emph{update}).
The prediction step is essentially the same across methods: propagate the ensemble \((x_{t-1}^{(i)})_{i=1}^N\) through the dynamics with process noise to obtain the forecast (prior) ensemble \((\hat x_t^{(i)})_{i=1}^N\).
The methods then differ primarily in the analysis step, where the new observation \(y_t\) is incorporated via a Bayesian update.
Two canonical realizations of this paradigm are the Bootstrap Particle Filter (BPF; \citealt{gordon1993novel,doucet2001SMC,chopin2020SMC}) and the Ensemble Kalman Filter (EnKF; \citealt{Evensen1994Sequential,Evensen2003}); \Cref{tab:BPF_vs_EnKF} contrasts their update cycles (see \Cref{sec:filtering} for details).
In the analysis step, the BPF performs no transport \((\tilde x_t^{(i)}=\hat x_t^{(i)})\) but corrects toward the data via likelihood-based reweighting followed by resampling. This yields consistency under broad conditions, yet can suffer from weight degeneracy and large Monte Carlo variance, especially in high dimensions.
The EnKF, by contrast, \emph{transports} the forecast particles to analysis (posterior) states \((\tilde x_t^{(i)})\) using a Kalman-type gain \(K_t\), typically retains equal weights, and omits resampling, an approach that scales well but is exact only in the linear-Gaussian setting and may be biased otherwise.
Both particle filters and ensemble Kalman filters face challenges in high-dimensional problems; in the EnKF, practical applicability can often be achieved through covariance localization, cf.\ \cite{Houtekamer2001sEnKF,Hamill2001DistanceDependentFiltering,Morzfeld2017EnKFcollapse}.

\newcolumntype{M}[1]{>{\centering\arraybackslash}m{#1}}
\begin{table}[t]
\caption{Comparison of the Bootstrap Particle Filter (BPF) and the Ensemble Kalman Filter (EnKF).}
\label{tab:BPF_vs_EnKF}
\centering
\renewcommand{\arraystretch}{1.15}
\begin{tabular}{@{} M{2.8cm} M{2.3cm} M{0.29\linewidth} M{0.29\linewidth} @{}}
    \toprule
    & & \textbf{BPF} & \textbf{EnKF} \\
    \midrule    
    \textbf{PREDICTION} &
    &    
    \multicolumn{2}{c}{%
        \begin{minipage}{0.45\linewidth}\centering
            Propagate the ensemble $(x_{t-1}^{(i)})_{i=1}^N$ through the dynamics with process noise to obtain the forecast (prior) ensemble $(\hat x_t^{(i)})_{i=1}^N$.
        \end{minipage}%
    }
    \\
    \midrule    
    \multirow{2}{*}[-4.5ex]{\textbf{ANALYSIS}}
    & \textbf{Transport}
    & \shortstack{$\tilde x_t^{(i)}=\hat x_t^{(i)}$\\ (no transport)}
    & \shortstack{Prior particles $(\hat x_t^{(i)})_{i=1}^N$\\ are shifted to the\\ posterior states $\tilde x_t^{(i)}$.}
    \\[6ex]    
    & \textbf{Reweight \& Resample}
    & \shortstack{Weight by likelihood,\\ normalize, then resample to\\ obtain equally weighted $x_t^{(i)}$.}
    & \shortstack{$x_t^{(i)}=\tilde x_t^{(i)}$\\ (no reweighting \\ or resampling)}
    \\
    \bottomrule
\end{tabular}
\end{table}
In this work, we leverage the complementary strengths of both paradigms: we retain the EnKF's transport to steer particles toward regions of high posterior density, and we restore a principled correction via importance sampling (IS), optionally followed by resampling, thereby combining the EnKF's computational efficiency with the posterior accuracy of particle methods while mitigating weight degeneracy as well as biases due to nonlinearities and non-Gaussianity, as illustrated by the following diagram:

\begin{figure}[t]
\centering
\begin{tikzcd}[row sep=1.5em, column sep=4.5em]
\textup{\textbf{BPF}} & \textup{\textbf{EnKF}} & \textup{\textbf{Combined}} \\
\big( x_{t-1}^{(i)}, \tfrac{1}{N} \big)_{i=1}^{N} \arrow[d, "\textup{prediction}", right] &
\big( x_{t-1}^{(i)}, \tfrac{1}{N} \big)_{i=1}^{N} \arrow[d, "\textup{prediction}", right] &
\big( x_{t-1}^{(i)}, \tfrac{1}{N} \big)_{i=1}^{N} \arrow[d, "\textup{prediction}", right]
\\
\big( \hat{x}_{t}^{(i)}, \tfrac{1}{N} \big)_{i=1}^{N} \arrow[d, equals] &
\big( \hat{x}_{t}^{(i)}, \tfrac{1}{N} \big)_{i=1}^{N} \arrow[d, "\textup{transport}", right] &
\big( \hat{x}_{t}^{(i)}, \tfrac{1}{N} \big)_{i=1}^{N} \arrow[d, "\textup{transport}", right] \\
\big( \tilde{x}_{t}^{(i)}, \tfrac{1}{N} \big)_{i=1}^{N} \arrow[d, "\textup{reweighting}", right] &
\big( \tilde{x}_{t}^{(i)}, \tfrac{1}{N} \big)_{i=1}^{N} \arrow[dd, equals] &
\big( \tilde{x}_{t}^{(i)}, \tfrac{1}{N} \big)_{i=1}^{N} \arrow[d, "\textup{reweighting}", right]
\\
\big( \tilde{x}_{t}^{(i)}, w_{t}^{(i)} \big)_{i=1}^{N} \arrow[d, "\textup{resampling}", right] &
{} &
\big( \tilde{x}_{t}^{(i)}, w_{t}^{(i)} \big)_{i=1}^{N} \arrow[d, "\textup{resampling}", right] \\
\big( x_{t}^{(i)}, \tfrac{1}{N} \big)_{i=1}^{N} &
\big( x_{t}^{(i)}, \tfrac{1}{N} \big)_{i=1}^{N} &
\big( x_{t}^{(i)}, \tfrac{1}{N} \big)_{i=1}^{N}
\end{tikzcd}
\caption{Schematic update at time $t$ for the bootstrap particle filter (BPF), the ensemble Kalman filter (EnKF), and the combined scheme discussed in this paper.
Conceptually, the combined scheme performs prediction, then transport, then reweighting and resampling. The BPF corresponds to skipping the transport step, while the EnKF corresponds to skipping the reweighting-and-resampling step.}
\label{fig:bpf_enkf_combined_flow}
\end{figure}
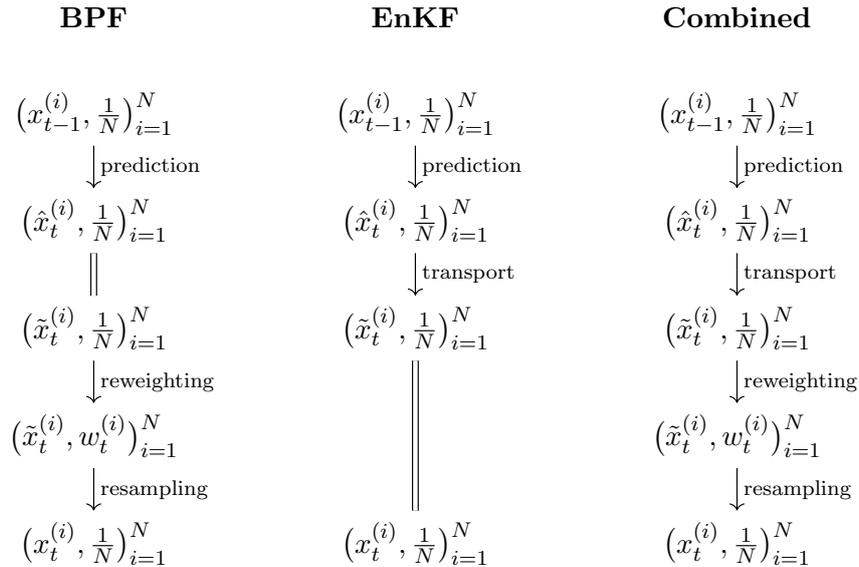

In order to perform an importance-reweighting step, one must specify the proposal distribution explicitly---i.e., the law from which the transported particles \((\tilde x_t^{(i)})\) are actually sampled. A well-established choice is the \emph{weighted EnKF} (WEnKF; \citealt{Papadakis2010DataAssimilation,vanLeeuwen2019ParticleFilters}), which treats each particle separately and takes as proposal the conditional law on the \emph{previous} ensemble,
\[
\qip = p\big(\tilde x_t^{(i)} \, \big| \, (x_{t-1}^{(j)})_{j=1}^N\big),
\]
which admits a closed-form Gaussian expression when the observation function \(h\) is linear and the gain is constructed accordingly.

In this paper we build on that paradigm and extend it in two complementary directions. 
First, beyond conditioning on the previous ensemble, we also work with the \emph{current–ensemble} conditional
\[
\qic = p\big(\tilde x_t^{(i)} \, \big| \, (\hat x_t^{(j)})_{j=1}^N\big),
\]
which is available in closed form even for nonlinear \(h\) \citep{Frei2013BridgingEKFandPF} and thereby supports importance reweighting for arbitrary observation functions. 
Second, we move from per-particle conditionals to \emph{mixture formulations}: the transported ensemble can be viewed as a stratified (balanced) sample from the Gaussian mixtures
\[
\qmp = \frac{1}{N}\sum_{i=1}^N \qip
\qquad\text{or}\qquad
\qmc = \frac{1}{N}\sum_{i=1}^N \qic,
\]
and similarly the Gaussian mixture prior induces a mixture target.
Using these mixtures---rather than only the individual Gaussian components---leads to several novel \emph{mixture-weighted} EnKF schemes that smooth across proposals and can thereby reduce weight variance.

Altogether, these choices (individual vs.\ mixture targets and proposals; conditioning on current vs.\ previous ensembles) yield six self–normalized IS (SNIS) schemes for the EnKF (one of them being WEnKF) that we analyze and compare, both theoretically and numerically.

\medskip

A second main contribution of this paper is to replace the \emph{random} generation of forecast and analysis ensembles by \emph{transported quasi--Monte Carlo} (TQMC) point sets tailored to the Gaussian--mixture laws that arise in EnKF-type methods.
Specifically, both the forecast ensemble \( (\hat x_t^{(i)})_{i=1}^N\) and the transported ensemble \((\tilde x_t^{(i)})_{i=1}^N\) can be viewed as (stratified) samples from explicit Gaussian mixtures.
Building on the transport framework of \citet{klebanov2023transporting}---which provides an essentially closed--form map sending a reference density \(\rho_{\mathsf{ref}}\) to any mixture of its affine transforms---we generate these ensembles by pushing low-discrepancy point sets, e.g.\ randomized quasi--Monte Carlo (RQMC) points, through the corresponding mixture transports.
This substitution preserves the algorithmic structure of BPF/EnKF (the prediction-analysis pipeline and the same target/proposal laws) while reducing sampling error: low-discrepancy point sets cover the state space more uniformly than independent random draws and, for sufficiently regular integrands, yield provably higher convergence rates than standard Monte Carlo (MC) or Markov chain Monte Carlo (MCMC) methodologies; see, e.g., \citet{owen2013mc}.
In our setting this manifests as lower variance of the corresponding empirical estimators and, consequently, more stable and accurate filtering, as confirmed by our numerical experiments.

\paragraph{Contributions.}
We make the following contributions.

\begin{enumerate}[label=(C\arabic*),leftmargin=3em]

\item \textbf{Importance sampling with mixture targets and proposals.}
In \Cref{sec:IS_with_mix_targets_and_proposals} we develop a unified IS framework for \emph{mixture} targets and \emph{mixture} proposals.
It yields several unbiased estimators (including stratified/balanced variants) together with variance comparisons that identify when mixture-based reweighting improves over per-component schemes.
This theory motivates mixture-weighted reweighting in EnKF-type filters.

\item \textbf{Explicit EnKF proposals via current- and previous-ensemble conditionals.}
We cast the stochastic EnKF analysis step as sampling from explicit Gaussian proposals.
For linear observations we introduce a \emph{previous-ensemble} gain $K_t^{\mathsf p}$, independent of the forecast noises, so that the \emph{previous-ensemble} conditional
$p(\tilde x_t^{(i)} \mid (x_{t-1}^{(j)})_{j=1}^N)$ is \emph{exactly} Gaussian, addressing the mean-field caveat in the common WEnKF construction.
For general (possibly nonlinear) $h$ we formalize the \emph{current-ensemble} conditional
$p(\tilde x_t^{(i)} \mid (\hat x_{t}^{(j)})_{j=1}^N)$, which enables principled importance reweighting beyond the linear setting of the WEnKF.

\item \textbf{Mixture-weighted EnKF schemes and consistency.}
Combining (i) individual vs.\ mixture targets, (ii) individual vs.\ mixture proposals, and (iii) conditioning on current vs.\ previous ensembles yields six self-normalized IS-EnKF variants (including WEnKF as a special case).
We prove consistency for all six schemes under a standard finite-variance condition on the IS weights (\Cref{thm:convergence}).

\item \textbf{Weight-variance optimality and finite-variance conditions.}
Let $\sigma_t^2$ denote the squared coefficient of variation of the (unnormalized) importance weights at time $t$
(equivalently, the variance of the weights after rescaling them to have mean one), which controls the SNIS error bounds in our convergence proof.
We show that $\sigma_t$ is minimized (among our schemes) by using mixture targets together with mixture proposals (\Cref{thm:weight_variance_inequality}).
Moreover, for linear $h$ and under bounded drift $f$, we establish the key condition $\sigma_t<\infty$ required for consistency for two of the schemes considered (\Cref{thm:IS_weight_bound}).

\item \textbf{TQMC for Gaussian mixtures.}
We propose replacing random sampling in both prediction and analysis by \emph{transported quasi-Monte Carlo} (TQMC) for Gaussian mixtures (\Cref{sec:QMC_for_EnKF}).
This preserves the BPF/EnKF pipeline while lowering sampling error and stabilizing moment and gain estimates.
We obtain TQMC-enhanced variants of BPF, EnKF, and their mixture-weighted variants, with substantial accuracy gains in low-to-moderate dimensions for consistent methods, while performing at least comparably in higher-dimensional tests.

\item \textbf{Numerical evidence.}
Across benchmarks, we find that mixture-based reweighting reduces weight degeneracy and improves estimation accuracy relative to per-component reweighting (including WEnKF-style baselines).
Moreover, in regimes where the plain EnKF exhibits an $N$-independent error plateau due to analysis--target mismatch, the weighted schemes exhibit the expected decay with $N$.

\end{enumerate}

\paragraph{Outline.}
After discussing related work and introducing notation in \Cref{sec:Related_Work,sec:Notation}, \Cref{sec:IS_with_mix_targets_and_proposals} develops a general importance-sampling theory for \emph{mixture} targets and proposals, including variance comparisons between individual- and mixture-based estimators (\Cref{thm:mixture_IS_schemes_comparison}). \Cref{sec:filtering} formulates the filtering setting and reviews the BPF/EnKF prediction--analysis pipeline, interpreting the stochastic EnKF analysis step as sampling from explicit Gaussian proposals obtained by ensemble conditioning and an appropriate gain construction.
Building on this, \Cref{sec:EnKF_reweighting_schemes} introduces six self-normalized IS--EnKF schemes, summarized in \Cref{table:EnKF_reweighting_options} and \Cref{algo:importance_sampling_strategies}. \Cref{sec:generalization_localized_EnKF} embeds these methods into a broader class of ensemble-based filters---including localized and inflated variants---via a generic proposal mechanism, and \Cref{sec:convergence} establishes consistency and error bounds, including weight-variance comparisons (\Cref{thm:weight_variance_inequality}) and finite-variance conditions (\Cref{thm:IS_weight_bound}). Next, \Cref{sec:QMC_for_EnKF} presents transported quasi-Monte Carlo (TQMC) constructions for the Gaussian-mixture laws governing prediction and analysis, yielding TQMC-enhanced variants summarized in \Cref{algo:TQMC_EnKF}. Finally, \Cref{sec:Numerics} compares the proposed mixture-weighted and TQMC-enhanced filters in numerical experiments on benchmark models, and \Cref{sec:Conclusion} concludes.

\section{Related Work}
\label{sec:Related_Work}

The Ensemble Kalman Filter (EnKF)~\citep{Evensen1994Sequential,Evensen2003} is widely used for sequential Bayesian estimation in high-dimensional systems. While computationally efficient, its Gaussian-linear assumptions lead to biased estimates in nonlinear or non-Gaussian settings. This has motivated a range of approaches that combine EnKF with particle filtering ideas.

Hybrid methods introduce importance weighting into the EnKF~\citep{Frei2013BridgingEKFandPF,ParticleKalman}, but may suffer from weight degeneracy unless resampling is used. Transform-based approaches employ optimal transport~\citep{reich2013ensemble}, reducing degeneracy at higher computational cost. Gaussian mixture formulations~\citep{Stordal2011Bridging,Frei2011MixtureEnKF} aim to capture non-Gaussian features but scale less favorably in high dimensions. 

Localization strategies have been proposed to improve scalability~\citep{Robert2017LocalizedEnKPF,Chen2020LWEnKF}, while further hybrid variants target specific regimes such as moderately non-Gaussian priors~\citep{Grooms2021HybridEnKPF} or Lagrangian data assimilation~\citep{Slivinski2015HybridLagrangian}. On the theoretical side, convergence and stability results have been established for ensemble Kalman particle filters~\citep{DelMoral2021AnalysisEnKPF} and optimal transport formulations of the EnKF~\citep{Taghvaei2019OTEnKF}. Comparative studies such as ~\citep{Pasetto2012EnKFvsPF} further highlight trade-offs between EnKF and particle filter methods.

The present work builds on these developments by introducing mixture-based importance weighting strategies for the EnKF. By combining ensemble transport with principled reweighting using mixture proposals, we extend weighted variants of the EnKF~\citep{Papadakis2010DataAssimilation,vanLeeuwen2019ParticleFilters}, improving robustness in nonlinear and non-Gaussian settings while maintaining computational scalability.
In practice, however, scalability relies on covariance localization (and often inflation). While we do not include localization in our experiments, our theoretical framework also covers localized variants, as discussed in \Cref{sec:generalization_localized_EnKF}.

\section{Notation}
\label{sec:Notation}

Throughout the manuscript we make two standard abuses of notation.
First, we do not distinguish between a probability distribution on $\bR^{d}$
and its density with respect to Lebesgue measure: if a random variable
$x$ has density $p$, we write both $x \sim p$ and $p(x)$ for its density at
$x\in\bR^{d}$.
Second, we reuse the letter $p$ for different probability densities, with the
arguments indicating which law is meant; for instance, $p(x)$ denotes the
density of $x$, while $p(y | x)$ denotes the conditional density of $y$
given $x$.

We write $I_{d}$ for the $d\times d$ identity matrix, $\ran A$ for the range of a matrix $A$ and $\SPD{d}$ and $\SPSD{d}$ for the sets of symmetric positive definite and symmetric positive semidefinite matrices $A \in \bR^{d\times d}$, respectively.
For $C\in\SPD{d}$ and $z\in\bR^{d}$ we use the induced norm
\[
|z|_{C}
\coloneqq
\bigl(z^{\top} C^{-1} z\bigr)^{1/2}
=
\bigl\|C^{-1/2}z\bigr\|_{2}.
\]
For integrals over $\bR^{d}$ we use the shorthand
\[
\int \psi \coloneqq \int_{\bR^{d}} \psi(x)\,\mathrm{d}x.
\]
For random vectors $x$ in $\bR^{d_x}$, $y$ in $\bR^{d_y}$ and $z$ in $\bR^{d_z}$ with finite second moments,
$\bE[x]$ and $\bV[x]$ denote expectation and variance of $x$ (we only use $\bV[x]$ when $d_{x}=1$), while
the (cross\nobreakdash-)covariance matrices are defined as
\[
\Cov[x] \coloneqq \bE\bigl[(x-\bE[x])(x-\bE[x])^{\top}\bigr],
\qquad
\Cov[x,y] \coloneqq \bE\bigl[(x-\bE[x])(y-\bE[y])^{\top}\bigr],
\]
with $\Cov[x \, | \, z],\, \Cov[x,y \, | \, z]$ denoting the corresponding conditional (cross-)covariances.
If \allowbreak$g\colon\bR^{d}\to\bR$ is integrable under the probability density $p \colon \bR^{d} \to \bR$, we use the additional notation
\[
\bE_{p}[g] \coloneqq \int p\, g = \bE[g(x)],
\qquad
x\sim p.
\]
Given vectors $x^{(i)}\in\bR^{d_x}$ and $y^{(i)}\in\bR^{d_y}$, $i=1,\dots,N$, with sample means
$
\bar x \coloneqq \frac{1}{N}\sum_{i=1}^{N} x^{(i)},
\,
\allowbreak \bar y \coloneqq \frac{1}{N}\sum_{i=1}^{N} y^{(i)}
$,
we define the empirical (cross-)covariance by
\[
\Cov^{\emp}\bigl[(x^{(i)})_{i=1}^{N},(y^{(i)})_{i=1}^{N}\bigr]
\coloneqq
\frac{1}{N-1}\sum_{i=1}^{N} \bigl(x^{(i)}-\bar x\bigr)\bigl(y^{(i)}-\bar y\bigr)^{\top}\in\bR^{d_x\times d_y},
\]
and abbreviate $
\Cov^{\emp}\bigl[(x^{(i)})_{i=1}^{N}\bigr]
\coloneqq
\Cov^{\emp}\bigl[(x^{(i)})_{i=1}^{N},(x^{(i)})_{i=1}^{N}\bigr]
$.

\section{Importance Reweighting with Mixture Targets and Proposals}
\label{sec:IS_with_mix_targets_and_proposals}

When estimating expected values $I=\bE_p[g]$, but direct sampling from the target distribution $p$ is not possible, importance sampling (IS) provides a strategy to instead sample from a different \emph{proposal} distribution $q$ with $p \ll q$, and to compensate for the discrepancy between $p$ and $q$ by employing importance weights $w(x)=p(x)/q(x)$. This leads to the IS estimator
\[
\hat I_{\mathrm{IS}}
=\frac1N\sum_{i=1}^N w(x^{(i)})\,g(x^{(i)}),
\qquad
x^{(i)} \stackrel{\textup{i.i.d.}}{\sim} q.
\]
In the setting common to Bayesian statistics, where $p$ is only known up to a normalization constant, one employs self-normalized IS (SNIS), where the weights above are renormalized to sum to one \citep{owen2013mc,rubinstein2016simulation}. 

Multiple importance sampling (MIS) considers a single target $p$ and several proposals $(q_i)_{i=1}^N$, which gives rise to a range of different IS strategies \citep{veach1995optimally}. Among others, one can employ \emph{random–mixture} sampling, drawing $x^{(i)}\stackrel{\mathrm{i.i.d.}}{\sim} q_{\mix}:=N^{-1}\sum_i q_i$, or \emph{deterministic/stratified} sampling, drawing one point from each $q_i$.

In our setting, both the target and the proposal densities are equally weighted mixtures,
\begin{equation}
\label{equ:target_and_proposal_equal_mixtures}
p_{\mix} = \frac{1}{N} \sum_{i=1}^{N} p_i,
\qquad
q_{\mix} = \frac{1}{N} \sum_{i=1}^{N} q_i.
\end{equation}
This section focuses on deriving and comparing several different unbiased importance sampling estimators based on the following heuristic approximation of integrals by empirical averages:
\begin{equation}
	\label{equ:IS_standard_derivation}
	\bE_{p_{\mix}}[g]
	=
	\int p_{\mix} \, g
	=
	N^{-1} \sum_{i=1}^{N} \int \check{p}_{i}\, g
	=
	N^{-1} \sum_{i=1}^{N} \int \frac{\check{p}_{i}\, g}{\check{q}_{i}} \check{q}_{i}
	\ \approx\ 
	N^{-1} \sum_{i=1}^{N} \frac{\check{p}_{i}\, g}{\check{q}_{i}} (x_{i}),
\end{equation}	
where either $\check{p}_{i} = p_{i}$ or $\check{p}_{i} = p_{\mix}$, either $\check{q}_{i} = q_{i}$ or $\check{q}_{i} = q_{\mix}$, and, for each $i$, $x_{i} \sim \check{q}_{i}$ is an independent draw from $\check{q}_{i}$.
This leads to the following five unbiased importance sampling estimators (first letter stands for individual ($\mytt{I}$) or mixture ($\mytt{M}$) target $\check{p}_{i}$, the second one for individual or mixture proposal $\check{q}_{i}$)
\begin{align*}
\II(g) &\coloneqq \frac{1}{N} \sum_{i=1}^{N} \frac{p_i(x_i)}{q_i(x_i)} \, g(x_i),
&
\MI(g) &\coloneqq \frac{1}{N} \sum_{i=1}^{N} \frac{p_{\mix}(x_i)}{q_i(x_i)} \, g(x_i),
&
&x_{i} \stackrel{\textup{indep.}}{\sim} q_{i};
\\
\IM(g) &\coloneqq \frac{1}{N} \sum_{i=1}^{N} \frac{p_i(x_i)}{q_{\mix}(x_i)} \, g(x_i),
&
\MM(g) &\coloneqq \frac{1}{N} \sum_{i=1}^{N} \frac{p_{\mix}(x_i)}{q_{\mix}(x_i)} \, g(x_i),
&
&x_{i} \stackrel{\textup{i.i.d.}}{\sim} q_{\mix};
\\
&&
\MM^{\strat}(g) &\coloneqq \frac{1}{N} \sum_{i=1}^{N} \frac{p_{\mix}(x_i)}{q_{\mix}(x_i)} \, g(x_i),
&
&x_{i} \stackrel{\textup{indep.}}{\sim} q_{i}.
\end{align*}
The last estimator, $\MM^{\strat}$, differs from $\MM$ in the way the points $x_{i}$ are drawn from $q_{\mix}$---instead of i.i.d.\ samples from $q_{\mix}$, one sample is chosen from each component.
This procedure is referred to as stratified, balanced or deterministic sampling, and the resulting estimator $\MM^{\strat}$ is also unbiased with a provably smaller variance than $\MM$ \citep[Proposition~5.5.1]{rubinstein2016simulation}.

\begin{remark}
In principle, one could also define a stratified version $\IM^{\strat}$ of the estimator $\IM$.
However, replacing the i.i.d.\ samples $x_i \sim q_{\mix}$ by stratified samples $x_j \sim q_j$ leads to a double-sum estimator in \eqref{equ:IS_standard_derivation} of the form
\[
\bE_{p_{\mix}}[g]
=
N^{-1} \sum_{i=1}^{N} \int \frac{p_{i}\, g}{q_{\mix}} \, q_{\mix}
=
N^{-2} \sum_{i=1}^{N} \sum_{j=1}^{N} \int \frac{p_{i}\, g}{q_{\mix}} \, q_{j}
\ \approx\
N^{-2} \sum_{i=1}^{N} \sum_{j=1}^{N} \frac{p_{i}\, g}{q_{\mix}}(x_j),
\]
This requires evaluating all pairs $(i,j)$ and thus has quadratic complexity.
In contrast, the $\MM^{\strat}$ estimator remains linear, which is why we do not pursue $\IM^{\strat}$ here.
\end{remark}

\begin{theorem}
\label{thm:mixture_IS_schemes_comparison}
Let $p_{i},q_{i} \colon \bR^{d} \to \bR$, $i=1,\dots,N$ be probability density functions, let $p_{\mix},q_{\mix}$ given by \eqref{equ:target_and_proposal_equal_mixtures} and let $g \in L^2(p_{\mix})$.
Assume that $p_{\mix} \ll q_{i}$ for each $i=1,\dots,N$.
Then $\II(g), \MI(g), \IM(g), \MM(g), \MM^{\strat}(g)$ are unbiased estimators of $I = \mathbb{E}_{p_{\mix}}[g]$ satisfying
\begin{equation}
    \label{equ:mixture_targets_and_proposals_variance_inequalities}
    \bV[\MM^{\strat}(g)]
    \leq
    \bV[\MM(g)]	\leq \min( \bV[\IM(g)] , \bV[\MI(g)] ).
\end{equation}
Moreover, apart from \eqref{equ:mixture_targets_and_proposals_variance_inequalities}, there is no further general variance inequality relating these five estimators.
\end{theorem}

\begin{proof}
The unbiasedness of the estimators is straightforward (cf.\ \eqref{equ:IS_standard_derivation}).
The first inequality in \eqref{equ:mixture_targets_and_proposals_variance_inequalities} is an application of the law of total variance,
\[
\bV[\MM(g)]
=
\bV[\MM^{\strat}(g)]
\;+\;
\frac{1}{N}\,
\bV_{j \sim \mathrm{Unif}(\{1,\dots,N\})}
\left[
\bE_{q_j}\!\left[
\frac{p_{\mix}}{q_{\mix}}\, g
\right]
\right],
\]
for details see \citep[Proposition~5.5.1]{rubinstein2016simulation}.
\newline		
In order to compare $\MI(g)$ with $\MM(g)$, consider their variances:
\begin{align*}
    \bV[\MI(g)] &= \frac{1}{N^{2}} \sum_{i=1}^{N} \bigg( \mathbb{E}_{q_i} \left[ \left( \frac{p_{\mix} \,g }{q_i} \right)^2 \right] - I^2 \bigg),
    &
    \bV[\MM(g)] &= \frac{1}{N} \bigg( \mathbb{E}_{q_{\mix}} \left[ \left( \frac{p_{\mix} \, g}{q_{\mix}} \right)^2 \right] - I^2 \bigg).
\end{align*}
Then, $\bV[\MM(g)] \leq \bV[\MI(g)] $ follows from
\begin{equation*}
    \mathbb{E}_{q_{\mix}} \left[ \left( \frac{p_{\mix}\, g}{q_{\mix}} \right)^2 \right]
    =
    \frac{1}{N}	\int  \frac{ (\sum_{i}  p_{\mix} g)^{2}}{ \sum_{i} q_{i}} 
    \leq
    \frac{1}{N} \int  \sum_{i=1}^{N}  \frac{(p_{\mix} g)^{2}}{q_{i}}
    =
    \frac{1}{N} \sum_{i=1}^{N} \mathbb{E}_{q_i} \left[ \left( \frac{p_{\mix}\, g}{q_i} \right)^2 \right],
\end{equation*}
where the inequality above is a pointwise application of Titu's lemma (also known as Cauchy–Schwarz inequality in Engel form or Sedrakyan's inequality):
\begin{equation}
    \label{equ:titus_lemma}
    \frac{(\sum_{i} a_i)^{2}}{\sum_{i} b_i}
    \leq
    \sum_{i} \frac{a_i^2}{b_i}
    \qquad
    \text{whenever } a_{i} \in \bR,\, 
    b_i > 0 \text{ for all } i.
\end{equation}
\newline
In order to compare $\MM(g)$ with $\IM(g)$, we denote $I_{i} \coloneqq \bE_{p_{i}}[g] = \int p_{i}\, g$ and derive
\begin{align*}
    \bV[\MM(g)]
    &=
    \frac{1}{N} \, \bE_{q_{\mix}} \bigg[ \Big( \frac{p_{\mix}\, g}{q_{\mix}} - I \Big)^{2} \bigg]
    \\
    &=
    \frac{1}{N} \, \bE_{q_{\mix}} \bigg[ \Big( \frac{1}{N} \sum_{i=1}^{N} \frac{p_{i}\, g}{q_{\mix}} - I_{i} \Big)^{2} \bigg]
    \\
    &\leq
    \frac{1}{N} \, \bE_{q_{\mix}} \bigg[ \frac{1}{N} \sum_{i=1}^{N} \Big( \frac{p_{i}\, g}{q_{\mix}} - I_{i} \Big)^{2} \bigg]
    \\
    &=
    \frac{1}{N^{2}} \sum_{i=1}^{N} \bE_{q_{\mix}} \bigg[ \Big( \frac{p_{i}\, g}{q_{\mix}} - I_{i} \Big)^{2} \bigg]
    \\
    &=
    \bV[\IM(g)],
\end{align*}
where the inequality is a (pointwise) application of Jensen's inequality to the convex function $x \mapsto x^{2}$.
The discrete counterexamples\footnote{For simplicity these counterexamples are given on the discrete two-point space $X=\{0,1\}$, but they can be embedded without difficulty into the setting of \Cref{thm:mixture_IS_schemes_comparison}.} collected in \Cref{sec:counterexamples_variance_inequalities} show that no additional general ordering between the variances of $\II, \MI, \IM, \MM$ and $\MM^{\strat}$ can hold.

\end{proof}

\begin{remark}
Clearly,
\begin{enumerate}[label = (\roman*)]
    \item
    \label{item:rem_all_targets_coincide}
    if all targets coincide, $p_{1} = \cdots = p_{N}$, then $\II = \MI$ and $\IM = \MM$;
    \item
    if all proposals coincide, $q_{1} = \cdots = q_{N}$, then $\II = \IM$ and $\MI = \MM$;
    \item
    if each proposal exactly matches its target, $q_{i} = p_{i}$ for all $i$, then $\II = \MI = \MM^{\strat}$ reduce to the stratified plain Monte--Carlo estimator (without importance reweighting).
\end{enumerate}
Case~\ref{item:rem_all_targets_coincide} corresponds to classical \emph{multiple importance sampling} (MIS) with a single target $p$ and multiple proposals $(q_i)_{i=1}^{N}$: here $\II=\MI$ is the standard na\"{\i}ve multi-proposal estimator, while $\IM=\MM$ and $\MM^{\strat}$ are the (random vs.\ stratified) mixture-sampling variants with weights $p/q_{\mix}$ \citep{veach1995optimally,Owen2000SafeIS,owen2013mc}.
Our framework additionally allows \emph{mixture targets}, and these results extend partially to the case of \emph{unequal mixture weights} and \emph{unequal per--component sample counts}, both of which are common in MIS.
\end{remark}

\begin{remark}
In this section we have assumed normalized target densities, so all importance sampling (IS) schemes are used in their standard form and yield unbiased estimators.
In our later applications the targets are Bayesian posteriors known only up to a normalizing constant, so we employ the corresponding self-normalized IS estimators.
These are generally biased but asymptotically unbiased and consistent, so the variance comparisons in this section should be understood as describing the
large-sample behavior of the corresponding self--normalized estimators.
Further comparisons of the self-normalized schemes in terms of their weights are given in \Cref{thm:weight_variance_inequality} in the context of filtering.
\end{remark}

\section{The Filtering Problem: BPF and EnKF}
\label{sec:filtering}

We consider the filtering problem for a state-space model, where we observe a sequence of measurements \( (y_t)_{t=1}^T \) in $\bR^{m}$ generated by an underlying, unobserved state sequence \( (x_t)_{t=0}^T \) in $\bR^{d}$.
For $t =1,\dots,T$, the dynamics and observation models are given by:
\begin{alignat}{2}
\label{equ:filtering_perturbed_dynamics}
x_{t}
&=
f(x_{t-1}) + \eta_{t},&
\qquad
\eta_{t}
&\sim
\mathcal{N}(0, Q),
\\
\label{equ:filtering_perturbed_observation}
y_t
&=
h(x_t) + \epsilon_t,&
\epsilon_t
&\sim
\mathcal{N}(0, R),
\end{alignat}
as illustrated in \Cref{fig:filtering_problem_dynamics_observations},
where we assume \( x_0 \sim p_0 \) and the i.i.d.\ noise sequences \( (\eta_t) \) and \( (\epsilon_t) \) to be independent, and $Q \in \SPD{d}$ and $R \in \SPD{m}$ to be symmetric and strictly positive definite covariance matrices.
Here, \( f \colon \bR^{d} \to \bR^{d} \) drives the underlying dynamics, and \( h \colon \bR^{d} \to \bR^{m}\) is the observation function.
We assume both $f$ and $h$ to be continuous and we often write $H$ in place of $h$ whenever it is linear, as is common in the literature.
We denote by \( Y_t \coloneqq (y_1, \ldots, y_t) \) the observed data up to time \( t \in \bN \) and introduce the conditional distributions
\begin{equation}
\label{equ:filtering_prior_and_posterior}
\pi_{t}^{\mathsf{prior}} = p(x_{t}|Y_{t-1}),
\qquad
p_{t}^{\mathsf{post}} = p(x_{t}|Y_{t}),
\qquad
t \in \bN.
\end{equation}
The objective of the filtering problem is to determine the \emph{filtering distribution} $p_{t}^{\mathsf{post}}$, which is typically intractable in high-dimensional systems.
Instead of maintaining an explicit density function, it is usually approximated using a finite set of (possibly weighted) samples, or \emph{particles}, $(x_{t}^{(i)})_{i=1}^{N}$ forming an ensemble representation in the Monte Carlo sense.
This ensemble is evolved through time, with prediction and update steps ensuring that it remains a reasonable approximation of the true posterior:
\begin{enumerate}
\item \textbf{Prediction Step}: Propagate the distribution forward using the system dynamics.
\item \textbf{Analysis Step}: Incorporate the new observation using Bayes' theorem to update the distribution.
\end{enumerate}
In the prediction step, propagating the particles is straightforward:
simply apply the dynamics function \( f \) and add independent noise to each particle:
\begin{equation}
\label{equ:prediction_step}
\hat{x}_{t}^{(i)} = f(x_{t-1}^{(i)}) + \eta_{t}^{(i)},
\qquad
\eta_{t}^{(i)} \sim \mathcal{N}(0, Q).
\end{equation}
This step produces an ensemble of particles $(\hat{x}_{t}^{(i)})_{i=1}^{N}$ that represents the prior distribution $\pi_{t}^{\mathsf{prior}}$ at the next time step.
For the update (or analysis) step, where the observation \( y_{t} \) is used to adjust the particle ensemble to reflect the posterior distribution $p_{t}^{\mathsf{post}}$, the methodologies differ based on the filtering approach.
In this paper, we concentrate on two common methods, the bootstrap particle filter (BPF; \citealt{gordon1993novel}) and the ensemble Kalman filter (EnKF; \citealt{Evensen1994Sequential,Evensen2003}), summarized in \Cref{algo:BPF_and_EnKF}:

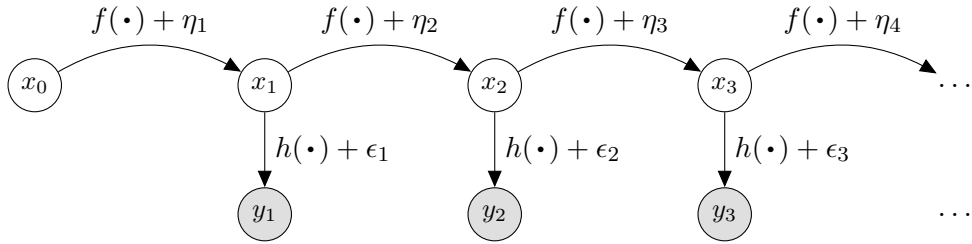
\begin{figure}[t]
\centering
\begin{tikzpicture}[scale=1, transform shape]
    \node[latent] (X0) {$x_0$};
    \node[latent, right=6em of X0] (X1) {$x_1$};
    \node[latent, right=6em of X1] (X2) {$x_2$};
    \node[latent, right=6em of X2] (X3) {$x_3$};
    \node[right=6em of X3] (dotsX) {$\dots$};
    
    \node[obs, below=6ex of X1] (Y1) {$y_1$};
    \node[obs, below=6ex of X2] (Y2) {$y_2$};
    \node[obs, below=6ex of X3] (Y3) {$y_3$};
    \node[right=6em of Y3] (dotsY) {$\dots$};
    
    \path[->]
    (X0) edge[bend left] node[above] {$f(\quark) + \eta_1$} (X1)
    (X1) edge[bend left] node[above] {$f(\quark) + \eta_2$} (X2)
    (X2) edge[bend left] node[above] {$f(\quark) + \eta_3$} (X3)
    (X3) edge[bend left] node[above] {$f(\quark) + \eta_4$} (dotsX)
    
    (X1) edge[->] node[right] {$h(\quark) + \epsilon_1$} (Y1)
    (X2) edge[->] node[right] {$h(\quark) + \epsilon_2$} (Y2)
    (X3) edge[->] node[right] {$h(\quark) + \epsilon_3$} (Y3);
\end{tikzpicture}
\caption{State-space model with underlying dynamics (top row) and observations (bottom row).}
\label{fig:filtering_problem_dynamics_observations}
\end{figure}

\paragraph{Bootstrap Particle Filter (BPF).}
Each particle\footnote{To align notation with the EnKF, we simply relabel the BPF forecast particles as $\tilde{x}_{t}^{(i)}$. This is entirely unnecessary for the BPF itself, but it makes the resampling step \eqref{equ:systematic_resampling} identical across the BPF and the (weighted) EnKF variants discussed later. No algorithmic change is implied.}
$\tilde{x}_{t}^{(i)} \coloneqq \hat{x}_{t}^{(i)}$ is reweighted according to its likelihood with respect to the new observation,
$w_{t}^{(i)} \propto p(y_t | \tilde{x}_{t}^{(i)})$
with
$\sum_{i=1}^N w_{t}^{(i)} = 1$,
followed by a systematic\footnote{We use \emph{systematic} rather than \emph{multinomial} resampling because, conditional on the weights, it remains unbiased but has lower conditional variance \citep[Section~3.4]{Doucet2011tutorialPF}.}
resampling step \citep[Section~3.4]{Doucet2011tutorialPF} to mitigate weight degeneracy, that is, a new set of $N$ equally weighted particles $(x_{t}^{(i)})_{i=1}^{N}$ is drawn from the distribution $\sum_{i=1}^{N} w_{t}^{(i)} \, \delta_{\tilde{x}_{t}^{(i)}}$ defined by the current set of weighted particles:
Draw \(u_{1}\sim \mathsf{Unif}(0,N^{-1})\) and set \(u_{i}\coloneqq u_{1} + (i-1)/N\) for \(i=1,\dots,N\) as well as
\begin{equation}
\label{equ:systematic_resampling}
x_{t}^{(i)} \;=\; \tilde{x}_{t}^{(\mathsf{ind}(i))},
\qquad
\mathsf{ind}(i) \;=\; \min\biggl\{\, j\in\{1,\dots,N\} : \sum_{\ell=1}^{j} w_{t}^{(\ell)} \ge u_{i} \,\biggr\}.
\end{equation}

\paragraph{Ensemble Kalman Filter (EnKF).}
The forecasted (prior) particles $(\hat x_t^{(i)})_{i=1}^N$ are shifted based on the \emph{Kalman gain} $K_{t}$, transporting them to the analysis (posterior) states $\tilde{x}_{t}^{(i)}$ which approximate the posterior without explicitly reweighting,
\begin{equation}
\label{equ:EnKF_update_step}
\tilde{x}_{t}^{(i)} = \hat{x}_{t}^{(i)} + K_{t} \big(y_t + \tilde{\epsilon}_{t}^{(i)} - h(\hat{x}_{t}^{(i)})\big),
\qquad
\tilde{\epsilon}_{t}^{(i)} \stackrel{\mathrm{i.i.d.}}{\sim} \mathcal{N}(0, R)
\end{equation}
Here, we adopted the standard stochastic EnKF with perturbed observations, which preserves ensemble spread by adding artificial observation noise $(\tilde{\epsilon}_{t}^{(i)})_{i=1}^{N}$ \citep[e.g.][]{Evensen1994Sequential,Evensen2003}, while the so-called Kalman gain
$K_{t}$
estimates $\Cov[x_{t},y_{t} | Y_{t-1} ]\, \Cov[y_{t}\, |\, Y_{t-1} ]^{-1}$
from ensemble covariances, see \Cref{sec:Kalman_gain_estimation} below.
Unless a resampling or mutation step follows, we set $x_{t}^{(i)} \coloneqq \tilde{x}_{t}^{(i)}$.

\begin{algorithm}[t]
\caption{Bootstrap Particle Filter (BPF) and Ensemble Kalman Filter (EnKF)}
\label{algo:BPF_and_EnKF}
\begin{algorithmic}[1]
    \State \textbf{Inputs:} ensemble size $N$, horizon $T$, prior $p_0$, dynamics $f$, observation function $h$, covariances $Q,R$, observations $(y_{t})_{t=1}^{T}$,
    choice of method $\in\{ \BPF, \EnKF \}$.
    \State \textbf{Initialize:} $x_{0}^{(i)} \stackrel{\mathrm{i.i.d.}}{\sim} p_0$.
    \For{$t = 1$ to $T$}
    \State \textbf{Prediction step:}
    Propagate each particle forward via \eqref{equ:prediction_step} $\ \leadsto\ (\hat{x}_{t}^{(i)})_{i=1}^{N}$.
    \State \textbf{Kalman gain estimation (not required for $\BPF$):}
    \begin{itemize}
        \item 
        If $h$ is nonlinear, compute $K_{t} = \hat C_t^{xy}(\hat C_t^y)^{-1}$ via \eqref{equ:covariance_estimates_for_general_Kalman_gain}.
        \item 
        If $h$ is linear, compute $K_t = K_{\mathsf p,t}$ via \eqref{equ:previous_based_gain_linear}.
    \end{itemize}        
    \State \textbf{Analysis step:}
    \begin{itemize}
        \item For $\BPF$ set $\tilde{x}_{t}^{(i)} = \hat{x}_{t}^{(i)}$. 
        \item For $\EnKF$ update each particle via \eqref{equ:EnKF_update_step} $\ \leadsto\ (\tilde{x}_{t}^{(i)})_{i=1}^{N}$. 
    \end{itemize}
    
    \State \textbf{Weights and Resampling:}
    \begin{itemize}
        \item For $\BPF$ set $w_t^{(i)} \propto p(y_t\mid \tilde{x}_t^{(i)})$ and normalize; resample via \eqref{equ:systematic_resampling} to get $(x_t^{(i)})_{i=1}^{N}$. 
        \item For $\EnKF$ set $x_{t}^{(i)} = \tilde{x}_{t}^{(i)}$ (no resampling).     
    \end{itemize}   
    \EndFor
\end{algorithmic}
\end{algorithm}

\paragraph{EnKF with Importance Weights.}
The EnKF avoids explicit reweighting by transporting particles directly to regions of higher posterior probability. This makes it computationally efficient and scalable to high-dimensional settings.
However, due to the derivation of the EnKF from the Gaussian conditioning formula, it is only consistent in the case of normally distributed $x_{0}$ and \emph{linear} dynamics $f$ and the observation function $h$ (and only if the noise $\eta_{t}$ and $\varepsilon_{t}$ are normally distributed).
In most other cases EnKF introduces a bias, because it ignores non-Gaussian features
of the forecast distribution and the observation.
One remedy to avoid this bias is by establishing the actual distribution that the points $\tilde{x}_{t}^{(i)}$ are drawn from and using it as the proposal distribution in a subsequent importance reweighting step, as for the well-known weighted EnKF ($\WEnKF$; \citealt{vanLeeuwen2019ParticleFilters,Papadakis2010DataAssimilation}).	
This can be viewed as combining EnKF with BPF and thereby getting the best from both worlds:
The ability of EnKF to transport/shift points to meaningful locations and the consistency of BPF.
It turns out that there is more than one option to establish the proposal distribution in a consistent and unbiased way, which is the topic of this paper.
In fact, the update rule \eqref{equ:EnKF_update_step} implies:

\begin{proposition}[conditional laws as proposals]
\label{prop:proposals_by_conditioning}
Fix $t\in\mathbb N$ and assume $\tilde\epsilon_t^{(i)}$ are independent of $(\hat x_t^{(j)})_{j=1}^N$ and of the process and observation noises in 
\eqref{equ:filtering_perturbed_dynamics},\eqref{equ:filtering_perturbed_observation} and \eqref{equ:prediction_step}.
\vspace{0.5ex}
\begin{enumerate}[label=(\alph*)]
    \item
    \label{item:proposals_by_conditioning_on_current_particles}
    \textbf{Conditioning on current particles.}
    If $K_{t}$ is a function solely of the forecast ensemble $(\hat x_t^{(j)})_{j=1}^N$, then, for each $i=1,\dots,N$,
    \begin{equation}
        \label{equ:qi_current_particles}
        \begin{aligned}
            \begin{split}
                \qic
                \coloneqq
                p(\tilde x_t^{(i)} \,\big|\, (\hat x_t^{(j)})_{j=1}^N)
                =
                \cN(\mic,\Qic),        
                \qquad        
                \mic
                &=
                \hat{x}_{t}^{(i)} + K_{t} \big(y_t - h(\hat{x}_{t}^{(i)})\big),
                \\
                \Qic
                &=
                K_{t} R K_{t}^{\top}.
            \end{split}
        \end{aligned}
    \end{equation}
    
    \item
    \label{item:proposals_by_conditioning_on_previous_particles}
    \textbf{Conditioning on previous particles for linear observations.}
    If $h(x)=Hx$ is linear and $K_t$ is a function solely of $(x_{t-1}^{(j)})_{j=1}^N$ (independent of the forecast noises $\eta_t^{(j)}$),
    then, for each $i=1,\dots,N$,
    \begin{equation}
        \label{equ:qi_previous_particles}
        \begin{aligned}
            \begin{split}
                \qip
                \coloneqq
                p(\tilde x_t^{(i)} \,\big|\, (x_{t-1}^{(j)})_{j=1}^N)
                =
                \cN(\mip,\Qip),
                \qquad
                \mip
                &=
                f(x_{t-1}^{(i)}) + K_{t} \big(y_t - H f(x_{t-1}^{(i)})\big),
                \\        
                \Qip
                &=
                (\Id - K_tH) Q (\Id - K_tH)^{\!\top} + K_{t} R K_{t}^{\top}.
            \end{split}
        \end{aligned}
    \end{equation}
\end{enumerate}
\end{proposition}

\begin{proof}
For \ref{item:proposals_by_conditioning_on_current_particles}, note that in \eqref{equ:EnKF_update_step} the only randomness conditional on $(\hat x_t^{(j)})_{j=1}^N$ and $K_t$ is $\tilde\epsilon_t^{(i)}$, linearly mapped by $K_t$, resulting in the stated Gaussian distribution.
For \ref{item:proposals_by_conditioning_on_previous_particles}, plugging $\hat x_t^{(i)}=f(x_{t-1}^{(i)})+\eta_t^{(i)}$ into \eqref{equ:EnKF_update_step} yields
\begin{align*}
    \tilde x_t^{(i)}
    &=
    f(x_{t-1}^{(i)}) + \eta_t^{(i)} + K_t \bigl(y_t + \tilde\epsilon_t^{(i)} - H f(x_{t-1}^{(i)}) - H\eta_t^{(i)}\bigr)
    \\
    &=
    \Big[ f(x_{t-1}^{(i)}) + K_t(y_t - H f(x_{t-1}^{(i)})) \Big]
    +
    \Big[ (I\!-\!K_tH)\eta_t^{(i)} \!+\! K_t \tilde\epsilon_t^{(i)} \Big].
\end{align*}
yielding the stated mean and covariance (where we used independence of $\eta_t^{(i)}$ and $\tilde\epsilon_t^{(i)}$).
\end{proof}

For linear observation functions $h$, the formulas for $\qic$ and $\qip$ appear in \citep{Frei2013BridgingEKFandPF,vanLeeuwen2019ParticleFilters,Papadakis2010DataAssimilation}; \Cref{prop:proposals_by_conditioning} unifies them and makes explicit the exact conditions under which they hold---most notably the permissible construction of $K_t$ from appropriate ensembles and the linearity requirement for $h$ in the previous-ensemble case.

As we will discuss in detail in \Cref{sec:EnKF_reweighting_schemes}, apart from distinguishing between conditioning on current or previous particles, we can also work with individual or mixture proposals, as well as with individual or mixture targets, allowing for more variety in performing importance reweighting, cf.\ \Cref{sec:IS_with_mix_targets_and_proposals}.
In this sense, our work can be seen as an extension of WEnKF, with an additional contribution of incorporating TQMC into this procedure (see \Cref{sec:QMC_for_EnKF}).

\begin{remark}[square-root/ETKF variants]
Deterministic/square-root EnKF updates avoid $\tilde\epsilon_t^{(i)}$ and adjust ensemble anomalies directly. Then the conditional proposal is singular (supported on an affine image of the forecast ensemble), which is less convenient for IS.
Our methodology therefore focuses on stochastic EnKF variants, for which the proposals admit densities.
While deterministic EnKF variants are more commonly used in large-scale applications, extending the present framework to transform filters would require handling singular proposals and is left for future work.
Localized stochastic EnKF variants~\cite{Houtekamer2001sEnKF,Hamill2001DistanceDependentFiltering} provide a practically relevant compromise, improving scalability while retaining proposal densities; as discussed in \Cref{sec:generalization_localized_EnKF}, our methodology extends naturally to these variants.
\end{remark}

\subsection{Estimating the Kalman Gain}
\label{sec:Kalman_gain_estimation}	

Let us discuss certain issues concerning the computation of the Kalman gain $K_{t}$, which, as mentioned above, attempts to approximate $\Cov[x_{t},y_{t} | Y_{t-1} ]\, \Cov[y_{t}\, |\, Y_{t-1} ]^{-1}$ .	
Since
\begin{align}
\notag
&\Cov[x_{t},y_{t} | Y_{t-1} ]
=
\Cov[x_{t},h(x_{t})| Y_{t-1} ],
&
&\Cov[y_{t}\, |\, Y_{t-1} ]
=
\Cov[h(x_{t}) | Y_{t-1} ] + R,
\intertext{it is both natural and common to estimate these covariances by}
\label{equ:covariance_estimates_for_general_Kalman_gain}
&\hat{C}_{t}^{xy}
=
\Cov^{\emp} \big[ (\hat{x}_{t}^{(i)})_{i=1}^{N} , (h(\hat{x}_{t}^{(i)}))_{i=1}^{N} \big],
&
&\hat{C}_{t}^{y}
=
\Cov^{\emp} \big[ (h(\hat{x}_{t}^{(i)}))_{i=1}^{N} \big] + R
\intertext{and to use the common Kalman gain $K_{t} = K_{t}^{\mathsf{c}}$ with}
\label{equ:common_Kalman_gain}
&K_{t}^{\mathsf{c}} = \hat{C}_{t}^{xy} (\hat{C}_{t}^{y})^{-1}.
\intertext{For linear $h(x)=Hx$, this Kalman gain reduces to}
\label{equ:Kalman_gain_linear_observation}
&K_{t}^{\mathsf{c}}
=
\hat C_t^{x} H^\top (H \hat C_t^{x} H^\top + R)^{-1},
&
&\hat{C}_{t}^{x}
=
\Cov^{\emp} \big[ (\hat{x}_{t}^{(i)})_{i=1}^{N} \big].	
\end{align}

\paragraph{Interaction issue when conditioning on previous particles.}
If $K_t=\hat C_t^{xy}(\hat C_t^{y})^{-1}$ is computed from the \emph{forecast} ensemble $(\hat x_t^{(i)})_{i=1}^N$ via \eqref{equ:common_Kalman_gain} or \eqref{equ:Kalman_gain_linear_observation}, it depends on the forecast noises $\eta_t^{(i)}$ and is therefore \emph{correlated} with the randomness in the analysis update \eqref{equ:EnKF_update_step}. Moreover, as noted by \citet[Section~3]{Papadakis2010DataAssimilation}, this dependence on the forecast ensemble induces \emph{interaction} between particles, so the Gaussian formula for $\qip$ in \Cref{prop:proposals_by_conditioning}\ref{item:proposals_by_conditioning_on_previous_particles} is only a mean-field approximation; the true conditional law $\qip$ is considerably more complex:
\begin{quote}
In the establishment of this distribution, we have neglected the dependence on the other particles (which appears through the forecast ensemble covariance) by treating the forecast ensemble covariance as a mean-field variable intrinsic to the system, independent of system realizations---an assumption that holds only asymptotically as the number of particles tends to infinity. 
\citep[Section~3]{Papadakis2010DataAssimilation}
\end{quote}

\paragraph{A consistent ``previous-ensemble'' gain for linear observations.}
For a linear observation function $h(x)=Hx$ we propose a Kalman gain $K_{t} = K_{t}^{\mathsf{p}}$ \emph{independent of the current forecast noises}:
\begin{equation}
\label{equ:previous_based_gain_linear}
K_{t}^{\mathsf{p}}
\coloneqq
\hat C^{x,\mathsf{p}}_{t} H^\top \bigl(H \hat C^{x,\mathsf{p}}_{t} H^\top + R\bigr)^{-1},
\qquad
\hat C^{x,\mathsf{p}}_{t}
\coloneqq
\Cov^{\emp} \bigl[(f(x_{t-1}^{(i)}))_{i=1}^N\bigr] + Q,
\end{equation}
where the estimate $\hat C^{x,\mathsf{p}}_{t}$ is based on the law of total covariance,
\begin{equation}
\label{equ:law_of_total_covariance_applied}
\Cov[x_t\mid Y_{t-1}] = \Cov[f(x_{t-1})\mid Y_{t-1}] + Q.
\end{equation}
By construction, $K_{t}^{\mathsf{p}}$ depends only on $(x_{t-1}^{(i)})_{i=1}^N$, hence, using $K_{t} = K_{t}^{\mathsf{p}}$ in \eqref{equ:EnKF_update_step} makes \Cref{prop:proposals_by_conditioning}\ref{item:proposals_by_conditioning_on_previous_particles} exact.
Moreover, $\hat C^{x,\mathsf{p}}_{t} \to \Cov[x_t | Y_{t-1}]$ almost surely (a.s.) as $N\to\infty$ under mild conditions:
\begin{lemma}[strong consistency]
\label{lem:consistency_Cp}
If $(x_{t-1}^{(i)})_{i=1}^N \stackrel{\textup{i.i.d.}}{\sim} p(x_{t-1} | Y_{t-1})$ and $\bE\big[ \|f(x_{t-1})\|^2 \big] < \infty$, then
\[
\hat{C}^{x,\mathsf{p}}_{t} \to \Cov[x_t\mid Y_{t-1}] \qquad \text{a.s. as } N\to\infty .
\]
\end{lemma}
\begin{proof}
Let $Z_i=f(x_{t-1}^{(i)})$. By the (conditional) strong law of large numbers,
$\bar Z_N \to \bE[Z_1\mid Y_{t-1}]$ and
$N^{-1}\sum_{i=1}^N Z_i Z_i^\top \to \bE[Z_1 Z_1^\top \mid Y_{t-1}]$ a.s., hence
$\Cov^{\emp}[(Z_i)_{i=1}^N]\to \Cov[Z_1\mid Y_{t-1}]$ a.s.
Since $\eta_t$ is independent of $x_{t-1}$, the law of total covariance \eqref{equ:law_of_total_covariance_applied} proves the claim.
\end{proof}

\noindent
Practically, the estimate $\hat{C}^{x,\mathsf{p}}_{t}$ typically yields a stabler gain (less sampling noise): compared with the forecast–ensemble covariance $\hat C_t^{x}=\Cov^{\emp}[(\hat x_t^{(i)})]$ with $\hat x_t^{(i)}=f(x_{t-1}^{(i)})+\eta_t^{(i)}$, the estimator $\hat{C}^{x,\mathsf{p}}_{t}$ uses the known $Q$ exactly and avoids sampling the additive-noise term and random cross-terms, thereby reducing the variance of the covariance estimate.
For this reason, whenever $h$ is linear, we recommend using the previous-ensemble gain $K_{t}^{\mathsf p}$ from \eqref{equ:previous_based_gain_linear} for the EnKF update---even if the proposal in the subsequent IS step is the current-ensemble law $\qic$.

\paragraph{A consistent ``current-ensemble'' gain for nonlinear observations.}
The construction above works only for linear observation functions $h(x)=Hx$.
For nonlinear $h$ it is unclear how to obtain consistent estimators $\hat C_t^{xy}\approx \Cov[x_t,y_t\mid Y_{t-1}]$ and $\hat C_t^{y}\approx \Cov[y_t\mid Y_{t-1}]$ that depend solely on $(x_{t-1}^{(j)})_{j=1}^{N}$ so that \Cref{prop:proposals_by_conditioning}\ref{item:proposals_by_conditioning_on_previous_particles} would apply.
Even if such estimators were available, the conditional law $\qip = p (\tilde x_t^{(i)} \, | \, (x_{t-1}^{(j)})_{j=1}^N )$
in general does not admit a closed-form expression, since the composition $h (f(x_{t-1}^{(i)})+\eta_t^{(i)})$ does not admit a decomposition similar to the one in the proof of \Cref{prop:proposals_by_conditioning}\ref{item:proposals_by_conditioning_on_previous_particles},
making an IS step based on $\qip$ unclear or infeasible without further approximations.

By contrast, the \emph{current-ensemble} gain $K_{t}^{\mathsf{c}}=\hat C_t^{xy}(\hat C_t^{y})^{-1}$ from \eqref{equ:common_Kalman_gain} is well defined for general $h$ and, on the full forecast ensemble $(\hat x_t^{(j)})_{j=1}^N$, yields a closed-form Gaussian proposal
$\qic = p(\tilde x_t^{(i)} \, | \, (\hat{x}_{t}^{(j)})_{j=1}^N )$ with mean and covariance as in \Cref{prop:proposals_by_conditioning}\ref{item:proposals_by_conditioning_on_current_particles}.
One of our contributions is to exploit this current-ensemble proposal $\qic$ in the reweighting strategies developed in \Cref{sec:EnKF_reweighting_schemes}.

\begin{remark}
\label{remark:restrictions_weighted_EnKF_schemes}
In summary, both conditionals $\qic$ and $\qip$ from \Cref{prop:proposals_by_conditioning} can serve as proposal distributions for importance sampling and may be paired with either gain $K_t^{\mathsf c}$ or $K_t^{\mathsf p}$. Their applicability comes with the following notes:
\begin{itemize}
\item \textbf{Linear-only:} $K_t^{\mathsf p}$ and the closed–form Gaussian expression for $\qip$ are defined only for linear observation functions $h(x)=Hx$.
\item \textbf{General but noisier:} $K_t^{\mathsf c}$ applies to general (possibly nonlinear) $h$, but its empirical construction is typically less stable.
\item \textbf{WEnKF caveat:} Combining $\qip$ with $K_t^{\mathsf c}$ (the choice made in WEnKF) violates the assumptions of \Cref{prop:proposals_by_conditioning}\ref{item:proposals_by_conditioning_on_previous_particles}; the Gaussian formula \eqref{equ:qi_previous_particles} should then be viewed as an approximation \citep[Sec.~3]{Papadakis2010DataAssimilation}.
\item \textbf{Degeneracy of $\qic$:} The covariance $\Qic$ becomes singular if $m < d$ or, when $K_t^{\mathsf c}$ is used, if $N < d+1$, which renders $\qic$ unusable as an IS proposal.
\end{itemize}
\noindent\textbf{Recommendation.} Whenever available (i.e., for linear $h$), prefer the pair $(\qip,\,K_t^{\mathsf p})$. For nonlinear observation functions, fall back to $(\qic,\,K_t^{\mathsf c})$, keeping the above restrictions in mind.
\end{remark}

\section{Importance Sampling Schemes for the Ensemble Kalman Filter}
\label{sec:EnKF_reweighting_schemes}

In the EnKF framework,
the forecast particles \( \hat{x}_{t}^{(i)} \)
and the 
updated particles \( \tilde{x}_{t}^{(i)} \)
can be viewed as stratified (or balanced/deterministic, see\ \citealt[Section~5.5]{rubinstein2016simulation}) samples from the mixtures
\begin{align}
\label{equ:mixture_prior_approximation}
\pim
&=
\frac{1}{N} \sum_{i=1}^{N} \pi_{t}^{(i)},
&
\pi_{t}^{(i)}
&\coloneqq
\cN( f(x_{t-1}^{(i)}) , Q),
&
&\hat{x}_{t}^{(i)} \stackrel{\textup{indep.}}{\sim} \pi_{t}^{(i)},
\\
\label{equ:mixture_proposal}
\qm
&=
\frac{1}{N} \sum_{i=1}^{N} \qi,
&
\qi
&=
\cN\big( \mi, \Qi \big),
&
&\tilde{x}_{t}^{(i)} \stackrel{\textup{indep.}}{\sim} \qi,
\end{align}
respectively, where $\mi$ and $\Qi$ are specified in \eqref{equ:qi_current_particles} and \eqref{equ:qi_previous_particles} below.
That is, both the approximate prior distribution \( \pim \approx \pi_{t}^{\mathsf{prior}} \) \citep[Equation~(26)]{vanLeeuwen2019ParticleFilters} and the proposal distribution \( \qm \) \citep[Equation~(1)]{Frei2013BridgingEKFandPF} are Gaussian mixtures.
It follows that the (approximate and unnormalized) posterior/target is again a mixture distribution:
\begin{align}
\label{equ:mixture_posterior_approximation}
&&p_{t}^{\mix}
&=
\ell_{t} \, \pim
=
\frac{1}{N} \sum_{i=1}^{N} p_{t}^{(i)},
&&
p_{t}^{(i)}
=
\ell_{t}\, \pi_{t}^{(i)},&&
\intertext{where, for $t \in \bN$,}
\label{equ:likelihood}
&&\ell_t(x)
&\coloneqq
\exp\left(-\tfrac12|y_t - h(x)|_R^2\right) \in (0,1],
&&
x\in \bR^d,
&&
\end{align}
denotes the (unnormalized) likelihood given observation $y_{t}$.
As we are working with \emph{stratified} proposal samples $\tilde{x}_{t}^{(i)} \stackrel{\textup{indep.}}{\sim} \qi$, the self-normalized versions of the estimators $\II$, $\MI$ and $\MM^{\strat}$ from
\Cref{sec:IS_with_mix_targets_and_proposals} are applicable.\footnote{In principle, one could also use $\IM$ and $\MM$ by drawing proposals i.i.d.\ from $\qm$ instead of stratified from $(\qi)_{i=1}^N$. This would replace the EnKF's balanced sampling with i.i.d.\ sampling and thus alter the algorithm; here we keep the EnKF prediction and update steps unchanged. A mixture-sampling analogue of $\MM$ will be discussed in \Cref{sec:QMC_for_EnKF}.}
This way the EnKF update \eqref{equ:EnKF_update_step} remains identical across all three strategies, the only difference lies in the \textit{interpretation}---each particle can either be viewed as an independent sample from \( q_{t}^{(i)} \) or, collectively, as a \textit{stratified} sampling from \( \qm \).

In addition to distinguishing between individual and mixture reweighting, there are at least two options to establish the individual proposals $\qi = \cN\big( \mi, \Qi \big)$, as has been established in \Cref{prop:proposals_by_conditioning}:
One can either condition on the current forecast ensemble $(\hat x_t^{(j)})_{j=1}^N$ via \eqref{equ:qi_current_particles}, or the previous ensemble $(x_{t-1}^{(j)})_{j=1}^N$ via \eqref{equ:qi_previous_particles},
\[
\qi
=
\qic
=
\cN(\mic,\Qic)
\qquad
\text{or}
\qquad
\qi
=
\qip
=
\cN(\mip,\Qip).
\]
Consequently, this results in \( 3 \times 2  = 6 \) different self-normalized importance sampling strategies for $p_{t}^{\mathsf{post}}$ of the form
\begin{equation}
\label{equ:post_approximation_weight_calculation}
p_{t}^{\mathsf{post}}
\approx
\sum_{i=1}^{N} w_{t}^{(i)}\, \delta_{\tilde{x}_{t}^{(i)}},
\qquad
w_{t}^{(i)} = \frac{v_{t}^{(i)}}{\sum_{j} v_{t}^{(j)}},
\qquad
v_{t}^{(i)}
=
\frac{\chposti}{\chqi}(\tilde{x}_{t}^{(i)}),
\end{equation}
summarized in \Cref{table:EnKF_reweighting_options} and \Cref{algo:importance_sampling_strategies} and denoted by
\begin{align*}
\II_{\mathsf{c}},
&&
\MI_{\mathsf{c}},
&&
\MM_{\mathsf{c}}^{\strat},
&&
\II_{\mathsf{p}},
&&
\MI_{\mathsf{p}},
&&
\MM_{\mathsf{p}}^{\strat},
\end{align*}
where
\begin{itemize}
\item 
the first letter corresponds to the choice of \( \chposti \):
\newline
\texttt{I} for ``individual'', $\chposti = \posti$, or \texttt{M} for ``mixture'', $\chposti = \postm$;
\item 
the second letter corresponds to the choice of \( \chqi \):
\newline
\texttt{I} for ``individual'', $\chqi = \qi$, or \texttt{M} for ``mixture'', $\chqi = \qm$;
\item 
the index (\textsf{c} for ``current'' or \textsf{p} for ``previous'') corresponds to the interpretation of the individual proposals $\qi$, cf. \eqref{equ:qi_current_particles} and \eqref{equ:qi_previous_particles} above.
\end{itemize}
Note that $\II_{\mathsf{p}}$ corresponds to the well-known weighted EnKF ($\WEnKF$; \citealt{vanLeeuwen2019ParticleFilters,Papadakis2010DataAssimilation}).

\begin{algorithm}[t]
\caption{Importance Sampling Schemes for the Ensemble Kalman Filter}
\label{algo:importance_sampling_strategies}
\begin{algorithmic}[1]
    \State \textbf{Inputs:} ensemble size $N$, horizon $T$, prior $p_0$, dynamics $f$, observation function $h$, covariances $Q,R$, observations $(y_{t})_{t=1}^{T}$,
    choice of IS scheme $\in\{ \II_{\mathsf{c}},\MI_{\mathsf{c}},\MM_{\mathsf{c}}^{\strat},\II_{\mathsf{p}},\MI_{\mathsf{p}},\MM_{\mathsf{p}}^{\strat}\}$ ($\mathsf{p}$-schemes require linear $h$).
    \State \textbf{Initialize:} Sample $x_{0}^{(i)} \stackrel{\mathrm{i.i.d.}}{\sim} p_0$.
    \For{$t = 1$ to $T$}
    \State \textbf{Prediction step:}
    Propagate each particle forward via \eqref{equ:prediction_step} $\ \leadsto\ (\hat{x}_{t}^{(i)})_{i=1}^{N}$.
    \State \textbf{Kalman gain estimation:}
    \begin{itemize}
        \item 
        If $h$ is nonlinear, compute $K_t = K_{t}^{\mathsf{c}}$ via \eqref{equ:common_Kalman_gain}.
        \item 
        If $h$ is linear, compute $K_t = K_{t}^{\mathsf{p}}$ via \eqref{equ:previous_based_gain_linear}.
    \end{itemize}
    \State \textbf{Build targets and proposals:}
    \begin{itemize}
        \item Define $p_{t}^{(i)} = \ell_{t}\, \pi_{t}^{(i)}$, where $\pi_{t}^{(i)} = \cN( f(x_{t-1}^{(i)}) , Q)$, cf.\ \eqref{equ:mixture_prior_approximation}--\eqref{equ:mixture_posterior_approximation}.       
        \item Define $\qi = \qic$ via \eqref{equ:qi_current_particles} for schemes of the form $\biggraybox\, \biggraybox_{\mathsf{c}}$.
        \item Define $\qi = \qip$ via \eqref{equ:qi_previous_particles} for schemes of the form $\biggraybox\, \biggraybox_{\mathsf{p}}$.
        \item 
        Define $p_{t}^{\mix} = N^{-1} \sum_{i=1}^{N} p_{t}^{(i)}$ and $\qm = N^{-1} \sum_{i=1}^{N} \qi.$
    \end{itemize}
    \State \textbf{Analysis step:}
    update each particle via \eqref{equ:EnKF_update_step} $\ \leadsto\ (\tilde{x}_{t}^{(i)})_{i=1}^{N}$.
    
    \State \textbf{Weights (self-normalized IS):} Set $w_t^{(i)} \propto \frac{\chposti}{\chqi} (\tilde{x}_t^{(i)})$ and normalize, cf.\ \eqref{equ:post_approximation_weight_calculation}, where
    \begin{itemize}
        \item $\chposti = \posti$, $\chqi = \qi$ for schemes of the form $\II_{\smallgraybox}$;
        \item $\chposti = \postm$, $\chqi = \qi$ for schemes of the form $\MI_{\smallgraybox}$;
        \item $\chposti = \postm$, $\chqi = \qm$ for schemes of the form $\MM^{\strat}_{\smallgraybox}$.
    \end{itemize}   
    \State \textbf{Resampling:} Choose $x_{t}^{(i)}$ via \eqref{equ:systematic_resampling}.
    \EndFor
\end{algorithmic}
\end{algorithm}

\begin{table}[t]
\caption{The $3 \times 2 = 6$ importance sampling strategies for the EnKF, determined by whether individual or mixture priors, as well as individual or mixture proposals, are used in the reweighting process, and by whether the proposals are conditioned on current or previous particles.
    Note that $\II_{\mathsf{p}}$ corresponds to the well-known weighted EnKF ($\WEnKF$; \citealt{vanLeeuwen2019ParticleFilters,Papadakis2010DataAssimilation}).
}  
\label{table:EnKF_reweighting_options}
\centering
\renewcommand{\arraystretch}{2.0}
\begin{tabular}{lcc}
    \toprule
    \tcbox[
    on line,
    tcbox raise base,
    colback=gray!3,
    colframe=gray,
    boxsep=2pt
    ]{
            \begin{minipage}{0.3\textwidth}
                \centering
                \textbf{EnKF Reweighting via}
                \\[0ex]
                $p_{t}^{\mathsf{post}} \approx
                \sum_{i=1}^{N} w_{t}^{(i)}\, \delta_{\tilde{x}_{t}^{(i)}}$
                \\[0ex]
                $w_{t}^{(i)} \propto \frac{\chposti}{\chqi}(\tilde{x}_{t}^{(i)}),$
            \end{minipage}
        }
        &
        \begin{tabular}{@{}c@{}}
            $\biggraybox \, \biggraybox_{\mathsf{c}}:$				
            \textbf{conditioning on}
            \\[-2ex]
            \textbf{current particles:}
            \\[-2ex]				
            \begin{tabular}{@{}c@{}}
                $\mi = \mic$
                \\[-2ex]
                $\Qi = \Qic$
            \end{tabular}
            \ cf.\ \eqref{equ:qi_current_particles}
        \end{tabular}
        &
        \begin{tabular}{@{}c@{}}
            $\biggraybox \, \biggraybox_{\mathsf{p}}:$
            \textbf{conditioning on}
            \\[-2ex]
            \textbf{previous particles:}
            \\[-2ex]
            \begin{tabular}{@{}c@{}}
                $\mi = \mip$
                \\[-2ex]
                $\Qi = \Qip$
            \end{tabular}
            \ cf.\ \eqref{equ:qi_previous_particles}
        \end{tabular}
        \\
        \midrule
        $\II_{\textcolor{lightgray}{\rule[-0.1ex]{0.45em}{1.2ex}}}:\ $					
        \begin{tabular}{@{}c@{}}
            $\chposti = \posti = \ell_{t} \, \cN( f(x_{t-1}^{(i)}) , Q)$\\[-2ex]
            $\chqi = \qi = \cN\big( \mi, \Qi \big)$
        \end{tabular}
        & 
        $\II_{\mathsf{c}}$ 
        &  
        $\II_{\mathsf{p}}$
        \\
        \midrule
        $\MI_{\textcolor{lightgray}{\rule[-0.1ex]{0.45em}{1.2ex}}}:\ $
        \begin{tabular}{@{}c@{}}
            $\chposti = \postm = \frac{1}{N} \sum_{i=1}^{N} \posti$ \\[-2ex]
            $\chqi = \qi = \cN\big( \mi, \Qi \big)$
        \end{tabular}
        & 
        $\MI_{\mathsf{c}}$ 
        &  
        $\MI_{\mathsf{p}}$
        \\
        \midrule
        $\MM^{\strat}_{\smallgraybox}:\ $
        \begin{tabular}{@{}c@{}}
            $\chposti = \postm = \frac{1}{N} \sum_{i=1}^{N} \posti$ \\[-2ex]
            $\chqi = \qm = \frac{1}{N} \sum_{i=1}^{N} \qi$
        \end{tabular}
        & 
        $\MM^{\strat}_{\mathsf{c}}$ 
        &  
        $\MM^{\strat}_{\mathsf{p}}$ 
        \\			
        \bottomrule
    \end{tabular}	
\end{table}

\subsection{Generalization to localized EnKF and other filtering algorithms}
\label{sec:generalization_localized_EnKF}

The reweighting strategy introduced in the previous subsection extends naturally to a broader class of ensemble--based particle filters, including localized variants of the EnKF and related algorithms.
The convergence analysis in the subsequent \Cref{sec:convergence} is formulated within this more general framework.

We continue to work with the state and observation models given by \eqref{equ:filtering_perturbed_dynamics}--\eqref{equ:filtering_perturbed_observation}.
At time step $t$, \Cref{algo:importance_sampling_strategies} can be viewed as drawing analysis samples $\tilde{x}_t^{(i)} \sim \qi$ from proposals $\qi = \qic$ or $\qi = \qip$, as established in \Cref{prop:proposals_by_conditioning}.
These proposals may depend on either $x_{t-1}^{(i)}$ or $\hat{x}_{t}^{(i)}$ as well as on $y_{t}$ and, through the Kalman gain $K_{t}$, on the full ensembles $(x_{t-1}^{(j)})_{j=1}^{N}$ or $(\hat{x}_{t}^{(j)})_{j=1}^{N}$, cf.\ \eqref{equ:qi_current_particles} and \eqref{equ:qi_previous_particles}.
To encompass localized EnKF and other ensemble--based filters, we replace this specific choice by a generic proposal mechanism of the form
\begin{equation}
    \label{equ:generic_proposal}
    \qi
    =
    \genq \bigg( x_{t-1}^{(i)} , \hat{x}_{t}^{(i)} , \frac{1}{N} \sum_{j=1}^{N} \delta_{x_{t-1}^{(j)}} , \frac{1}{N} \sum_{j=1}^{N} \delta_{\hat{x}_{t}^{(j)}} , y_{t}\bigg),
\end{equation}
where
$\genq
\colon
\bR^{d} \times \bR^{d} \times \cP(\bR^{d}) \times \cP(\bR^{d}) \times \bR^{m}
\to \cP_{+}(\bR^{d})$
is an arbitrary function that we call \emph{proposal function}
and $\cP_{+}(\bR^{d})$ denotes the set of probability distributions on $\bR^{d}$ with strictly positive probability densities.
As before, we define the proposal mixture $\qm \coloneqq \frac{1}{N} \sum_{i=1}^N \qi$.
Using the same three importance sampling strategies as in the previous subsection, but now with the generic \emph{proposal function} $\genq$, we obtain the schemes $\II_{\genq}$, $\MI_{\genq}$, and $\MM_{\genq}^\strat$, which are summarized in \Cref{algo:importance_sampling_strategies_general}.

Examples of ensemble-based filtering algorithms to which our reweighting framework can be applied include, in addition to the standard stochastic EnKF discussed above, the covariance--localized EnKF~\citep{Houtekamer2001sEnKF,Hamill2001DistanceDependentFiltering} and EnKF with covariance inflation~\citep{Anderson1999MCfiltering,Whitaker2012Evaluating}, which are most transparently described in the linear observation case $h(x) = Hx$.
In their simplest formulations, both methods amount to modifying the empirical forecast covariance matrix
$\hat{C}_{t}^{x}
=
\Cov^{\emp} \big[ (\hat{x}_{t}^{(i)})_{i=1}^{N} \big]$
via
\[
\hat{C}_{t}^{x,\mathrm{loc}}
\coloneqq
L \circ \hat{C}_{t}^{x},
\qquad
\hat{C}_{t}^{x,\mathrm{infl}} \coloneqq \delta_{t}^{2} \hat{C}_{t}^{x}
\]
where $L$ is a localization matrix whose entries depend on the physical distance between state components and $\circ$ denotes the Schur product, while $\delta_{t} \ge 1$ is an inflation factor.\footnote{The choice of localization and inflation parameters (e.g., length scales or inflation factors) is problem-dependent and typically guided by the spatial correlation structure of the system and the ensemble size, often requiring empirical tuning.}
Consequently, these modified covariance matrices enter the Kalman gain
$K_{t}^{\mathsf{c}}
=
\hat C_t^{x} H^\top (H \hat C_t^{x} H^\top + R)^{-1}$
in \eqref{equ:Kalman_gain_linear_observation}, where $\hat C_t^{x}$ is replaced by either $\hat{C}_{t}^{x,\mathrm{loc}}$ or $\hat{C}_{t}^{x,\mathrm{infl}}$, and thus they directly modify the Gaussian proposals in \eqref{equ:qi_current_particles} and \eqref{equ:qi_previous_particles}.
We do not include localized variants in the numerical experiments, since this would introduce additional tuning choices beyond the focus of the present work. A systematic study of how localization interacts with mixture-based importance weighting in high-dimensional settings is an interesting direction for future work.

\begin{algorithm}[t]
\caption{Importance Sampling Schemes in General Framework}
\label{algo:importance_sampling_strategies_general}
\begin{algorithmic}[1]
    \State \textbf{Inputs:} ensemble size $N$, horizon $T$, prior $p_0$, dynamics $f$, observation function $h$, covariances $Q,R$, observations $(y_{t})_{t=1}^{T}$, proposal function $\genq$,
    choice of IS scheme $\in\{ \II_{\genq},\MI_{\genq},\MM_{\genq}^{\strat}\}$
    \State \textbf{Initialize:} Sample $x_{0}^{(i)} \stackrel{\mathrm{i.i.d.}}{\sim} p_0$.
    \For{$t = 1$ to $T$}
    \State \textbf{Prediction step:}
    Propagate each particle forward via \eqref{equ:prediction_step} $\ \leadsto\ (\hat{x}_{t}^{(i)})_{i=1}^{N}$.
    
    \State \textbf{Targets} $p_{t}^{(i)} = \ell_{t} \cdot \cN( f(x_{t-1}^{(i)}) , Q)$ and $p_{t}^{\mix} = N^{-1} \sum_{i=1}^{N} p_{t}^{(i)}$.
    
    \State \textbf{Proposals} $q_t^{(i)} =  \genq \big( x_{t-1}^{(i)} , \hat{x}_{t}^{(i)} , \frac{1}{N} \sum_{j=1}^{N} \delta_{x_{t-1}^{(j)}} , \frac{1}{N} \sum_{j=1}^{N} \delta_{\hat{x}_{t}^{(j)}} , y_{t}\big)$ and $\qm = \frac{1}{N} \sum_{i=1}^{N} \qi$.
    
    \State \textbf{Analysis step:}
    Draw independent samples $\tilde{x}_t^{(i)} \sim \qi$, $i=1,\dots,N$.    
    
    \State \textbf{Weights (self-normalized IS):} Set $w_t^{(i)} \propto \frac{\chposti}{\chqi} (\tilde{x}_t^{(i)})$ and normalize, cf.\ \eqref{equ:post_approximation_weight_calculation}, where
    \begin{itemize}
        \item $\chposti = \posti$, $\chqi = \qi$ for $\II_{\genq}$;
        \item $\chposti = \postm$, $\chqi = \qi$ for $\MI_{\genq}$;
        \item $\chposti = \postm$, $\chqi = \qm$ for $\MM^{\strat}_{\genq}$.
    \end{itemize}   
    \State \textbf{Resampling:} Choose $x_{t}^{(i)}$ via \eqref{equ:systematic_resampling}.
    \EndFor
\end{algorithmic}
\end{algorithm}

\subsection{Convergence Analysis}
\label{sec:convergence}

To demonstrate that importance sampling estimates introduced in \Cref{algo:importance_sampling_strategies_general} consistently approximate the posterior distribution in a meaningful way,
we introduce the same distance between random probability density functions as \citet[Section~11.3]{SanzAlonso2023IPandDA} did for the analysis of the BPF. For two random probability density functions $p, p'$ on $\mathbb{R}^d$ we define the metric
\[
d(p, p') := \sup_{ \|g\|_\infty \le 1}
\bigg(\mathbb{E}\bigg[\bigg(
\int p \, g - \int p'\, g
\bigg)^2\bigg]\bigg)^{1/2},
\]
where the supremum is taken over all $g$ in the space $C_{b}(\bR^{d})$ of continuous bounded functions on $\bR^{d}$. Note that, in our case, the randomness of the probability densities arises from sampling in the prediction, analysis and resampling step, while we consider the observations $(y_{t})_{t=1}^{T}$ to be fixed.

Throughout the analysis, we frequently condition expectations on the `history' $\sigma$-algebra
\begin{equation}
    \label{eq:history_sigma_algebra}
    \mathcal{F}_t =
    \sigma \big( (x_{t-1}^{(j)})_{j=1}^{N} , (\hat{x}_{t}^{(j)})_{j=1}^{N} \big).
\end{equation}
generated by the resampled particles from the previous analysis step and the current forecast particles.

\begin{theorem}
    \label{thm:convergence}
    Consider the dynamics and observation models given by \eqref{equ:filtering_perturbed_dynamics}--\eqref{equ:filtering_perturbed_observation} with fixed observations $(y_{t})_{t=1}^{T} \in (\bR^{m})^{T}$, continuous observation function $h\colon \bR^d \to \bR^k$
    and any of the sampling schemes $\II_{\genq},\, \MI_{\genq},\, \MM_{\genq}^{\strat}$ defined in \Cref{algo:importance_sampling_strategies_general} for some proposal function $\genq$ of the form \eqref{equ:generic_proposal}.     
    Further, denote $p_{0}^{\mathsf{post}} \coloneqq p_0$, $c_0
    \coloneqq 1$, and, for $t = 1,\dots,T$,
    \begin{align*}
    c_t
    &\coloneqq
    2\bigg(\frac{c_{t-1}}{\Zpost} + \sigma_t  +1\bigg) \in [0,\infty],
    &
    \sigma_t^2
    &\coloneqq
    \mathbb{V}\Biggl[\frac{\check{p}_{t}^{(1)}}{\Zmix \, \check{q}_t^{(1)}}(\tilde{x}_t^{(1)}) \Biggr]
    =
    \mathbb{V}\Biggl[\frac{v_{t}^{(1)}}{\Zmix} \Biggr] \in [0,\infty], 
    \\
    \Zpost
    &\coloneqq    
    \int \ell_t\, \pi_t^{\mathsf{prior}},
    &
    \Zmix
    &\coloneqq
    \int \postm 
    = \int \ell_t\, \pim,
    \end{align*}
    where $\pi_{t}^{\mathsf{prior}}$, $p_{t}^{\mathsf{post}}$, $\pim$, $\postm$, $\posti$, $\ell_t$ and $v_{t}^{(i)}$ are defined in \eqref{equ:filtering_prior_and_posterior}, \eqref{equ:mixture_prior_approximation}, \eqref{equ:mixture_posterior_approximation}, \eqref{equ:likelihood} and \eqref{equ:post_approximation_weight_calculation}, respectively.
    Then the particle approximation of the filtering distribution $p_{t}^{\mathsf{post}}$ by the ensemble $(x_{t}^{(i)})_{i=1}^{N}$ satisfies
    \begin{equation}
        \label{equ:convergence}
        d_{t}
        \coloneqq
        d\bigg(p_{t}^{\mathsf{post}},\frac{1}{N} \sum_{i=1}^{N} \delta_{x_{t}^{(i)}} \bigg)        
        \le \frac{c_t}{\sqrt{N}}
        \qquad
        \text{for every } t = 0,\dots,T.    
    \end{equation}
\end{theorem}

\begin{proof}
The claim will be proven by induction, starting with $d_{0} \leq N^{-1/2}$ \citep[Lemma~11.2]{SanzAlonso2023IPandDA}.
For $t\in \bN$, we decompose the error $d_{t}$ into the approximation error for the filtering distribution, the (self-normalized) importance sampling error and the resampling error, each of which will be bounded by one of the three lemmas in \Cref{section:auxiliary_results_thm_convergence}:
\Cref{lem:mixture_error,lem:IS_error,lem:resampling_error}.
The proof is then a simple application of the triangle inequality:
    \begin{align*}
    d_t
    &\leq
    \underbrace{d\left(p_{t}^{\mathsf{post}}, \frac{\postm}{\Zmix} \right)}_{\shortstack{\scriptsize\textup{approximation error for}\\\scriptsize\textup{the filtering distribution}}}
    +
    \underbrace{d\left( \frac{\postm}{\Zmix} \, ,\,  \sum_{i=1}^{N} w_{t}^{(i)}\delta_{\tilde{x}_{t}^{(i)}}\right)}_{\textup{SNIS error}}
    +
    \underbrace{d\left(\sum_{i=1}^{N} w_{t}^{(i)}\delta_{\tilde{x}_{t}^{(i)}}, \frac1N\sum_{i=1}^{N}
    \delta_{x_{t}^{(i)}}\right)}_{\textup{resampling error}}
    \\[1ex]
    &
    \leq \quad
    \frac{2\, d_{t-1}}{Z_t^{\mathsf{post}}}  + \frac{2 \, \sigma_t + 1}{\sqrt{N}}  + \frac{1}{\sqrt{N}}
    \quad \leq \quad
    \frac{2}{\sqrt{N}}\left(\frac{c_{t-1}}{Z_t^{\mathsf{post}}} + \sigma_t + 1\right)
    \quad = \quad
    \frac{c_t}{\sqrt{N}} \, .
\end{align*}
\end{proof}

The critical condition for \Cref{thm:convergence} to be useful is that the SNIS weights have finite variance, $\sigma_{t}<\infty$, a standard requirement in importance sampling.
Before showing that this assumption can indeed be ensured under suitable conditions (\Cref{thm:IS_weight_bound}), we first compare the schemes in terms of the variance of their self-normalized weights and show that the mixture--mixture estimators are preferable in this sense.
This result can be seen as a complement to \Cref{thm:mixture_IS_schemes_comparison}: here we work in the self-normalized setting and compare weight variance rather than the variance of the estimators themselves, which in general still depends on the integrand $g$.

\begin{theorem}
\label{thm:weight_variance_inequality}
  Under the assumptions of \Cref{thm:convergence} and using the notation therein,
  \[
    \sigma_t \geq \sigma_t^{\mix},
    \qquad
    c_t \geq c_t^{\mix},
  \]
  where $\sigma_t^{\mix}$ and $c_t^{\mix}$ are the corresponding quantities for the scheme $\MM_{\genq}^\strat$.
\end{theorem}

\begin{proof}
Let $\mathcal{F}_t$ be given by \eqref{eq:history_sigma_algebra}. Using \Cref{lem:conditional_expectations},
we obtain in all schemes
\begin{align*}
    \sigma_t^2 
    &= \mathbb{V}\biggl[\frac{\check{p}_t^{(1)}}{\Zmix\, \check{q}_t^{(1)}} \bigl(\tilde{x}_t^{(1)}\bigr) \biggr] \\
    &= \mathbb{E}\biggl[\biggl(\frac{\check{p}_t^{(1)}}{\Zmix\, \check{q}_t^{(1)}} \bigl(\tilde{x}_t^{(1)}\bigr) \biggr)^2\biggr] 
    - 
    \mathbb{E}\biggl[\frac{\check{p}_t^{(1)}}{\Zmix\, \check{q}_t^{(1)}} \bigl(\tilde{x}_t^{(1)}\bigr) \biggr]^2\\
    &= \frac{1}{N} \sum_{i=1}^N 
       \mathbb{E}\biggl[\mathbb{E}\biggl[\biggl(\frac{\chposti}{\Zmix\, \chqi}\bigl(\tilde{x}_t^{(i)}\bigr) \biggr)^2\Bigm| \mathcal{F}_t\biggr]\biggr] - 1\\
    &= \mathbb{E} \biggl[\frac1N \sum_{i=1}^{N} \int \biggl(\frac{\chposti}{\Zmix\, \chqi}\biggr)^2 q_t^{(i)}\biggr] - 1 .
\end{align*}
The scheme $\MM_{\genq}^{\strat}$ satisfies $\sigma_t = \sigma_t^{\mix}$ by definition.
For the schemes $\II_{\genq}$ and $\MI_{\genq}$ we have $\check{q}_t^{(i)} = q_t^{(i)}$,
in which case an application of Titu's lemma stated in \eqref{equ:titus_lemma} yields
\begin{align*}
    \sum_{i=1}^{N} \int \biggl(\frac{\chposti}{\chqi}\biggr)^2 q_t^{(i)}
    &= \int \sum_{i=1}^{N} \frac{(\chposti)^2}{q_t^{(i)}} 
    \ge \int \frac{1}{N}
    \frac{\bigl(\sum_{i=1}^N \chposti\bigr)^2}{\sum_{i=1}^N q_t^{(i)}}\\
    &=\int
    \biggl( \frac{\frac1N\sum_{i=1}^N \chposti}{\frac1N\sum_{i=1}^N q_t^{(i)}}\biggr)^2
    \sum_{i=1}^N q_t^{(i)} 
    = \sum_{i=1}^N \int \biggl(\frac{\postm}{\qm}\biggr)^2 q_t^{(i)} .
\end{align*}
Therefore,
\[
    \sigma_t^2  
    = \mathbb{E} \biggl[\frac1N \frac{1}{(\Zmix)^2}\sum_{i=1}^{N} \int 
      \biggl(\frac{\chposti}{\chqi}\biggr)^2 q_t^{(i)}\biggr] - 1
    \ge \mathbb{E} \biggl[\frac{1}{N} \frac{1}{(\Zmix)^2}\sum_{i=1}^N \int 
      \biggl(\frac{\postm}{\qm}\biggr)^2 q_t^{(i)}\biggr] - 1 
    = (\sigma_t^{\mix})^2 .
\]
Since both variances are non–negative, this implies $\sigma_t \ge \sigma_t^{\mix}$.
$c_t \geq c_t^{\mix}$ follows directly from the definition.
\end{proof}

Next, we show that, for the sampling schemes $\II_{\mathsf{p}}, \ \MI_{\mathsf{p}}, \ \MM_{\mathsf{p}}^\strat$, the finite-variance condition $\sigma_{t}<\infty$ in \Cref{thm:convergence} is indeed satisfied whenever the observation function $h$ is linear and the model drift $f$ is bounded.
The assumption that $f$ is bounded is mainly technical and is not satisfied by many standard dynamical systems, such as the Lorenz models. Similar boundedness assumptions are commonly employed in the theoretical analysis of particle filtering and related methods to ensure stability and finite variance of estimators; see, e.g., \cite{Crisan}. In practice, such conditions can often be relaxed to suitable growth and moment assumptions, which we leave for future work.
Further note that we now work in the setting of the weighted EnKF from \Cref{algo:importance_sampling_strategies} rather than the more general framework of \Cref{algo:importance_sampling_strategies_general}.

\begin{theorem}
	\label{thm:IS_weight_bound}
	Let the assumptions of \Cref{thm:convergence} hold and consider the notation therein. In the setting of \Cref{algo:importance_sampling_strategies}, fix $i=1,\dots,N$ and $t=1,\dots,T$, assume that the observation function $h(x) = Hx$ is linear, that
	$K_{t} = K_{t}^{\mathsf{p}}
	= 
	\hat{C}^{x,\mathsf{p}}_{t} H^\top \bigl(H \hat{C}^{x,\mathsf{p}}_{t} H^\top + R\bigr)^{-1}$ is computed via \eqref{equ:previous_based_gain_linear}, and denote
	\[
	S \coloneqq HQH^\top + R \in \SPD{m}, \qquad
	K \coloneqq QH^\top S^{-1} \in \bR^{d \times m}, \qquad
	\Sigma \coloneqq Q - KHQ \in \bR^{d \times d}.
	\]			
	\begin{enumerate}[label = (\alph*)]
	\item
	\label{item:thm_IS_weight_bound_posti_formula}
	Denoting $\hat{m}_{t}^{(i)} = f(x_{t-1}^{(i)}) + K(y_t - Hf(x_{t-1}^{(i)}))$, we have $\Sigma \in \SPD{d}$ and
	\begin{equation}
		\label{equ:single_posterior_approx_density}
		p_t^{(i)}
		=
		\biggl(\frac{\det R}{\det S}\biggr)^{\! \! 1/2} 
		\exp \big( - \tfrac{1}{2} | y_t - H f(x_{t-1}^{(i)}) |_S^2 \big)
		\cdot \mathcal{N}(\hat{m}_{t}^{(i)}, \Sigma).
	\end{equation}
	\item
	\label{item:thm_IS_weight_bound_ratio_formula}	
	There exists $u_{t}^{(i)} \in \bR^{d}$ such that, for each $x\in \mathbb{R}^d$
	\begin{equation}
		\label{equ:relation_single_posterior_proposal_approx}
		\frac{\posti}{q_{\mathsf{p}, t}^{(i)}} (x)
		\le
		\biggl(\frac{\det  \Qip }{\det Q }\biggr)^{\! \! 1/2}
		\exp\big( -\tfrac{1}{2}(x-u_{t}^{(i)})^\top (\Sigma^{-1} - \Qip^{-1}) (x-u_{t}^{(i)}) \big),
	\end{equation}
	with equality if $\Cov^{\emp}[(f(x_{t-1}^{(i)})_{i=1}^N])$ is positive definite and $H$ has full row rank.
	\item
	\label{item:thm_IS_weight_bound_IS_weight_bound}
	For each $x\in \mathbb{R}^d$ and $i =1,\dots,N$,
	\begin{equation}
		\label{equ:weight_estimation}
		\frac{\posti}{Z^{\mathsf{mix}} \, \qip}(x)
		\leq
		\biggl(\frac{\det \hat{C}^{x,\mathsf{p}}_{t} }{\det Q}\biggr)^{\! \! 1/2}
		\biggl(\frac{1}{N} \sum_{i=1}^N
		\exp\big( -\tfrac{1}{2} |y_t - Hf(x_{t-1}^{(i)})|_S^2 \big)\biggr)^{-1}.
	\end{equation}		
	Further, for the sampling schemes $\II_{\mathsf{p}}$ and $\MM_{\mathsf{p}}^{\strat}$ defined in \Cref{algo:importance_sampling_strategies}, we have $\sigma_t^2 < \infty$ whenever $f$ is bounded.
	\end{enumerate}
\end{theorem}

\begin{proof}
	The argument proceeds in several steps which are established in \Cref{section:auxiliary_results_thm_IS_weight_bound}.
	First, we collect several identities among the EnKF matrices in \Cref{lemma:EnKF_matrix_identities}.
	Next, using these identities together with a completion-of-squares argument, we prove \ref{item:thm_IS_weight_bound_posti_formula} in \Cref{lemma:thm_IS_weight_bound_posti_formula}.
	We then establish \ref{item:thm_IS_weight_bound_ratio_formula} in \Cref{lemma:thm_IS_weight_bound_ratio_formula}---again by completing the square---taking particular care with symmetric matrices that are only positive \emph{semi}definite.
	Finally, \Cref{lemma:thm_IS_weight_bound_IS_weight_bound} applies \ref{item:thm_IS_weight_bound_ratio_formula} to the targets and proposals in the schemes $\II_{\mathsf{p}}$ and $\MM_{\mathsf{p}}^{\strat}$, yielding \ref{item:thm_IS_weight_bound_IS_weight_bound}.
\end{proof}

\section{TQMC-Enhanced EnKF}
\label{sec:QMC_for_EnKF}

As established above, both the approximate prior and the proposal used by the EnKF can be written as Gaussian mixtures:
\[
\pim = \frac{1}{N}\sum_{i=1}^N \cN\bigl(f(x_{t-1}^{(i)}),Q\bigr),
\qquad
\qm = \frac{1}{N}\sum_{i=1}^N \cN\bigl(\mi,\Qi\bigr),
\]
where $(\mi,\Qi)$ are given by \eqref{equ:qi_current_particles} or \eqref{equ:qi_previous_particles}.
Instead of drawing i.i.d.\ (or stratified) samples from these mixtures, we propose to use \emph{transported quasi Monte–Carlo} ($\TQMC$) point sets \citep{klebanov2023transporting}:
low-discrepancy points in a reference space are pushed forward through an essentially closed-form transport to the mixture distribution.
Compared with i.i.d.\ or stratified sampling, quasi Monte--Carlo (QMC) and randomized quasi Monte--Carlo (RQMC) achieve lower discrepancy and---for sufficiently smooth integrands---provably better mean-squared error rates than $O(N^{-1})$, see, e.g., \citet{owen2013mc}.

\begin{notation}[RQMC transport to mixtures]
\label{notation:RQMC_transport}
    For any Gaussian mixture $p=\sum_{k=1}^K w_k\,\cN(m_k,C_k)$, we write
    \[
    \bigl(z^{(i)}\bigr)_{i=1}^N = \TQMC_{N}(p)
    \]
    to denote $N$ RQMC points pushed forward by a transport map $T_p$ from the uniform distribution on $(0,1)^{d}$ \citep{klebanov2023transporting} so that $(z^{(i)})_{i=1}^N = (T_p(u^{(i)}))_{i=1}^N$ has sampling law $p$.
    Throughout our experiments, we choose $(u^{(i)})_{i=1}^N$ to be a scrambled Sobol' sequence (specifically, we employ the Matou\v{s}ek–Affine–Owen scrambling \citep{owen1995randomly,owen1997scrambling,Matousek1998Scrambling}).
\end{notation}

\begin{algorithm}[t]
    \caption{TQMC-Enhanced BPF and (Weighted) EnKF}
    \label{algo:TQMC_EnKF}
    \begin{algorithmic}[1]
        \State \textbf{Inputs:} ensemble size $N$, horizon $T$, Gaussian or Gaussian-mixture prior $p_0$, dynamics $f$, observation function $h$, covariances $Q,R$, observations $(y_{t})_{t=1}^{T}$, choice of scheme $\in\{ \QMC\BPF, \QMC\EnKF_{\mathsf{c}}, \QMC\MM_{\mathsf{c}}, \QMC\EnKF_{\mathsf{p}}, \QMC\MM_{\mathsf{p}} \}$ ($\mathsf{p}$-schemes require linear $h$).
        \State \textbf{Initialize:} $(x_{0}^{(i)})_{i=1}^N = \TQMC_{N}(p_0)$ and set $w_{0}^{(i)} = \tfrac{1}{N}$ for all $i=1,\dots,N$ .
        \For{$t=1$ to $T$}
        \State \textbf{Forecast mixture:}
        $\rho_{t}^{\mix} = \sum_{i=1}^{N} w_{t-1}^{(i)}\, \cN \bigl(f(x_{t-1}^{(i)}),Q\bigr)$.
        \State \textbf{Forecast ensemble via TQMC:} $(\hat x_{t}^{(i)})_{i=1}^N = \TQMC_{N}(\rho_{t}^{\mix})$.
        \State \textbf{Kalman gain estimation (not required for $\QMC\BPF$):}
        \begin{itemize}
            \item 
            If $h$ is nonlinear, compute $K_{t} = \hat C_t^{xy}(\hat C_t^y)^{-1}$ via \eqref{equ:covariance_estimates_for_general_Kalman_gain}.
            \item 
            If $h$ is linear, compute $K_t = K_{\mathsf p,t}$ via \eqref{equ:previous_based_gain_linear}.
        \end{itemize}
        \State \textbf{Build proposals (not required for $\QMC\BPF$):}
        \begin{itemize}
            \item For $\QMC\MM_{\mathsf{c}}$ and $\QMC\EnKF_{\mathsf{c}}$, set $\qm = N^{-1} \sum_{i=1}^{N} \qic$ via \eqref{equ:qi_current_particles}.
            \item For $\QMC\MM_{\mathsf{p}}$ and $\QMC\EnKF_{\mathsf{p}}$, set $\qm = \sum_{i=1}^{N} w_{t-1}^{(i)}\, \qip$ via \eqref{equ:qi_previous_particles}, cf.\ \eqref{equ:modified_mixture_proposal_for_QMC}.
        \end{itemize}        
        \State \textbf{Analysis ensemble via TQMC:}
        \begin{itemize}
            \item For $\QMC\BPF$ set $\tilde{x}_{t}^{(i)} \coloneqq \hat{x}_{t}^{(i)}$. 
            \item Otherwise, set $(\tilde x_t^{(i)})_{i=1}^N = \TQMC_{N}(\qm)$.
        \end{itemize}
        \State \textbf{Weights (self-normalized IS):}
        \begin{itemize}
            \item For $\QMC\BPF$, set $w_t^{(i)} \propto p(y_t\mid \tilde{x}_t^{(i)})$ and normalize. 
            \item For $\QMC\EnKF_{\mathsf{c}}$ and $\QMC\EnKF_{\mathsf{p}}$, set $w_t^{(i)} = \frac{1}{N}$.     
            \item For $\QMC\MM_{\mathsf{c}}$ and $\QMC\MM_{\mathsf{p}}$ set $\displaystyle w_t^{(i)} \propto \frac{p(y_{t} | \tilde{x}_t^{(i)}) \, \rho_{t}^{\mix}(\tilde{x}_t^{(i)})}{\qm(\tilde{x}_t^{(i)})}$ and normalize.
        \end{itemize}
        \State \textbf{No resampling:} Set $(x_t^{(i)})_{i=1}^N = (\tilde x_t^{(i)})_{i=1}^N $.
        \EndFor
    \end{algorithmic}
\end{algorithm}

In the standard EnKF pipeline, the \emph{resampling step} \eqref{equ:systematic_resampling} followed by the \emph{prediction} step \eqref{equ:prediction_step} is equivalent to drawing i.i.d.\ samples from the weighted Gaussian mixture
\[
\rho_{t}^{\mix} \coloneqq \sum_{i=1}^{N} w_{t-1}^{(i)}\, \cN\bigl(f(\tilde{x}_{t-1}^{(i)}),Q\bigr).
\]
With TQMC we can skip explicit resampling and directly generate the forecast ensemble as well as the proposals from $\qm$ via
\[
\bigl(\hat x_t^{(i)}\bigr)_{i=1}^N \;=\; \TQMC_{N}(\rho_t^{\mix}),
\qquad
\bigl(\tilde x_t^{(i)}\bigr)_{i=1}^N \;=\; \TQMC_{N}(\qm).
\]
As before, both $\qmc$ and $\qmp$ can be used for $\qm$, with $\qmp$ limited to \emph{linear} observation functions $h$, cf.\ \Cref{sec:Kalman_gain_estimation}, the only difference being that, in the case $\qm = \qmp$, we are conditioning on a \emph{weighted} ensemble of particles $(x_{t-1}^{(i)},w_{t-1}^{(i)})_{i=1}^{N}$ (due to the missing resampling step), and therefore
\begin{equation}
\label{equ:modified_mixture_proposal_for_QMC}
\qmp
=
\sum_{i=1}^{N} w_{t-1}^{(i)}\, \qip.
\end{equation}
The reweighting can then be performed by $\MM$ (with i.i.d.\ samples replaced by these $\TQMC$ points).
Note that resampling and prediction in BPF is identical to the one of the weighted EnKF variants, so the same TQMC-enhancement works for the generation of the forecast ensemble in BPF.
All of these methods are summarized in \Cref{algo:TQMC_EnKF}.

\begin{remark}
    Note that the transport methodology of \citet{klebanov2023transporting} is not limited to QMC and RQMC, but can be employed to transport any cubature rule (e.g.\ sparse grids or digital nets) to a mixture distribution.
    Moreover, this approach extends to the generalized framework in \Cref{sec:generalization_localized_EnKF} provided that the proposals $\qi \in \cP_{+}(\bR^{d})$ generated by \eqref{equ:generic_proposal} admit an explicit transport map from the reference distribution underlying the cubature rule (typically the uniform distribution on the unit cube) to $\qi$.
    In particular, this includes covariance-localized and covariance-inflated EnKF variants, for which the proposals remain Gaussian.
    However, as in standard QMC methodology, the effectiveness of transported QMC may deteriorate in high dimensions unless additional structure (e.g.\ low effective dimension) is present.
    It would be an interesting direction for future work to investigate whether localization can in fact be leveraged to induce such structure, in particular by reducing the effective dimension.
    For a discussion of the computational overhead of $\TQMC$ see \citep{klebanov2023transporting}.
\end{remark}

\section{Numerical Experiments}
\label{sec:Numerics}

\begin{figure}[t]
    \centering         	
    \begin{subfigure}[b]{0.48\textwidth}
        \centering
        \includegraphics[width=\textwidth]{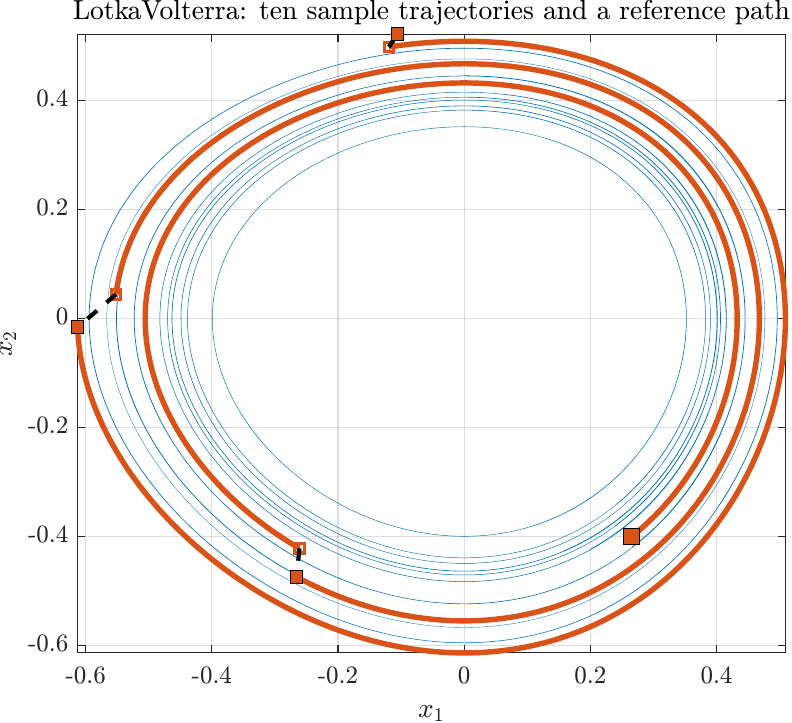}
    \end{subfigure}
    \hfill
    \begin{subfigure}[b]{0.48\textwidth}
        \centering
        \includegraphics[width=\textwidth]{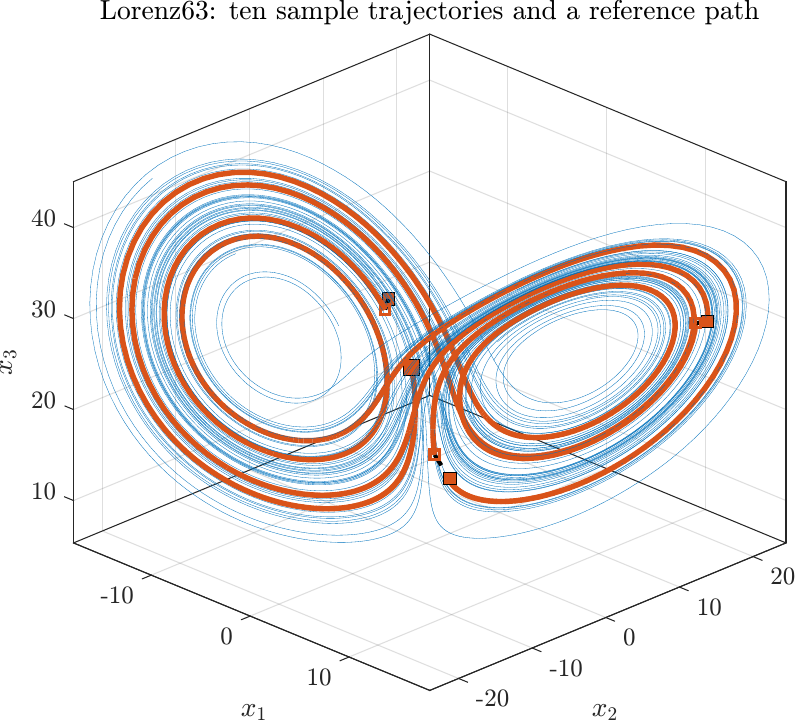}
    \end{subfigure}	
    \caption{Sample ODE trajectories (blue) and underlying (true) perturbed dynamics (red) for the Lotka--Volterra model (left) and the Lorenz--63 model (right).
    The solid red lines correspond to the ODE transport $f = \Phi^{\Delta\tau}$ in \eqref{equ:filtering_perturbed_dynamics} while the dotted black lines visualize the random perturbations $\eta_{t}$.
    }
    \label{fig:trajectories_both_systems}
\end{figure}

We study three standard filtering problems \citep{Evensen2022DA,Law2015DA} with dynamics and observation models of the form \eqref{equ:filtering_perturbed_dynamics}--\eqref{equ:filtering_perturbed_observation}:
the Lotka--Volterra model, the Lorenz--63 system and the 40-dimensional Lorenz--96 system, introduced below.
In all three models, the underlying dynamics is governed by an explicit ordinary differential equation (ODE) $z'(\tau) = \varphi(z(\tau))$ with closed-form right-hand side $\varphi \colon \bR^{d} \to \bR^{d}$, that is, denoting the flow operator of this ODE by $\Phi^{\tau}$, the dynamics model is given by $f = \Phi^{\Delta\tau}$ for a certain fixed time step $\Delta\tau > 0$.
The initial distribution $p_0 = \cN(m_{0},\Sigma_{0})$ of $x_{0}$ is chosen as a Gaussian with mean $m_{0}$ and covariance $\Sigma_{0}$.

For each system, we consider two observation models,
\begin{itemize}
    \item the (linear) identity observation function $h(x) = x$ with observation covariance $R = R_{\textup{lin}} = (10\, \gamma)^{-2}\, I_{d}$,
    \item the nonlinear observation function $h(x) = \arctan\big(\frac{\gamma x}{20}\big)$, applied componentwise, with observation covariance $R = R_{\textup{nonlin}} = 10^{-4}\, I_{d}$.
\end{itemize}
Here, $\gamma > 0$ calibrates the overall scaling across models: the state variables naturally evolve on different magnitudes (on the order of $1$, $20$, and $5$ per component for the Lotka--Volterra, Lorenz--63, and Lorenz--96 systems, respectively), and $\gamma$ is chosen accordingly so that both the observation models and the noise levels operate on comparable scales (see \Cref{table:system_parameters}).
Accordingly we set
$\Sigma_{0} = Q = \gamma^{-2}\, I_{d}$ and $g(x) \coloneqq \sin (4 \gamma  \sum_{j=1}^{d} x_j)$.
The specific values of the parameters $m_{0}$, $\Delta\tau$ and $\gamma$ are stated in the description of each system separately.

\begin{table}[!t]
    \caption{Parameters used in the numerical experiments for each dynamical system.}
    \label{table:system_parameters}
    \centering
    \renewcommand{\arraystretch}{1.2}
    \begin{tabular}{lccc}
        \toprule
        \textbf{System} & $\boldsymbol{m_0}$ & $\boldsymbol{\Delta \tau}$ & $\boldsymbol{\gamma}$ \\
        \midrule
        Lotka--Volterra
        & $\log\big(1.25,\ 0.66\bigr)^\top$        
        & $5$
        & $20$
        \\
        Lorenz--63
        & $(0,\,0,\,22)^\top$        
        & $2$
        & $1$
        \\
        Lorenz--96
        & $0$        
        & $0.5$
        & $4$
    \end{tabular}
\end{table}

For each system and observation function, each time $t=1,2,3$, each particle number\footnote{For the Lorenz--96 system with nonlinear $h$ we are forced to start with $N = 2^6$ to ensure $N > d$, cf.\ \Cref{remark:restrictions_weighted_EnKF_schemes}.} \(N \in \{2^2,\dots,2^{10}\}\) and each of $M = 10$ independent random runs, every filtering method produces an ensemble estimate $\hat{P}_{t}$ of the filtering distribution, which we compare to a reference approximation $\hat{P}_{\mathsf{ref},t}$ that is computed once per dataset using the $\QMC\MM_{\mathsf{c}}$ estimator\footnote{Strictly speaking, the $\QMC\MM_{\mathsf{c}}$ scheme, as one of the proposed methods, has not been independently validated a priori. It is used here as an empirical benchmark, supported by its stable performance and the observed convergence of other methods toward this reference solution as the ensemble size increases.} with \(N_{\mathsf{ref}}=2^{13}\) particles ($N_{\mathsf{ref}} = 2^{12}$ for the Lorenz--96 system):
\[
\hat{P}_{t}
=
\sum_{i=1}^{N} w_{t}^{(i)}\, \delta_{\tilde{x}_{t}^{(i)}},
\qquad
\hat{P}_{\mathsf{ref},t}
=
\sum_{i=1}^{N_{\mathsf{ref}}} w_{\mathsf{ref},t}^{(i)}\, \delta_{\tilde{x}_{\mathsf{ref},t}^{(i)}}.
\]
Note that the \emph{pre–resampling} analysis ensembles are used, and that the same artificially generated observation data \((y_{t})_{t=1}^{T}\) is used for all methods and each run.
The comparison will be performed using two diagnostics, which are plotted against the particle numbers \(N\) in log-log scale in \Crefrange{fig:linear_LotkaVolterra_run_1}{fig:nonlinear_Lorenz96_run_1}:
The mean absolute error (MAE)
\begin{align*}
\mathsf{MAE}_t(N)
&=
\abs{ \bE_{\hat{P}_{t}}[g] - \bE_{\hat{P}_{\mathsf{ref},t}}[g] }
=
\abs{ \sum_{i=1}^{N} w_{t}^{(i)}\, g(\tilde{x}_{t}^{(i)}) - \sum_{i=1}^{N_{\mathsf{ref}}} w_{\mathsf{ref},t}^{(i)}\, g(\tilde{x}_{\mathsf{ref},t}^{(i)}) }  \, ,
\intertext{and the squared maximum mean discrepancy (MMD) between $\hat{P}_{t}$ and $\hat{P}_{\mathsf{ref},t}$ \citep{Gretton2012MMD}}   
\mathsf{MMD}_{\kappa_{t}}^{2}\big( \hat{P}_{t} , \hat{P}_{\mathsf{ref},t} \big)
&=
\sup_{\psi \in \mathcal{H}_{\kappa_{t}}}  \abs{ \bE_{\hat{P}_{t}}[\psi] - \bE_{\hat{P}_{\mathsf{ref},t}}[\psi] }^{2}
=
\Norm{ \int \kappa_{t}( x , \quark)\, \hat{P}_{t}(\rd x) - \int \kappa_{t}( x , \quark)\, \hat{P}_{\mathsf{ref},t}(\rd x) }_{\mathcal{H}_{\kappa_{t}}}^{2}
\\
&=
w_{t}^\top K_{XX}\,w_{t} + w_{\mathsf{ref},t}^\top K_{RR}\,w_{\mathsf{ref},t} - 2\,w_{t}^\top K_{XR}\,w_{\mathsf{ref},t} \, ,
\end{align*}
where $\mathcal{H}_{\kappa_{t}}$ is a reproducing kernel Hilbert space associated with a Gaussian RBF kernel $\kappa_{t}$, the bandwidth of which is set from the reference only (to make the kernel independent of the tested method and of $N$) using a median heuristic on pairwise distances,
\[
\kappa_{t} (x,y) \coloneqq \exp \bigl(-\tfrac{\|x-y\|^2}{2\ell_{t}^2}\bigr),
\qquad
\ell_t^2 \coloneqq \frac{ \operatorname{median} \bigl\{ \norm{\tilde{x}_{\mathsf{ref},t}^{(i)} - \tilde{x}_{\mathsf{ref},t}^{(j)} }^2 \colon 1\le i<j\le N_{\mathsf{ref}}\bigr\} }{ \log N_{\mathsf{ref}} },
\]
and $[K_{XX}]_{ij} = \kappa_{t}(\tilde{x}_{t}^{(i)},\tilde{x}_{t}^{(j)})$, $[K_{RR}]_{ij}=\kappa_{t}( \tilde{x}_{\mathsf{ref},t}^{(i)} , \tilde{x}_{\mathsf{ref},t}^{(j)})$ and $[K_{XR}]_{ij}=\kappa_{t}(\tilde{x}_{t}^{(i)} , \tilde{x}_{\mathsf{ref},t}^{(j)})$.

\Cref{fig:linear_LotkaVolterra_run_1,fig:linear_Lorenz63_run_1,fig:linear_Lorenz96_run_1} present our results for the three models, respectively, under the linear observation function $h(x)=Hx=x$.
Similarly, \Cref{fig:nonlinear_LotkaVolterra_run_1,fig:nonlinear_Lorenz63_run_1,fig:nonlinear_Lorenz96_run_1} show the results under the nonlinear observation function $h(x) = \arctan\big(\frac{\gamma x}{20}\big)$.
Each figure has four rows, where the first two show the $\mathsf{MAE}$ and the last two the $\mathsf{MMD}^{2}$, with the upper row comparing BPF and EnKF against the six EnKF reweighting schemes from \Cref{algo:importance_sampling_strategies} across the three time steps $t=1,2,3$ and the lower one showing how BPF, EnKF, $\MM_{\mathsf{c}}^{\strat}$, and $\MM_{\mathsf{p}}^{\strat}$ improve when combined with TQMC points (indicated by markers), cf.\ \Cref{algo:TQMC_EnKF}.
Note that in the nonlinear case no schemes based on previous particles (those with the index $\mathsf{p}$) are applicable (cf.\ \Cref{remark:restrictions_weighted_EnKF_schemes}), which is why the corresponding lines are absent from the plots.

For ease of reference, the interpretation of the numerical results is given directly in the caption of each figure.
Overall, the mixture-weighted schemes improve accuracy and eliminate the EnKF error plateau caused by analysis--target mismatch, while the TQMC variants often yield additional gains by reducing sampling error.
Before presenting the results, we briefly introduce the three models of interest:

\paragraph{Lotka--Volterra (Predator--Prey).}
As a prototypical example of biological interaction, we study the Lotka--Volterra equations, which capture the cyclic dynamics of predator and prey through a simple nonlinear feedback mechanism. We follow the nondimensionalized formulation introduced by \cite[Section~3.1]{Murray2002}:
\[
\dot{N} = N(1 - P), \quad \dot{P} = \alpha P(N - 1),
\]
where $N$ denotes the prey population size and $P$ the predator population size and $\alpha > 0$ is a positive parameter.
To ensure the positivity of the state variable under the addition of noise, we reformulate the system in logarithmic coordinates via $z = (u, v)^{\top} = (\log N, \log P)^{\top}$:
\[
\dot{u} = 1 - e^v, \quad \dot{v} = \alpha (e^u - 1).
\]
The parameters are chosen as follows:
$\alpha = 1$ and $m_{0} = (\log(1.25), \log(0.66))^{\top}$ (in accordance with the example of \citealt[Section~3.1]{Murray2002}), $\Delta\tau = 5$ and $\gamma = 20$.

\paragraph{Lorenz--63.}
The classical Lorenz 1963 system \citep{lorenz63} is widely used as a low-dimensional benchmark for data assimilation methods due to its chaotic behavior. The system of three coupled nonlinear ordinary differential equations is given by
\[
\begin{aligned}
    \dot{u} &= \sigma (v - u), \\
    \dot{v} &= \rho u - v - u w, \\
    \dot{w} &= u v - \beta w,
\end{aligned}
\]
where $z = (u, v, w)$ represents the system state and we use the standard parameter values $\sigma = 10$, $\rho = 28$, and $\beta = 8/3$, which lead to chaotic dynamics.
The remaining parameters are chosen as follows: the initial mean is set to $m_{0} = (0,0,22)^{\top}$, the time step is $\Delta\tau = 2$, and the scaling parameter is $\gamma = 1$.

\paragraph{Lorenz--96.}
The Lorenz model from 1996 \citep{lorenz96} serves as an example for a high-dimensional chaotic system. The model consists of $d$ coupled scalar equations:
\[
\dot z_j = (z_{j+1} - z_{j-2})z_{j-1} - z_j + F,
\]
for $j = 1, \ldots, d$.
The parameters are chosen as follows:
We set $d = 40$ and consider the forcing parameter $F = 8$.
Further, $m_{0} = 0$, $\Delta\tau = 0.5$ and $\gamma = 4$.

\begin{figure}[t]
    \centering         	
    \begin{subfigure}[b]{0.99\textwidth}
        \centering
        \includegraphics[width=\textwidth]{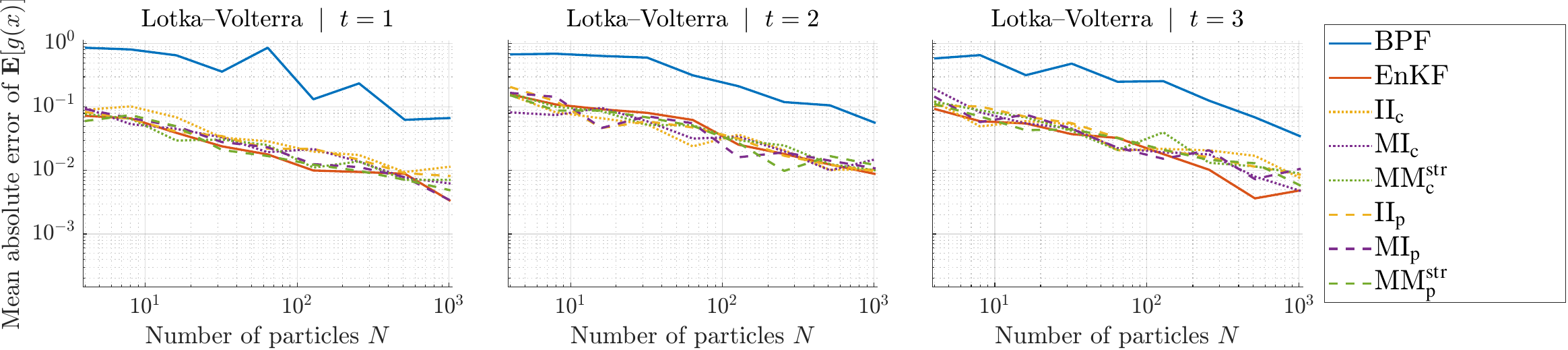}
    \end{subfigure}    
    \vfill
    \vspace{1ex}
    \vfill
    \begin{subfigure}[b]{0.99\textwidth}
        \centering
        \includegraphics[width=\textwidth]{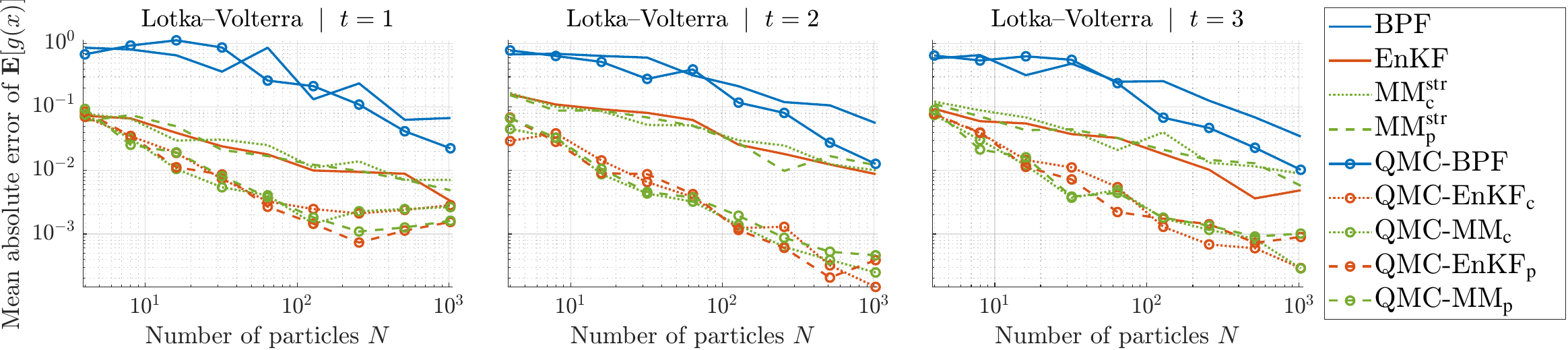}
    \end{subfigure}
    \vfill    
    \vspace{1ex}               
    \hrule
    \vspace{0.2ex}               
    \hrule    
    \vspace{1ex}     
    \begin{subfigure}[b]{0.99\textwidth}
        \centering
        \includegraphics[width=\textwidth]{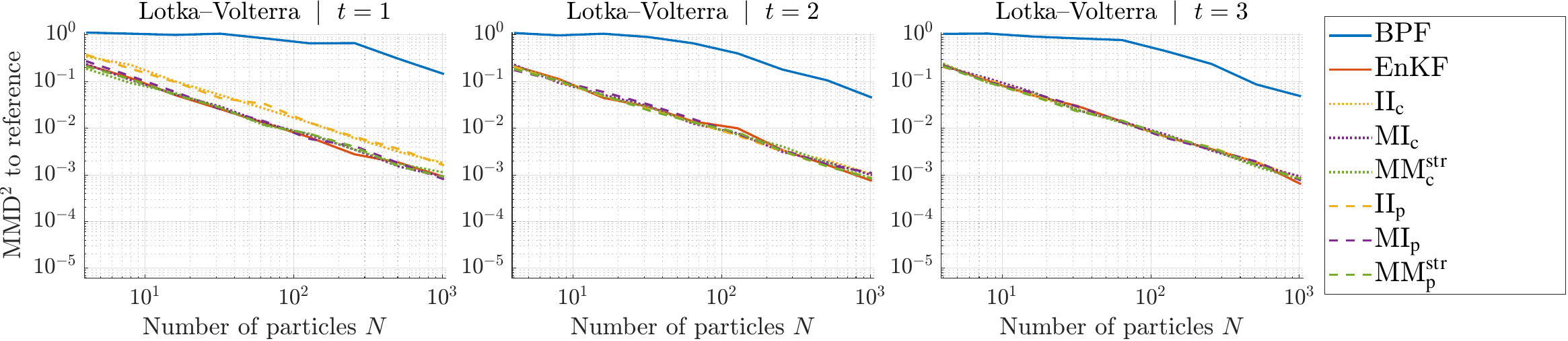}
    \end{subfigure}    
    \vfill
    \vspace{1ex}
    \vfill
    \begin{subfigure}[b]{0.99\textwidth}
        \centering
        \includegraphics[width=\textwidth]{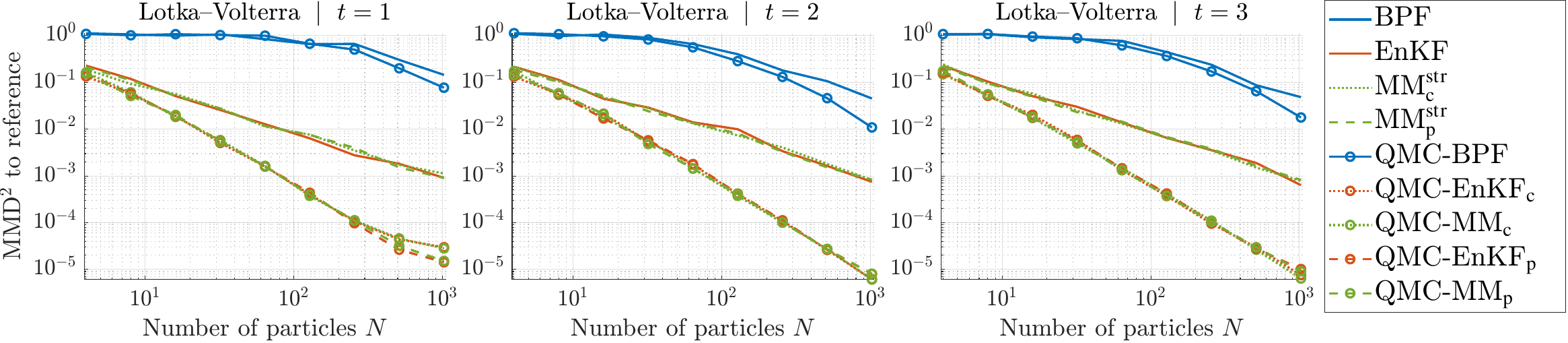}
    \end{subfigure}        
    \caption{Results for the \textbf{Lotka--Volterra} model with \textbf{linear observation} function.
    The Lotka--Volterra flow is close to linear (nearly a rotation) over the ranges explored here.
    Therefore, under linear observations, the setting is almost linear--Gaussian and the EnKF almost consistent.
    Accordingly, all weighted EnKF variants exhibit very similar accuracy. By contrast, BPF displays substantially larger variance. Under \emph{low observation noise} (sharp likelihood), only few prior particles land in the high--posterior--density region, which exacerbates weight degeneracy for BPF. Using TQMC point sets (indicated by markers) improves performance across methods by covering state space more evenly and thus enhancing the convergence rate.
    }
    \label{fig:linear_LotkaVolterra_run_1}    
\end{figure}

\begin{figure}[t]
    \centering         	
    \begin{subfigure}[b]{0.99\textwidth}
        \centering
        \includegraphics[width=\textwidth]{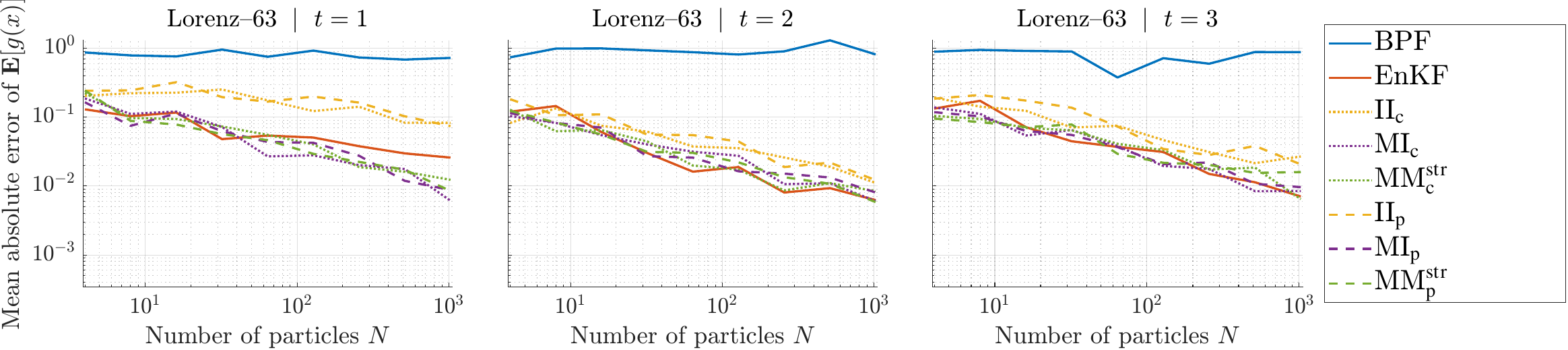}
    \end{subfigure}    
    \vfill
    \vspace{1ex}
    \vfill
    \begin{subfigure}[b]{0.99\textwidth}
        \centering
        \includegraphics[width=\textwidth]{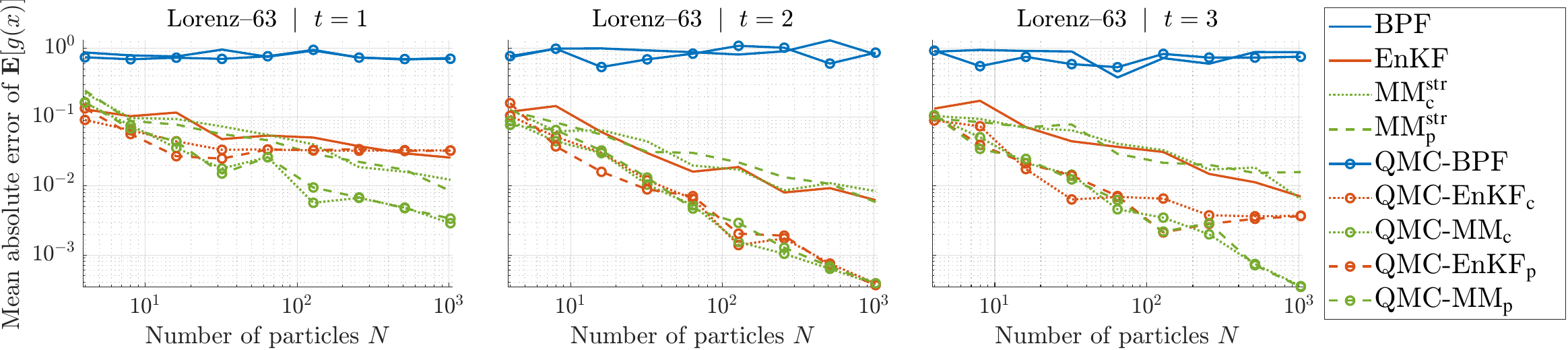}
    \end{subfigure}
    \vfill    
    \vspace{1ex}               
    \hrule
    \vspace{0.2ex}               
    \hrule    
    \vspace{1ex}     
    \begin{subfigure}[b]{0.99\textwidth}
        \centering
        \includegraphics[width=\textwidth]{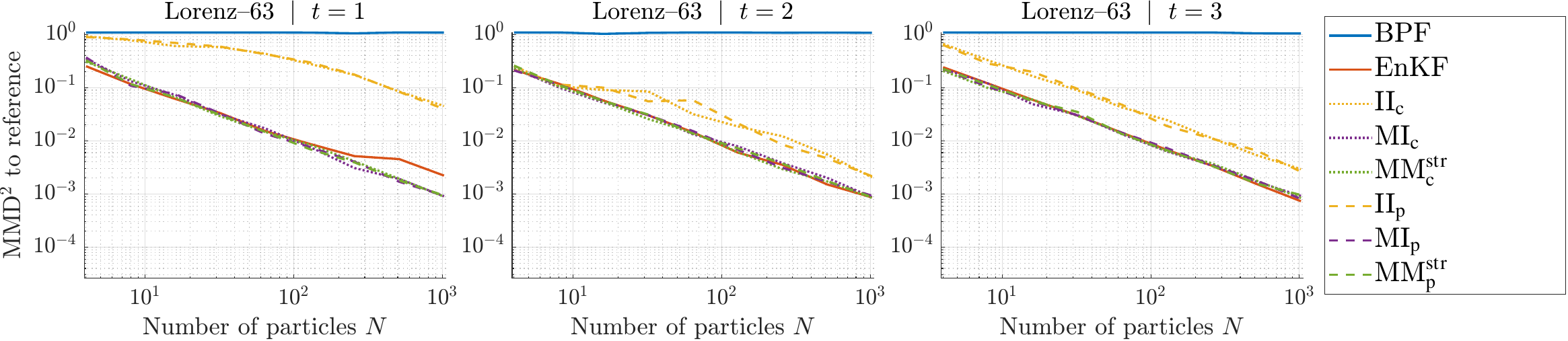}
    \end{subfigure}    
    \vfill
    \vspace{1ex}
    \vfill
    \begin{subfigure}[b]{0.99\textwidth}
        \centering
        \includegraphics[width=\textwidth]{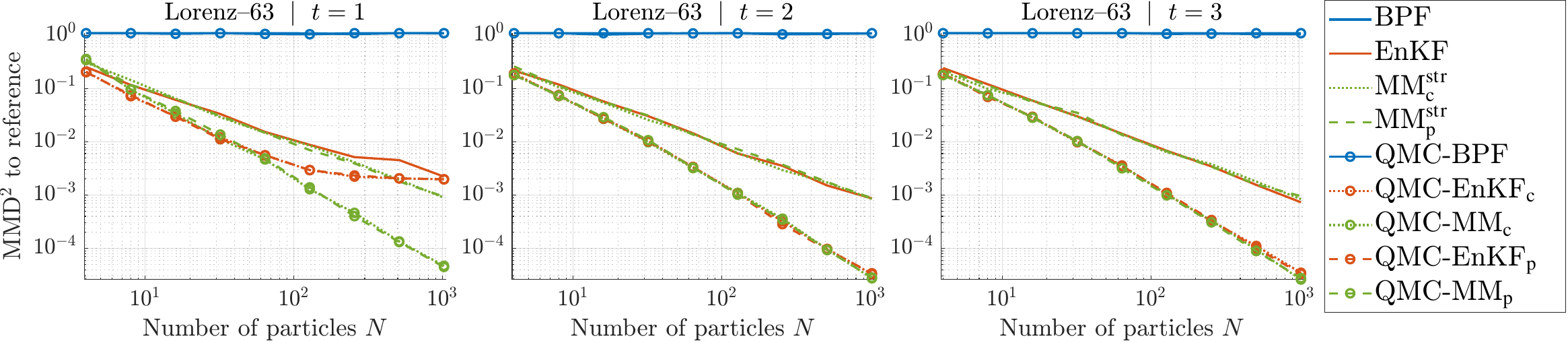}
    \end{subfigure}        
    \caption{
Results for the \textbf{Lorenz--63} model with \textbf{linear observation} function.
    At $t=1$ the prior and posterior are strongly non-Gaussian, and the EnKF analysis--target mismatch introduces a bias that persists even as $N\to\infty$, yielding an early \emph{error floor}.
Among the weighted EnKF schemes, $\II_{\smallgraybox}$ schemes perform worst throughout, consistent with the theoretically larger variance when individual proposals and individual targets place mass on different regions of state space.
Averaging over components ($\MI_{\smallgraybox}$, $\MM_{\smallgraybox}^{\strat}$) mitigates this effect and lowers variance.
As assimilation proceeds ($t=2,3$), repeated conditioning pulls the forecast closer to Gaussianity, and the gap between EnKF and its weighted variants narrows.
The penalty for using individual proposals against individual targets in $\II_{\smallgraybox}$ also diminishes:
this is consistent with the ensemble becoming more tightly clustered relative to the process--noise scale (set by $Q$), given that the observation noise is small ($R \ll Q$), which reduces the discrepancy between per-component targets and proposals.
\newline
TQMC improves convergence for all six EnKF reweighting schemes but cannot remove the $t=1$ saturation of the unweighted EnKF, which is bias dominated.
For the BPF, severe weight degeneracy under low observation noise keeps the method far from the asymptotic regime over the $N$ considered, so the rate improvement from TQMC would only become visible for much larger $N$ (beyond the range shown).}
    \label{fig:linear_Lorenz63_run_1}    
\end{figure}

\begin{figure}[t]
    \centering         	
    \begin{subfigure}[b]{0.99\textwidth}
        \centering
        \includegraphics[width=\textwidth]{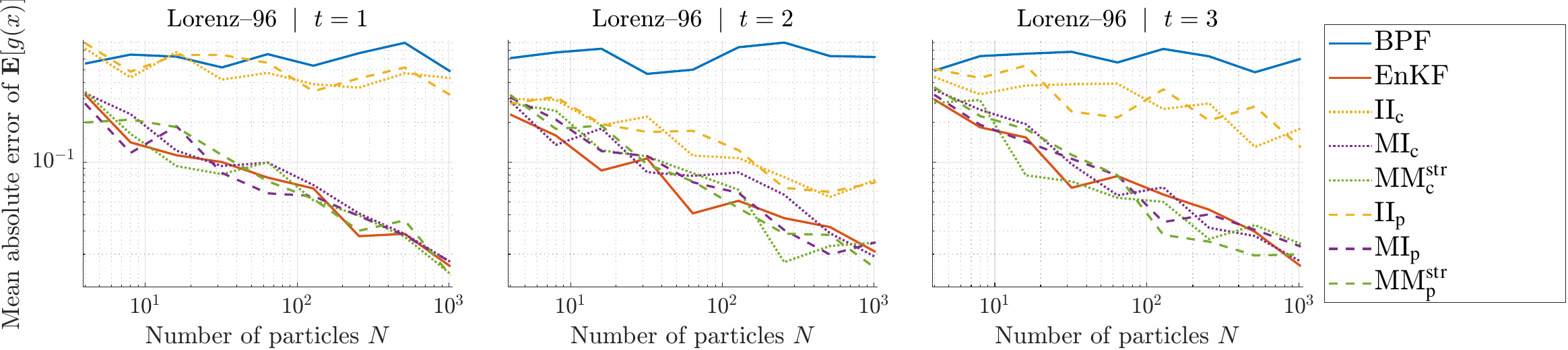}
    \end{subfigure}    
    \vfill
    \vspace{1ex}
    \vfill
    \begin{subfigure}[b]{0.99\textwidth}
        \centering
        \includegraphics[width=\textwidth]{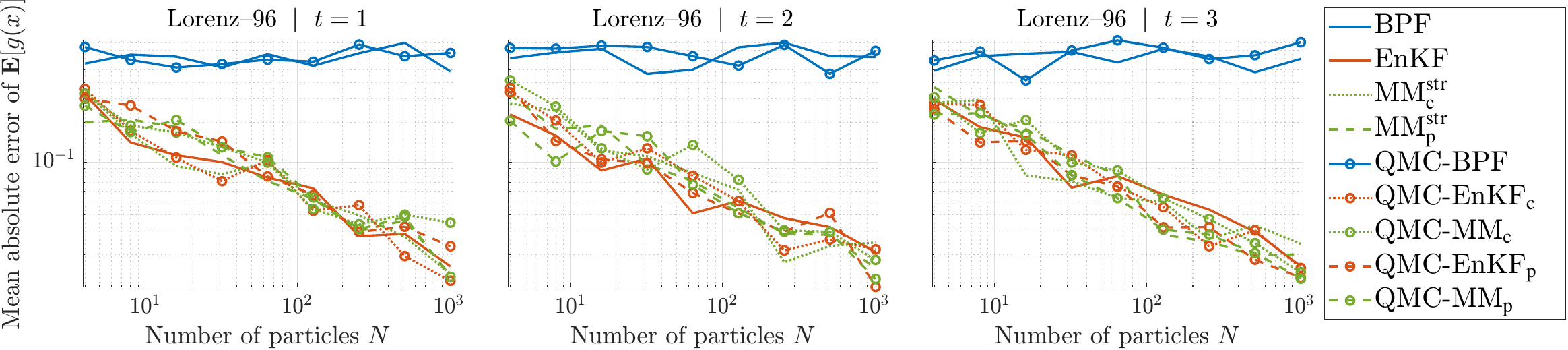}
    \end{subfigure}
    \vfill    
    \vspace{1ex}               
    \hrule
    \vspace{0.2ex}               
    \hrule    
    \vspace{1ex}     
    \begin{subfigure}[b]{0.99\textwidth}
        \centering
        \includegraphics[width=\textwidth]{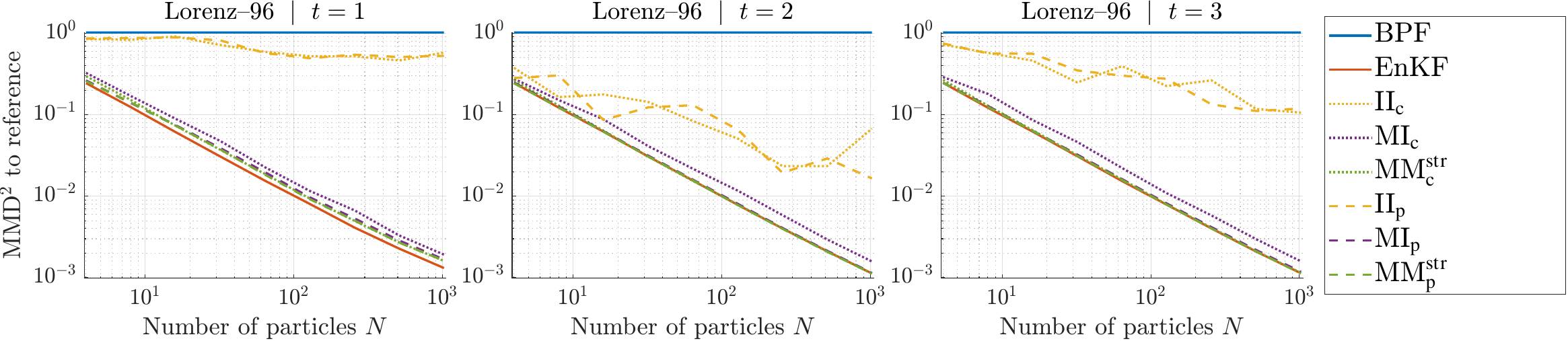}
    \end{subfigure}    
    \vfill
    \vspace{1ex}
    \vfill
    \begin{subfigure}[b]{0.99\textwidth}
        \centering
        \includegraphics[width=\textwidth]{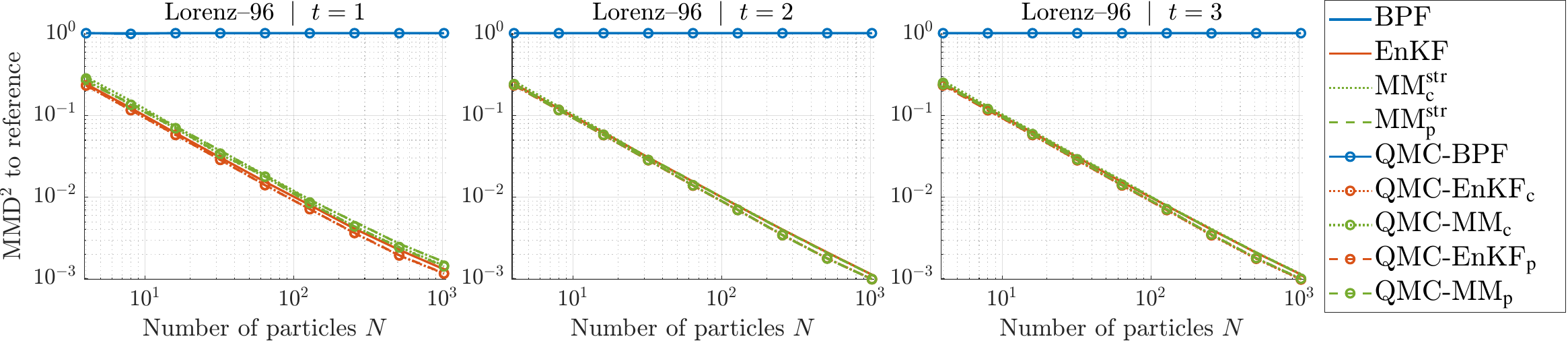}
    \end{subfigure}        
    \caption{Results for the \textbf{Lorenz–96} model with \textbf{linear observation} function.
As in the previous examples, the BPF fails to place a meaningful number of particles in the high-posterior-density region, leading to pronounced weight degeneracy.
Among the EnKF-based methods, the $\II_{\smallgraybox}$ schemes again perform worst.
For this experiment, the remaining weighted EnKF schemes exhibit very similar errors to the vanilla EnKF over the range of ensemble sizes considered, so reweighting neither clearly improves nor degrades performance relative to the standard EnKF.
However, in contrast to the vanilla EnKF, these weighted schemes come with consistency guarantees as $N \to \infty$ by \Cref{thm:convergence}, so the results indicate that such guarantees can be obtained without sacrificing performance at the moderate ensemble sizes tested here.
In this higher-dimensional setting ($d = 40$), using $\TQMC$ point sets does not lead to visible additional gains beyond the underlying schemes.}
    \label{fig:linear_Lorenz96_run_1}    
\end{figure}

\begin{figure}[t]
    \centering         	
    \begin{subfigure}[b]{0.99\textwidth}
        \centering
        \includegraphics[width=\textwidth]{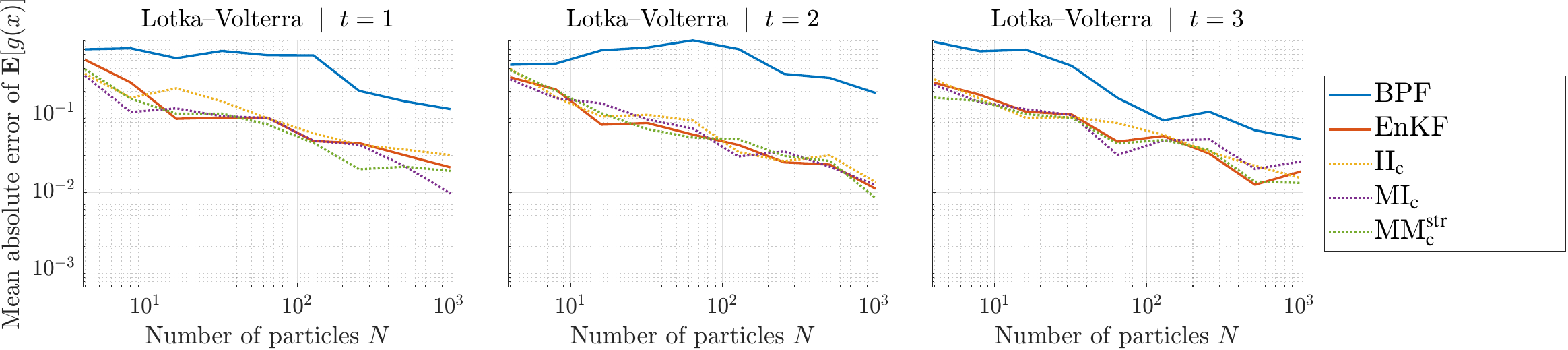}
    \end{subfigure}    
    \vfill
    \vspace{1ex}
    \vfill
    \begin{subfigure}[b]{0.99\textwidth}
        \centering
        \includegraphics[width=\textwidth]{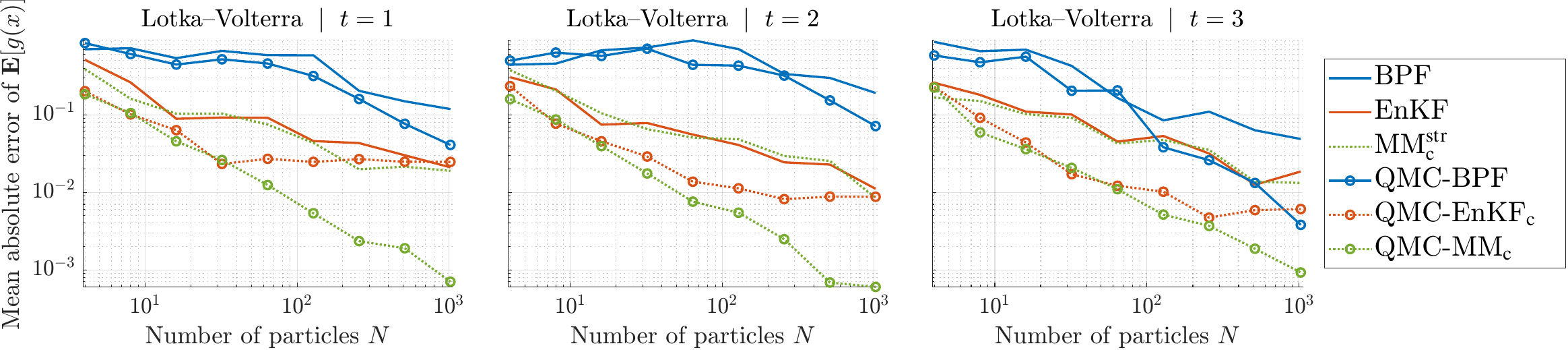}
    \end{subfigure}
    \vfill    
    \vspace{1ex}               
    \hrule
    \vspace{0.2ex}               
    \hrule    
    \vspace{1ex}     
    \begin{subfigure}[b]{0.99\textwidth}
        \centering
        \includegraphics[width=\textwidth]{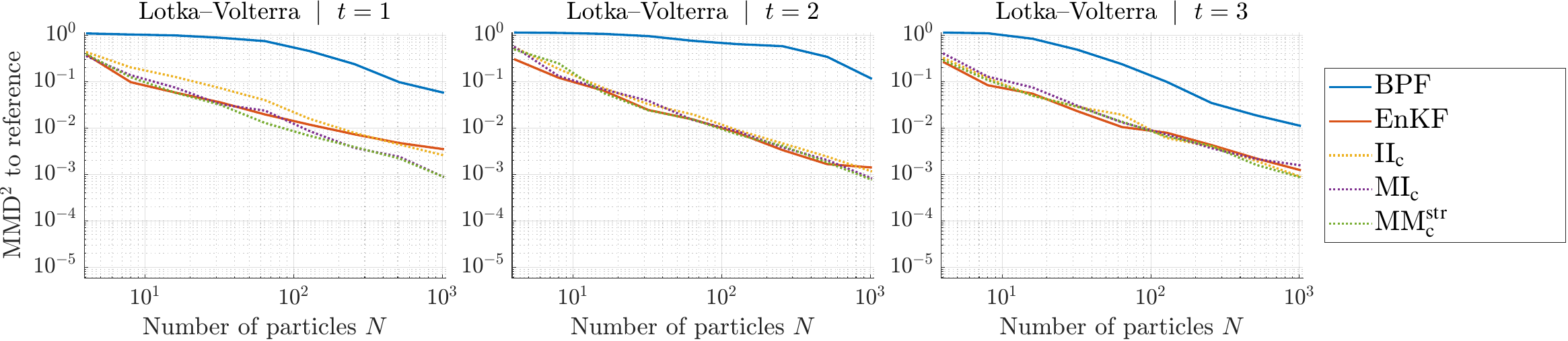}
    \end{subfigure}    
    \vfill
    \vspace{1ex}
    \vfill
    \begin{subfigure}[b]{0.99\textwidth}
        \centering
        \includegraphics[width=\textwidth]{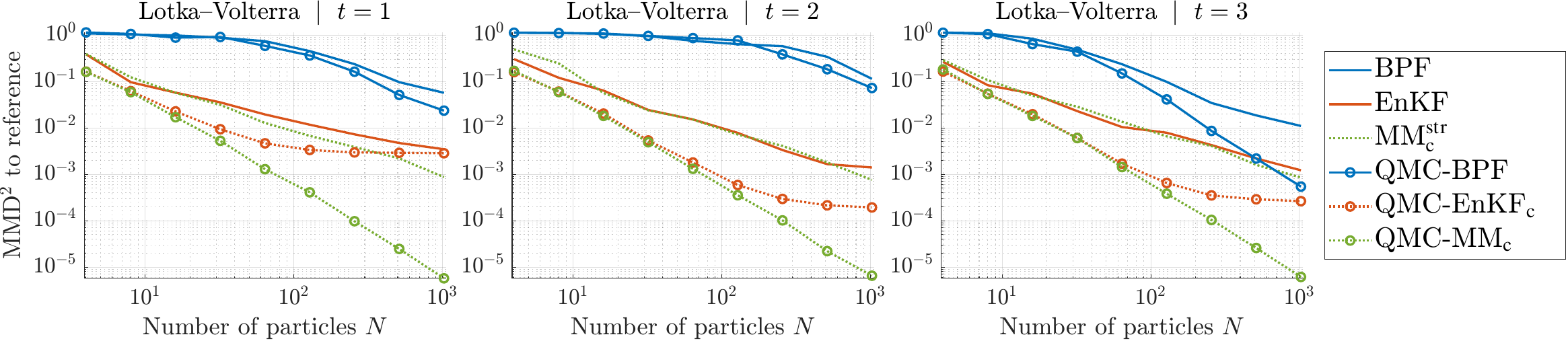}
    \end{subfigure}
    \caption{Results for the \textbf{Lotka--Volterra} model with \textbf{nonlinear observation} function.
        In this nonlinear and strongly non-Gaussian setting, the unweighted EnKF lacks consistency: its errors saturate for all times $t = 1,2,3$ and do not decrease beyond moderate ensemble sizes.
        For the weighted EnKF schemes, the improvement over the considered range of $N$ is harder to see in the comparatively noisy $\mathsf{MAE}$ curves, but becomes more apparent in the $\mathsf{MMD}^{2}$ plots, where the reweighted schemes continue to approach the reference distribution as $N$ increases, while the unweighted EnKF reaches a floor.
        This effect is particularly visible for the $\TQMC$-enhanced schemes, where the inconsistency of the vanilla EnKF prevents any analogous improvement in its $\TQMC$ variant.}
    \label{fig:nonlinear_LotkaVolterra_run_1}    
\end{figure}
\clearpage

\begin{figure}[t]
    \centering         	
    \begin{subfigure}[b]{0.99\textwidth}
        \centering
        \includegraphics[width=\textwidth]{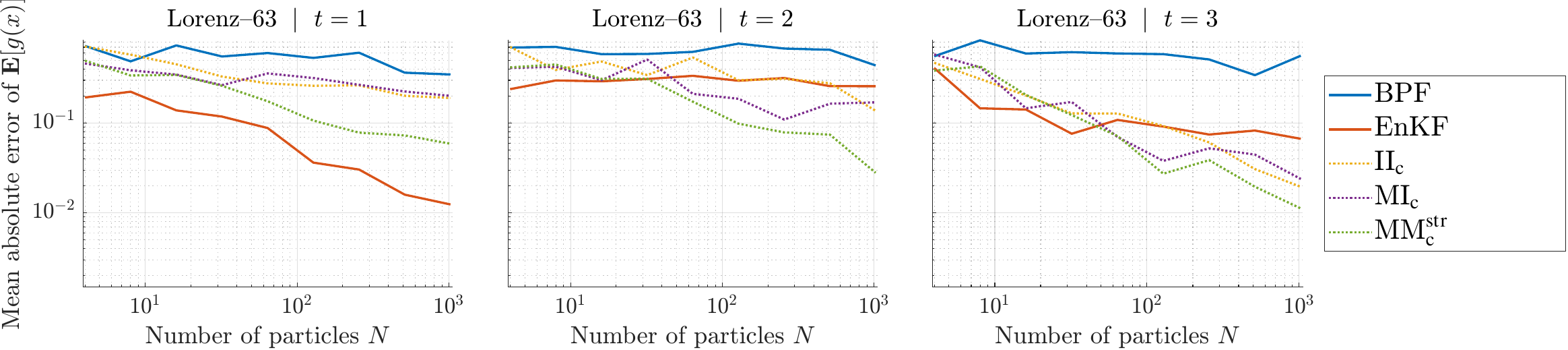}
    \end{subfigure}    
    \vfill
    \vspace{1ex}
    \vfill
    \begin{subfigure}[b]{0.99\textwidth}
        \centering
        \includegraphics[width=\textwidth]{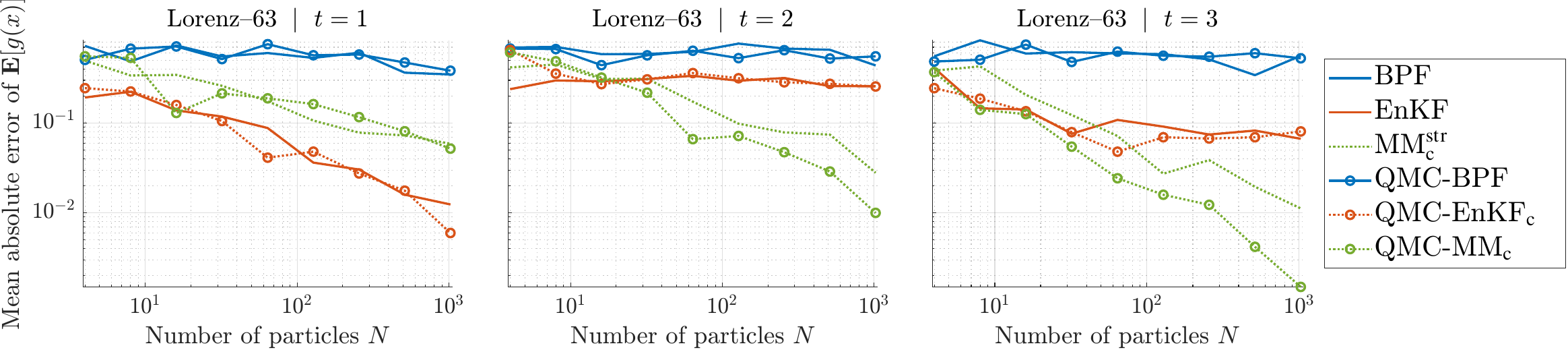}
    \end{subfigure}
    \vfill    
    \vspace{1ex}               
    \hrule
    \vspace{0.2ex}               
    \hrule    
    \vspace{1ex}     
    \begin{subfigure}[b]{0.99\textwidth}
        \centering
        \includegraphics[width=\textwidth]{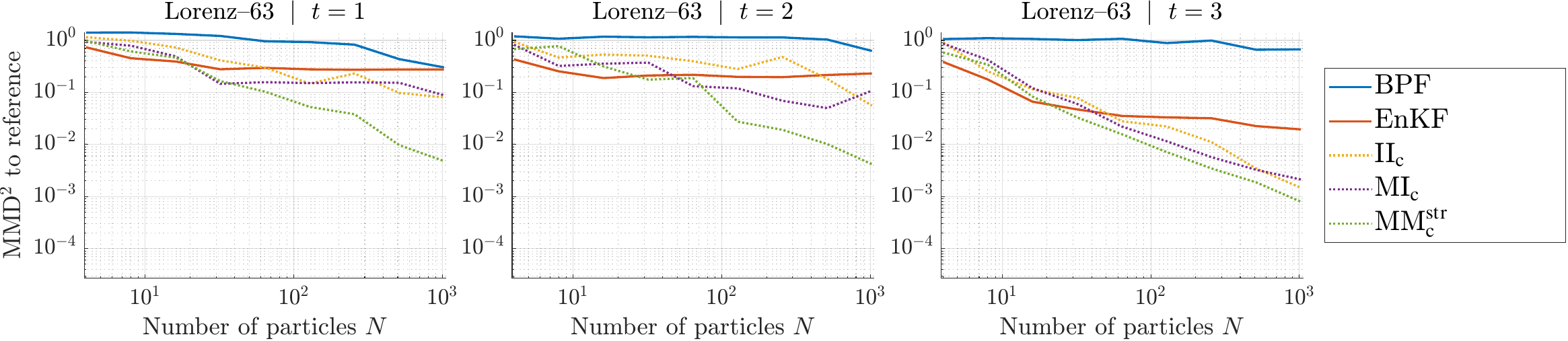}
    \end{subfigure}    
    \vfill
    \vspace{1ex}
    \vfill
    \begin{subfigure}[b]{0.99\textwidth}
        \centering
        \includegraphics[width=\textwidth]{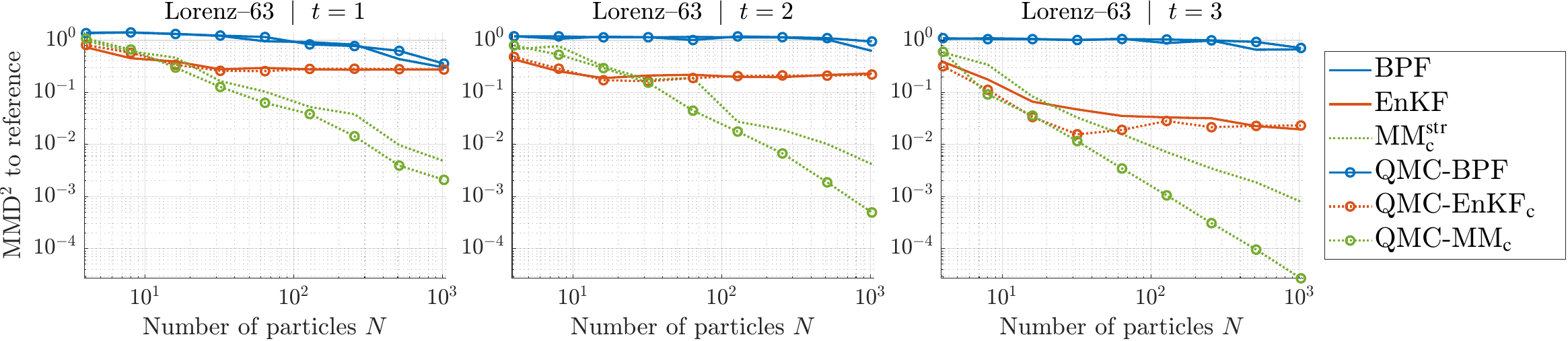}
    \end{subfigure}
    \caption{Results for the \textbf{Lorenz--63} model with \textbf{nonlinear observation} function.
Compared to the linear observation case, the lack of consistency of the unweighted EnKF is even more pronounced: for all times $t=1,2,3$ its error curves quickly hit a floor and show little improvement with increasing $N$, and the use of $\TQMC$ point sets does not alleviate this saturation.
In contrast, the weighted EnKF schemes continue to improve with $N$ and, for sufficiently large ensemble sizes, eventually outperform the vanilla EnKF, with the mixture--mixture scheme $\MM_{\mathrm{c}}^{\strat}$ giving the best overall performance.
The only exception is the $\mathsf{MAE}$ plot at $t=1$, where the vanilla EnKF attains the smallest error for the particular test integrand $g$ considered; however, it simultaneously performs very poorly in terms of $\mathsf{MMD}^{2}$, which measures the worst-case integration error over \emph{all} integrands from the corresponding RKHS.}
    \label{fig:nonlinear_Lorenz63_run_1}    
\end{figure}
\clearpage

\begin{figure}[t]
    \centering         	
    \begin{subfigure}[b]{0.99\textwidth}
        \centering
        \includegraphics[width=\textwidth]{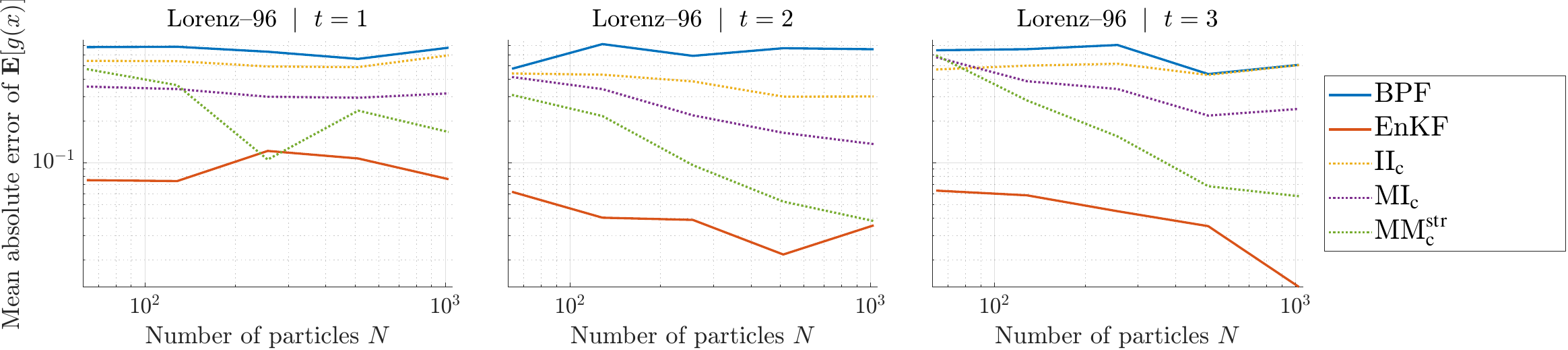}
    \end{subfigure}    
    \vfill
    \vspace{1ex}
    \vfill
    \begin{subfigure}[b]{0.99\textwidth}
        \centering
        \includegraphics[width=\textwidth]{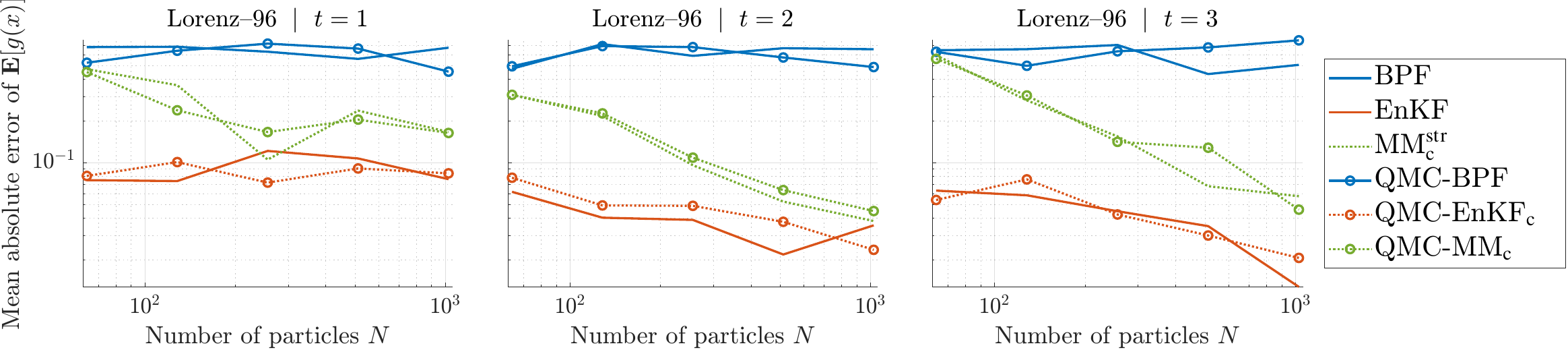}
    \end{subfigure}
    \vfill    
    \vspace{1ex}               
    \hrule
    \vspace{0.2ex}               
    \hrule    
    \vspace{1ex}     
    \begin{subfigure}[b]{0.99\textwidth}
        \centering
        \includegraphics[width=\textwidth]{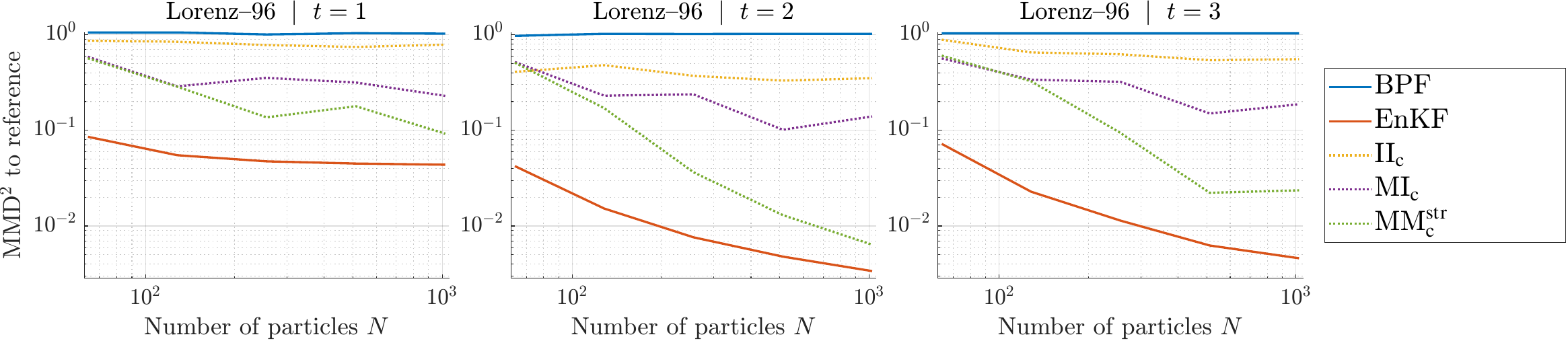}
    \end{subfigure}    
    \vfill
    \vspace{1ex}
    \vfill
    \begin{subfigure}[b]{0.99\textwidth}
        \centering
        \includegraphics[width=\textwidth]{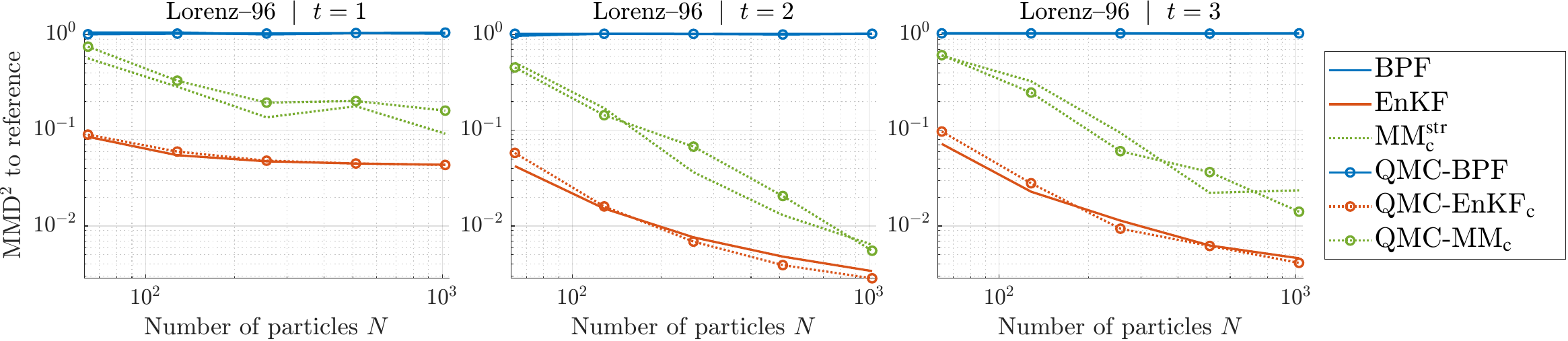}        
    \end{subfigure}
    \caption{Results for the \textbf{Lorenz--96} model with \textbf{nonlinear observation} function.
The mixture-based weighted EnKF schemes again outperform those with individual targets and proposals: across the range of ensemble sizes considered, $\MM_{\mathrm{c}}^{\strat}$ performs best among the weighted schemes, followed by $\MI_{\mathrm{c}}$ and then $\II_{\mathrm{c}}$.
Throughout this experiment, the vanilla EnKF attains the smallest errors for the tested values of $N$, but the $\mathsf{MMD}^{2}$ curves indicate a clear saturation of the EnKF as $N$ increases, whereas the error of $\MM_{\mathrm{c}}^{\strat}$ continues to decrease.
In contrast to the vanilla EnKF, the weighted schemes (and in particular $\MM_{\mathrm{c}}^{\strat}$) come with consistency guarantees as $N \to \infty$ by \Cref{thm:convergence}, so the observed behavior suggests that for sufficiently large ensemble sizes (beyond those shown here) $\MM_{\mathrm{c}}^{\strat}$ would eventually cross and outperform the EnKF in terms of distributional accuracy.
As in the linear observation case and consistent with the high dimension of the state space ($d = 40$), using $\TQMC$ point sets does not yield visible improvements over the corresponding Monte Carlo schemes.}
    \label{fig:nonlinear_Lorenz96_run_1}    
\end{figure}
\clearpage

\section{Conclusion}
\label{sec:Conclusion}

We developed a unified importance-sampling framework for \emph{mixture} targets and \emph{mixture} proposals and used it to derive and analyze mixture-weighted variants of the EnKF.
On the algorithmic side, we interpreted the stochastic EnKF analysis step as sampling from explicit Gaussian proposals and formalized both previous-ensemble and current-ensemble conditionals, enabling principled reweighting beyond the linear observation setting of the classical WEnKF.
These constructions come with important applicability and stability caveats (in particular, linear-only availability of the previous-ensemble conditional, the WEnKF gain/proposal mismatch, and potential singularity issues), which we summarize in \Cref{remark:restrictions_weighted_EnKF_schemes}.
This led to six self-normalized IS-EnKF schemes, for which we proved consistency under a standard finite-variance condition on the importance weights (\Cref{thm:convergence}).

Our theoretical analysis identifies the (mean-one normalized) importance-weight variability, quantified by the squared coefficient of variation $\sigma_t^2$, as the key quantity governing SNIS error bounds.
Within our family of schemes we showed that using mixture targets together with mixture proposals minimizes $\sigma_t$ (\Cref{thm:weight_variance_inequality}), thereby providing a rigorous justification for mixture-based reweighting.
Moreover, for linear observations and bounded drift, we verified the key condition $\sigma_t<\infty$ required for consistency for two representative previous-ensemble schemes (\Cref{thm:IS_weight_bound}), ensuring that the consistency theory applies in this setting.
Our numerical experiments support these findings: mixture-based reweighting reduces weight degeneracy and improves estimation accuracy relative to per-component reweighting, and in regimes where the plain EnKF exhibits an $N$-independent error plateau due to analysis--target mismatch, the weighted schemes exhibit the expected decay with $N$.

To reduce sampling error in both prediction and analysis, we further proposed transported quasi-Monte Carlo (TQMC) for Gaussian mixtures (\Cref{sec:QMC_for_EnKF}), yielding TQMC-enhanced variants of BPF, EnKF, and their mixture-weighted counterparts with substantial gains in low-to-moderate dimensions for consistent methods and competitive performance in higher-dimensional tests.
Naturally, TQMC addresses sampling variance only and cannot resolve the EnKF analysis--target mismatch, so any associated error saturation persists --- providing an additional motivation to employ weighted EnKF variants.

Several directions remain for future work.
First, it would be valuable to extend the finite-variance theory beyond the presently covered cases, in particular to broader classes of current-ensemble schemes, nonlinear observation models $h$ and unbounded drifts $f$.
Second, combining mixture-weighted reweighting with standard EnKF stabilizations such as localization and inflation is an important step towards large-scale applications: since the reweighting step targets the correct filtering distribution asymptotically, it does not introduce the systematic bias that can accompany localization or inflation, and instead offers a principled route to mitigate analysis-target mismatch while retaining the practical benefits of these stabilizations.
Third, for nonlinear observation maps $h$, a key open problem is how to control the range-related singularities highlighted in \Cref{remark:restrictions_weighted_EnKF_schemes} and to design robust proposals and reweighting strategies that avoid weight blow-up.
Finally, it remains to better understand how (transported) QMC should be deployed in data assimilation, in particular in higher dimensions where low-discrepancy structure, transport quality, and the sequential nature of the problem interact in subtle ways.

\appendix

\section{Counterexamples to additional variance inequalities between the estimators $\II,\MI,\IM,\MM$ and $\MM^{\strat}$}
\label{sec:counterexamples_variance_inequalities}

Theorem~\ref{thm:mixture_IS_schemes_comparison} shows that, for any
$g\in L^2(p_{\mix})$,
\begin{equation*}
    \bV[\MM^{\strat}(g)]
    \leq
    \bV[\MM(g)]
    \leq \min\bigl( \bV[\IM(g)] , \bV[\MI(g)] \bigr).
\end{equation*}
In particular, the ``mixture-mixture'' estimators $\MM$ and $\MM^{\strat}$ are
never worse (in variance) than either $\IM$ or $\MI$.
The goal of this appendix is to show that \emph{no further} general variance
inequalities between the five estimators $\II, \MI, \IM, \MM$ and $\MM^{\strat}$ from \Cref{sec:IS_with_mix_targets_and_proposals} hold:
for every pair of estimators that is not ordered by the display above,
both directions of the inequality can occur for suitable choices of
$(p_i,q_i,g)$.

For simplicity, throughout this section, we work on a two--point state space $X=\{0,1\}$ with $N=2$ components, noting that these discrete examples can be embedded without difficulty into the setting of \Cref{thm:mixture_IS_schemes_comparison}.
We represent probability mass functions on $X$ by two--dimensional vectors and
write, for example, $p=(\alpha,\beta)$ to mean $p(0)=\alpha$ and $p(1)=\beta$.

\begin{example}[{$\bV[\II(1)]>\bV[\IM(1)]>\bV[\MI(1)]$}]
\label{ex:II_ge_IM_ge_MI}
We construct the example so that $p_{\mix}=q_{\mix}$ for $g\equiv 1$, which makes
the mixture-mixture estimators $\MM$ and $\MM^{\strat}$ \emph{deterministic},
while the ``individual'' estimators $\II$, $\IM$ and $\MI$ still have strictly positive (but different) variances.
Let $g\equiv 1$ and choose
\[
p_1=\Bigl(\frac{1}{10},\frac{9}{10}\Bigr),\quad
p_2=\Bigl(\frac{2}{5},\frac{3}{5}\Bigr),\quad
q_1=\Bigl(\frac{3}{10},\frac{7}{10}\Bigr),\quad
q_2=\Bigl(\frac{1}{5},\frac{4}{5}\Bigr).
\]
Then
\[
p_{\mix}
=
\frac12(p_1+p_2)
=
\Bigl(\frac{1}{4},\frac{3}{4}\Bigr),
\qquad
q_{\mix}
=
\frac12(q_1+q_2)
=
\Bigl(\frac{1}{4},\frac{3}{4}\Bigr),
\]
so $p_{\mix}=q_{\mix}$ and $p_{\mix}\ll q_i$ holds for $i=1,2$.
All five estimators target $I=\bE_{p_{\mix}}[g]=1$.
Since $p_{\mix}=q_{\mix}$ and $g\equiv 1$, each summand in $\MM(g)$ and
$\MM^{\strat}(g)$ is identically $g$ and therefore
$
\bV[\MM(g)]=\bV[\MM^{\strat}(g)]=0
$.
In contrast, the pairs $(p_i,q_i)$ are not aligned for $i=1,2$, so
the ratios $p_i/q_i$, $p_i/q_{\mix}$ and $p_{\mix}/q_i$ are all
non--constant on $X$, and the estimators $\II(g)$, $\IM(g)$ and
$\MI(g)$ remain random.
A direct calculation of the variances gives
\[
\bV[\II(g)] = \frac{37}{336}
\ \  > \ \ 
\bV[\IM(g)] = \frac{3}{50}
\ \  > \ \ 
\bV[\MI(g)] = \frac{37}{5376}
\ \  > \ \ 
\bV[\MM(g)]
=
\bV[\MM^{\strat}(g)]
=
0.
\]
\end{example}

\begin{example}[{$\bV[\MI(g)]>\bV[\IM(g)]>\bV[\MM(g)]>\bV[\MM^{\strat}(g)]>\bV[\II(g)]$}]
\label{ex:MI_ge_IM_ge_MMstr_ge_II}
We now construct an example in which $\II$ has \emph{zero} variance, while
the other four estimators have strictly positive variance, and at
the same time $\bV[\MI(g)]>\bV[\IM(g)]$.
Let
\[
q_1=\Bigl(\frac{2}{3},\frac{1}{3}\Bigr),
\qquad
q_2=\Bigl(\frac{1}{3},\frac{2}{3}\Bigr),
\qquad
g = (1,2).
\]
We now choose $p_1,p_2$ so that, for each $i = 1,2$, the quantity
$\bigl(p_i/q_i\bigr)\,g$ is constant on $X$, that is
\[
\frac{p_i(0)}{q_i(0)}g(0)
=
\frac{p_i(1)}{q_i(1)}g(1).
\]
These conditions yield
\[
\frac{p_1(0)}{p_1(1)}
=
\frac{q_1(0)\,g(1)}{q_1(1)\,g(0)}
=
\frac{\frac{2}{3}\cdot 2}{\frac{1}{3}\cdot 1}
=
4,
\qquad
\frac{p_2(0)}{p_2(1)}
=
\frac{q_2(0)\,g(1)}{q_2(1)\,g(0)}
=
\frac{\frac{1}{3}\cdot 2}{\frac{2}{3}\cdot 1}
=
1,
\]
and thereby $p_1=(\tfrac45,\tfrac15)$ and $p_2=(\tfrac12,\tfrac12)$.
Then
$q_{\mix}
=
\tfrac12(q_1+q_2)
=
(\frac12,\frac12)$
and
$p_{\mix}
=
\tfrac12(p_1+p_2)
=
(\frac{13}{20},\frac{7}{20})$,
so $p_{\mix}\ll q_i$ holds for $i=1,2$.
By construction, each summand in
\[
\II(g)
=
\frac12\Bigl(\frac{p_1(X_1)}{q_1(X_1)}g(X_1)
+
\frac{p_2(X_2)}{q_2(X_2)}g(X_2)\Bigr)
\]
is constant, implying $\bV[\II(g)]=0$.
Further, a direct computation gives
\[
\bV[\MI(g)] = \frac{369}{3200} 
\ >\ 
\bV[\IM(g)] = \frac{41}{400} 
\ >\ 
\bV[\MM(g)] = \frac{1}{800} 
\ >\ 
\bV[\MM^{\strat}(g)] = \frac{1}{900} 
\ >\ 
\bV[\II(g)]=0.
\]
\end{example}

\medskip

Hence, for every pair $(A,B) \in \{\II,\MI,\IM,\MM,\MM^{\strat}\}$ of distinct estimators that is not already ordered by
Theorem~\ref{thm:mixture_IS_schemes_comparison}, we see that both
inequalities $\bV[A(g)]<\bV[B(g)]$ and $\bV[A(g)]>\bV[B(g)]$ can arise for
suitable choices of $(p_i,q_i,g)$.
No additional general variance inequalities between
$\II,\MI,\IM,\MM$ and $\MM^{\strat}$ therefore hold beyond those
stated in Theorem~\ref{thm:mixture_IS_schemes_comparison}.

\section{Auxiliary Results for the Proof of \Cref{thm:convergence}}
\label{section:auxiliary_results_thm_convergence}

We now provide details for the proof of \Cref{thm:convergence}. In particular, we derive the bounds required to control the three contributing error terms: the approximation error for the filtering distribution, the self-normalized importance sampling (SNIS) error, and the resampling error.

\begin{lemma}[approximation error for the filtering distribution]
\label{lem:mixture_error}
Under the assumptions of \Cref{thm:convergence} and using the notation therein,
\begin{equation}
\label{equ:distance_posterior_posteriormixture}
d\bigg(p_t^{\mathsf{post}}, \frac{\postm}{\Zmix} \bigg)
\le \frac{2 d_{t-1}}{\Zpost},
\qquad
t = 1,\dots,T.
\end{equation}
\end{lemma}

\begin{proof}
Let $g \in \mathcal{C}_b(\mathbb{R}^d)$ with $\|g \|_\infty \le 1$.
Since $0 \leq \ell_{t} \leq 1$, we also have $\ell_t, g\ell_t \in \mathcal{C}_b(\mathbb{R}^d)$ with $\| g\ell_t \|_\infty \le \|\ell_t\|_\infty \le 1$, and thereby
\begin{equation}
\label{equ:lemma_mixture_error_bounding_E1_E2}
\begin{alignedat}{2}
E_1(g) 
&\coloneqq
\mathbb{E}\bigg[\bigg(\int g \, \ell_t \, \pi_t^{\mathsf{prior}} - \int g \, \ell_t \, \pim\bigg)^2\bigg]
&\quad&\leq d \big(\pi_t^{\mathsf{prior}}, \pim \big)^{2}
\\[2ex]
E_2(g) 
&\coloneqq
\mathbb{E}\bigg[\bigg(
\int g \, \frac{\postm}{\Zmix}
\bigg)^2 (\Zpost - \Zmix)^2\bigg]
&\quad&\leq \mathbb{E}[(\Zpost - \Zmix)^2]
\\
&=
\mathbb{E}\bigg[\bigg( \int \ell_t \pi_t^{\mathsf{prior}} - \int \ell_t\pim\bigg)^2\bigg]
&\quad&\leq d(\pi_t^{\mathsf{prior}}, \pim)^{2}.
\end{alignedat}
\end{equation}
By Bayes' theorem,
\begin{align*}
p_t^{\mathsf{post}} - \frac{\postm}{\Zmix}
&= \frac{\ell_t \pi_t^{\mathsf{prior}}}{\Zpost} - \frac{\ell_t \pim}{\Zmix}\\
&= \frac{\ell_t \pi_t^{\mathsf{prior}}}{\Zpost}- \frac{\ell_t \pim}{\Zpost}+ \frac{\ell_t \pim}{\Zpost} - \frac{\ell_t \pim}{\Zmix}\\
&= \frac{\ell_t\pi_t^{\mathsf{prior}} - \ell_t\pim}{\Zpost} + \frac{\ell_t \pim}{\Zmix} \frac{\Zmix - \Zpost}{\Zpost},
\end{align*}
which, together with \eqref{equ:lemma_mixture_error_bounding_E1_E2} and the triangle inequality, implies
\begin{align*}
d\bigg(p_t^{\mathsf{post}}, \frac{\postm}{\Zmix} \bigg)
&=
\sup_{\|g\|_\infty \le 1} \mathbb{E}\Bigg[\Bigg(\int g\bigg(p_t^{\mathsf{post}} - \frac{\postm}{\Zmix} \bigg)\Bigg)^{\!\!\! 2}\, \Bigg]^{1/2}
\\
&\leq
\sup_{\|g\|_\infty \le 1} \frac{ \sqrt{E_1(g)} + \sqrt{E_2(g)}}{\Zpost}
\\
&\leq
\frac{2}{\Zpost} \, d(\pi_{t}^{\mathsf{prior}}, \pim)
\\
&\leq
\frac{2}{\Zpost} \, d\bigg(p_{t-1}^{\mathsf{post}}, 
\frac{1}{N} \sum_{i=1}^N \delta_{x_{t-1}^{(i)}}\bigg),
\end{align*}
where the last inequality follows from \citep[Lemma~11.2]{SanzAlonso2023IPandDA}. 
\end{proof}

\begin{lemma}
    \label{lem:conditional_expectations}
    Let $g \in \mathcal{C}_b(\mathbb{R}^d)$, and let $\mathcal{F}_t$ be the $\sigma$-algebra defined in \eqref{eq:history_sigma_algebra}.
    Then, under the assumptions of \Cref{thm:convergence} and using the notation therein, for every $t =1,\dots, T$,
    \[
    \mathbb{E}\Biggl[\frac{1}{N} \sum_{i=1}^N g(\tilde{x}_t^{(i)}) \frac{v_{t}^{(i)}}{\Zmix}\, \Bigg| \, \mathcal{F}_t\Biggr]
    =
    \int \frac{g\,\postm}{\Zmix} \, ,
    \qquad
    \mathbb{E}\Biggl[\frac{v_{t}^{(1)}}{\Zmix} \Biggr]=1.
    \]
\end{lemma}

\begin{proof}
    Since $\frac{1}{N} \sum_{i=1}^N \frac{\chposti}{\chqi} \, \qi = \postm$ for each sampling scheme, for every $g \in \mathcal{C}_b(\mathbb{R}^d)$,
	{\small\begin{align*}
			&\mathbb{E}\Biggl[\frac{1}{N} \sum_{i=1}^N g(\tilde{x}_t^{(i)}) \frac{v_{t}^{(i)}}{\Zmix}\, \Bigg| \, \mathcal{F}_t\Biggr]
			=
			\mathbb{E}\Biggl[ \frac{1}{N} \sum_{i=1}^N \frac{g\,\chposti}{\Zmix\, \chqi} (\tilde{x}_t^{(i)}) \, \Bigg| \, \mathcal{F}_t\Biggr]
			=
			\frac{1}{N} \sum_{i=1}^N \int \frac{g\, \chposti}{\Zmix\, \chqi} \, \qi
			=
			\int \frac{g\,\postm}{\Zmix},
			\\[1ex]
			&\mathbb{E}\Biggl[\frac{v_{t}^{(1)}}{\Zmix} \Biggr]
			=
			\frac{1}{N} \sum_{i=1}^N \mathbb{E}\Biggl[\frac{v_{t}^{(i)}}{\Zmix} \Biggr]
			=
			\mathbb{E}\Biggl[\mathbb{E}\Biggl[\frac{1}{N} \sum_{i=1}^N \frac{v_{t}^{(i)}}{\Zmix} \Bigg| \mathcal{F}_t \Biggr]\Biggr]
			\stackrel{(\ast)}{=}
			\mathbb{E}\Biggl[\int \frac{\postm}{\Zmix} \Biggr]
			=
			1,
		\end{align*}}
	where $(\ast)$ is an application of the previous result for $g \equiv 1$.
	In the second line we used that, by the symmetry of
	\[
	\qi
	=
	\genq \bigg( x_{t-1}^{(i)} , \hat{x}_{t}^{(i)} , \frac{1}{N} \sum_{j=1}^{N} \delta_{x_{t-1}^{(j)}} , \frac{1}{N} \sum_{j=1}^{N} \delta_{\hat{x}_{t}^{(j)}} , y_{t}\bigg)
	\]
	with respect to the particles  $(x_{t-1}^{(j)})_{j=1}^{N}$ in the third and $(\hat{x}_{t}^{(j)})_{j=1}^{N}$ in fourth argument, cf.\ \eqref{equ:generic_proposal}, each family of corresponding random variables is exchangeable across the indices $i=1,\dots,N$, in particular, the weights $v_{t}^{(i)}$ have the same distribution.
\end{proof}

\begin{lemma}[SNIS error]
	\label{lem:IS_error}
	Under the assumptions of \Cref{thm:convergence} and using the notation therein,
	\[
	d\left( \frac{\postm}{\Zmix} \, ,\,  \sum_{i=1}^{N} w_{t}^{(i)}\delta_{\tilde{x}_{t}^{(i)}}\right) 
	\le \frac{2\sigma_t +1}{\sqrt{N}},
	\qquad
	t=1,\dots,T.
	\]
\end{lemma}

\begin{proof}
	Let $g \in \mathcal{C}_b(\mathbb{R}^d)$ with $\|g \|_\infty \le 1$ and let the $\sigma$-algebra $\mathcal{F}_t$ be defined by \eqref{eq:history_sigma_algebra}.
	By \Cref{lem:conditional_expectations}, using the laws of total expectation and total variance as well as $\norm{g}_{\infty} \leq 1$,
	\begin{align*}
		E_1(g) 
		&\coloneqq
		\mathbb{E}\bigg[\bigg(\int g\, \frac{\postm}{\Zmix} - \frac1N \sum_{i=1}^{N} g(\tilde{x}_t^{(i)}) \frac{v_{t}^{(i)}}{\Zmix} \bigg)^{\! \! 2} \, \bigg]
		\\
		&=
		\mathbb{E}\bigg[\mathbb{V}\bigg[\frac1N \sum_{i=1}^{N} g(\tilde{x}_t^{(i)}) \frac{v_{t}^{(i)}}{\Zmix} \, \bigg| \, \mathcal{F}_t\bigg]\bigg]\\
		&=
		\frac{1}{N} \, \mathbb{E}\bigg[\mathbb{V}\bigg[g(\tilde{x}_t^{(1)}) \frac{v_{t}^{(1)}}{\Zmix} \, \bigg| \,  \mathcal{F}_t\bigg]\bigg]\\
		&\leq
		\frac{1}{N} \, \mathbb{E}\bigg[\bigg(g(\tilde{x}_t^{(1)}) \frac{v_{t}^{(1)}}{\Zmix} \bigg)^{\! \! 2} \, \bigg]\\
		&\leq
		\frac{1}{N} \, \mathbb{E}\bigg[\bigg(\frac{v_{t}^{(1)}}{\Zmix} \bigg)^{\! \! 2} \, \bigg]
		\\
		&=
		\frac{1}{N}\biggl( \mathbb{V}\bigg[\frac{v_{t}^{(1)}}{\Zmix} \bigg]
		+ \mathbb{E}\bigg[\frac{v_{t}^{(1)}}{\Zmix} \bigg]^{\! 2}\biggr)
		\\
		&=
		\frac{\sigma_t^2 + 1}{N}.
	\end{align*}
	Since $w_{t}^{(i)} = v_{t}^{(i)} \big/ \sum_{j} v_{t}^{(j)}$ and $\norm{g}_{\infty} \leq 1$, \Cref{lem:conditional_expectations} implies
	\begin{align*}    
		E_2(g) 
		&\coloneqq  \mathbb{E}\bigg[\bigg(\frac1N \sum_{i=1}^{N} g(\tilde{x}_t^{(i)}) \frac{v_{t}^{(i)}}{\Zmix} - \sum_{i=1}^{N} g(\tilde{x}_t^{(i)})w_{t}^{(i)} \bigg)^{\! \! 2}\, \bigg]
		\\
		&=
		\mathbb{E}\bigg[\bigg(\sum_{i=1}^{N} g(\tilde{x}_t^{(i)})w_{t}^{(i)} \bigg)^{\! \! 2}
		\bigg(\frac1N\sum_{j=1}^{N} \frac{v_{t}^{(j)}}{\Zmix} - 1\bigg)^{\! \! 2}\, \bigg]\\
		&\leq
		\mathbb{E}\bigg[\bigg(\frac1N\sum_{j=1}^{N} \frac{v_{t}^{(j)}}{\Zmix} - 1\bigg)^{\! \! 2}\, \bigg]\\
		&=
		\mathbb{V}\bigg[\frac1N\sum_{j=1}^{N} \frac{v_{t}^{(j)}}{\Zmix}\bigg]\\
		&=
		\frac{\sigma_t^2}{N}.
	\end{align*} 
	Using the triangle inequality, it follows that
	\begin{align*}
		d\left( \frac{\postm}{\Zmix} \, , \, \sum_{i=1}^{N} w_{t}^{(i)}\delta_{\tilde{x}_{t}^{(i)}}\right) 
		&=
		\sup_{\|g\|_\infty \le 1} \mathbb{E}\bigg[\bigg(\int g \, \frac{\postm}{\Zmix} - \sum_{i=1}^{N} g(\tilde{x}_t^{(i)})w_{t}^{(i)} \bigg)^{\! \! 2}\, \bigg]^{1/2}\\
		&\leq
		\sup_{\|g\|_\infty \le 1} \big( \sqrt{E_1(g)} + \sqrt{E_2(g)} \big)
		\\
		&\leq
		\frac{\sigma_{t} + 1}{\sqrt{N}} + \frac{\sigma_{t}}{\sqrt{N}}.
	\end{align*}
\end{proof}

\begin{lemma}[{resampling error; adaptation of \citealt[Lemma~11.1]{chopin2020SMC}}]
    \label{lem:resampling_error}
    Under the assumptions of \Cref{thm:convergence} and using the notation therein,
    \[
    d\left(\sum_{i=1}^{N} w_{t}^{(i)} \delta_{\tilde{x}_{t}^{(i)}},\frac{1}{N} \sum_{i=1}^{N} \delta_{x_{t}^{(i)}} \right)
    \le \frac{1}{\sqrt{N}},
    \qquad
    t=1,\dots,T.
    \]
\end{lemma}

\begin{proof}
    Due to \eqref{equ:systematic_resampling},
    \(
    \mathbb{E}\left[ \frac{1}{N} \sum_{i=1}^{N} g(x_{t}^{(i)}) \mid \{\tilde{x}_t^{(i)}\}\right]
    = \sum_{i=1}^{N} w_{t}^{(i)} g(\tilde{x}_{t}^{(i)}).
    \) 
    Therefore,
    \[
    d\bigg(\sum_{i=1}^{N} w_{t}^{(i)} \delta_{\tilde{x}_{t}^{(i)}},\frac{1}{N} \sum_{i=1}^{N} \delta_{x_{t}^{(i)}} \bigg)^{\! \! 2}
    =
    \sup_{\|g\|_\infty \le 1} \mathbb{E}\bigg[\mathbb{V}\bigg[ \frac{1}{N} \sum_{i=1}^{N} g(x_{t}^{(i)}) \mid \{\tilde{x}_t^{(i)}\}\bigg]\bigg]\le \frac{1}{N},
    \]
    since the supremum is taken over functions $g$ which satisfy $g(x_t^{(i)})^{2} \le 1$.
\end{proof}

\section{Auxiliary Results for the Proof of \Cref{thm:IS_weight_bound}}
\label{section:auxiliary_results_thm_IS_weight_bound}

We will now show the details for each step in the proof of \Cref{thm:IS_weight_bound}, using the assumptions and notation therein.
For this purpose, we begin by introducing certain Kalman and EnKF matrices and establishing their relations, which we will use throughout this section:\footnote{Note that the definition of $\Sigma_{t}$ involves $Q$ rather than $Q_{t}$. This is deliberate and not a typographical error.}
\begin{align*}
Q&\in \SPD{d}&	
Q_{t}
&\coloneqq
\hat{C}_{t}^{x,\mathsf{p}}
=
\Cov^{\emp} \bigl[(f(x_{t-1}^{(i)}))_{i=1}^N\bigr] + Q \in \SPD{d},
\\
S
&\coloneqq
H Q H^\top + R \in \SPD{m},&
S_{t}
&\coloneqq
H Q_{t} H^{\top} + R \in \SPD{m},
\\
K
&\coloneqq
Q H^\top S^{-1} \in \bR^{d \times m},&
K_{t}
&\coloneqq
K_{t}^{\mathsf{p}}
=
Q_{t} H^\top S_{t}^{-1} \in \bR^{d \times m},
\\
\Sigma
&\coloneqq
(I_{d} - K H) Q (I_{d} - K H)^\top + K R K^\top,&
\Sigma_{t}
&\coloneqq
\Qip
=
(I_{d} - K_t H) Q (I_{d} - K_t H)^{\!\top} + K_{t} R K_{t}^{\top},
\end{align*}
cf.\ \eqref{equ:previous_based_gain_linear} for the definition of $\hat{C}_{t}^{x,\mathsf{p}}$ and $K_{t}^{\mathsf{p}}$ and \eqref{equ:qi_previous_particles} for the definition of $\Qip$.
We also denote		
\[
\Delta_{Q}
\coloneqq
Q_{t} - Q,
\qquad
\Delta_{K}
\coloneqq
K_{t} - K,
\qquad
\Delta_{\Sigma}
\coloneqq
\Sigma_{t} - \Sigma.
\]

\begin{lemma}[EnKF matrix identities]
\label{lemma:EnKF_matrix_identities}
The above-defined matrices satisfy:
\begin{enumerate}[label = (\roman*)]
	\item
	\label{item:posterior_covariance_Kalman}			
	$\Sigma = Q - KHQ$, in line with its definition in \Cref{thm:IS_weight_bound}; 
	\item
	\label{item:PHKR}
	$\Sigma H^\top = KR$;
	\item
	\label{item:determinant_P} 
	$\det \Sigma \det S = \det Q \det R$, in particular, $\Sigma \in \SPD{d}$ is strictly positive definite;
	\item
	\label{item:Kalman_difference_full_rank}
	$\Delta_{K} =  (I_d - K_{t} H) \Delta_{Q} H^\top S^{-1}$;
	\item
	\label{item:covariance_difference}
	$\Delta_{\Sigma} = \Delta_{K} S \Delta_{K}^\top$ and $\ran(\Delta_{\Sigma}) = \ran(\Delta_{K})$,
	in particular, $\Delta_{\Sigma} \in \SPSD{d}$ and $\Sigma_{t} \in \SPD{d}$.
\end{enumerate}
\end{lemma}

\begin{proof}
\ref{item:posterior_covariance_Kalman} follows from resolving the parentheses in the definition of $\Sigma$:
\[
\Sigma 
= Q - KHQ - QH^\top K^\top + K(HQH^\top + R)K^\top
= Q - KHQ - QH^\top K^\top + QH^\top S^{-1} S K^\top.
\]
By \ref{item:posterior_covariance_Kalman},
$
\Sigma H^\top
= QH^\top - KHQH^\top 
= K(S - HQH^\top) 
= KR
$,
proving \ref{item:PHKR}.		
\ref{item:determinant_P} follows from \ref{item:posterior_covariance_Kalman} and the Weinstein--Aronszajn identity,
\[
\det \Sigma 
=
\det(I_{d} - KH) \det Q
=
\det(I_{m} - HK) \det Q
=
\det(S - H Q H^{\top}) \det S^{-1} \det Q,
\]
while \ref{item:Kalman_difference_full_rank} can be shown as follows:
\[
\Delta_{K} S
=
K_{t} (S - S_{t}) + K_{t} S_{t} - Q H^{\top}
=
- K_{t} H \Delta_{Q} H^{\top} + \Delta_{Q} H^{\top}
=
(I_{d} - K_{t} H) \Delta_{Q} H^{\top}.
\]
Finally, using \ref{item:posterior_covariance_Kalman} as well as \ref{item:PHKR}, we obtain \ref{item:covariance_difference}:
\begin{align*}
	\Sigma_{t} &= (I-K_{t} H)Q(I-K_{t} H)^\top + K_{t} R K_{t}^\top\\
	&= (I-KH - \Delta_{K} H)Q(I-KH-\Delta_{K} H)^\top + (K+\Delta_{K})R (K+\Delta_{K})^\top
    \\	
	&= \Sigma - \Delta_{K} H \Sigma - \Sigma H^\top \Delta_{K}^\top + \Delta_{K} H Q H^\top \Delta_{K}^\top + \Delta_{K} RK^\top + KR\Delta_{K}^\top + \Delta_{K} R\Delta_{K}^\top\\
	&= \Sigma + \Delta_{K} (HQH^\top + R) \Delta_{K}^\top + \Delta_{K} ( RK^\top - H \Sigma) + (KR -\Sigma H^\top )\Delta_{K}^\top\\
	&= \Sigma + \Delta_{K} S \Delta_{K}^\top.
\end{align*}	
Since $S \in \SPD{m}$ and $\ran(A A^\top) = \ran(A)$ for every matrix $A$, 
$
\ran \Delta_{\Sigma} = \ran(\Delta_{K} S^{1/2}) = \ran \Delta_{K}
$,
while $\Delta_{\Sigma} \in \SPSD{d}$ and $\Sigma_{t} \in \SPD{d}$ follow from $\Sigma_{t} \in \SPD{d}$ and the above representation by considering the Loewner order.
\end{proof}

\begin{lemma}[{filtering component formula, cf.\ \citealt[Section~III]{Alspach1972GaussianSum}}]
\label{lemma:thm_IS_weight_bound_posti_formula}
Under the assumptions of \Cref{thm:IS_weight_bound} and using the notation therein, $\Sigma \in \SPD{d}$ and
\begin{equation*}
	p_t^{(i)}
	=
	\biggl(\frac{\det R}{\det S}\biggr)^{\! \! 1/2} 
	\exp \big( - \tfrac{1}{2} |y_t - H f(x_{t-1}^{(i)})|_S^2 \big)
	\cdot \mathcal{N}(\hat{m}_{t}^{(i)}, \Sigma).
\end{equation*}
In other words, \Cref{thm:IS_weight_bound}\ref{item:thm_IS_weight_bound_posti_formula} holds.
\end{lemma}

\begin{proof}
$\Sigma \in \SPD{d}$ was proven in \Cref{lemma:EnKF_matrix_identities}\ref{item:determinant_P}.
Denote $z = f(x_{t-1}^{(i)})$.
By the Woodbury matrix identity as well as \Cref{lemma:EnKF_matrix_identities}\ref{item:posterior_covariance_Kalman} and \ref{item:PHKR},
\begin{alignat*}{3}
(H^\top R^{-1} H + Q^{-1})^{-1}
&= Q - QH^\top(HQH^\top + R)^{-1}HQ&
&= Q - KHQ&
&= \Sigma,
\\
H^\top R^{-1} y_t + Q^{-1} z
&= \Sigma^{-1} (Ky_t + (I_{d} - KH) z)&
&&&= \Sigma^{-1} \hat{m}_{t}^{(i)},
\\
R^{-1} - R^{-1}H \Sigma H^\top R^{-1}
&= R^{-1} - R^{-1} H K&
&= R^{-1}(S - HQH^\top)S^{-1}&
&= S^{-1},
\\
R^{-1}H \Sigma Q^{-1} 
&= R^{-1}(KR)^\top Q^{-1}&
&= S^{-1} H Q Q^{-1} &
&= S^{-1} H,
\\
Q^{-1} - Q^{-1} \Sigma Q^{-1}
&= Q^{-1} - Q^{-1}(I - KH)&
&= Q^{-1} K H &
&= H^\top S^{-1} H.
\end{alignat*}
Using these identities and completing squares yields
\begin{align*}
	I_t^{(i)}&(x) 
	\coloneqq |Hx-y_t|_R^2 + |x - z|_Q^2
	\\
	&=
	x^\top(H^\top R^{-1} H + Q^{-1}) x - 2 x^\top (H^\top R^{-1} y_t + Q^{-1} z) + |y_t|_{R}^2 + |z|_Q^2
	\\
	&=
	x^\top \Sigma^{-1} x - 2 x^\top \Sigma^{-1} \hat{m}_t^{(i)} + |\hat{m}_t^{(i)}|_{\Sigma}^2 - |\hat{m}_t^{(i)}|_{\Sigma}^2 + |y_t|_{R}^2 + |z|_Q^2
	\\
	&=
	|x-\hat{m}_t^{(i)}|_{\Sigma}^2
	+ y_t^\top (R^{-1} - R^{-1}H \Sigma H^\top R^{-1})y_t
	+ z^\top (Q^{-1} - Q^{-1} \Sigma Q^{-1}) z
	- 2y_t^\top R^{-1}H \Sigma Q^{-1}z
	\\
	&=
	|x-\hat{m}_t^{(i)}|_{\Sigma}^2 + y_t^\top S^{-1} y_t + z^\top H^\top S^{-1}Hz - 2y_t^\top S^{-1} Hz
	\\
	&=
	|x-\hat{m}_t^{(i)}|_{\Sigma}^2 + |y_t - Hz|_S^2.	
\end{align*}
It follows from the definitions in \eqref{equ:mixture_prior_approximation} and \eqref{equ:mixture_posterior_approximation} as well as \Cref{lemma:EnKF_matrix_identities}\ref{item:determinant_P} that
\begin{align*}
\posti(x)
&=
(2\pi)^{-d/2} \det Q^{-1/2} \exp \big( - \tfrac{1}{2} I_t^{(i)}(x) \big)
\\
&=
\det Q^{-1/2} \det \Sigma^{1/2}  \exp \big( - \tfrac{1}{2} |y_t - Hz|_S^2 \big) \cdot \mathcal{N}(\hat{m}_{t}^{(i)}, \Sigma)
\\
&=
\biggl(\frac{\det R}{\det S}\biggr)^{\! \! 1/2} 
\exp \big( - \tfrac{1}{2} |y_t - Hz|_S^2 \big)
\cdot \mathcal{N}(\hat{m}_{t}^{(i)}, \Sigma).
\end{align*}
\end{proof}

\begin{lemma}[bound for target-proposal ratio]
\label{lemma:thm_IS_weight_bound_ratio_formula}
Under the assumptions of \Cref{thm:IS_weight_bound} and using the notation therein, there exists $u_{t}^{(i)} \in \bR^{d}$ such that, for each $x \in \bR^{d}$,
\begin{equation}
	\frac{\posti}{q_{\mathsf{p}, t}^{(i)}} (x)
	\le
	\biggl(\frac{\det \Qip }{\det Q}\biggr)^{\! \! 1/2}
	\exp\big( -\tfrac{1}{2}(x-u_{t}^{(i)})^\top (\Sigma^{-1} - \Qip^{-1}) (x-u_{t}^{(i)}) \big),
\end{equation}
with equality if $\Cov^{\emp}[(f(x_{t-1}^{(i)})_{i=1}^N])$ is positive definite and $H$ has full row rank.
In other words, \Cref{thm:IS_weight_bound}\ref{item:thm_IS_weight_bound_ratio_formula} holds.
\end{lemma}

\begin{proof}
Recall the definition of $q_{\mathsf{p}, t}^{(i)} = \mathcal{N}(m_{\mathsf{p}, t}^{(i)}, \Qip)$, cf.\ \eqref{equ:qi_previous_particles}.
To simplify notation, we drop the indices $i$ and $\mathsf{p}$ throughout the proof and denote $z = f(x_{t-1}^{(i)})$ and $w = y_{t} - Hz$.
Since $A^\top(A A^\top)^{\dagger}A = A^{\dagger} A = \Proj_{\ran A^{\top}}$
for every matrix $A$, in particular for $A = \Delta_{K} S^{1/2}$, \Cref{lemma:EnKF_matrix_identities}\ref{item:covariance_difference} implies that
\[
S^{-1} - \Delta_{K}^{\top} \Delta_{\Sigma}^{\dagger} \Delta_{K}
=
S^{-1/2}(I_{m} - A^\top(A A^\top)^{\dagger}A)S^{-1/2}
=
S^{-1/2}(I_{m} - \Proj_{\ran (S^{1/2} \Delta_{K}^{\top}) })S^{-1/2}	
\in
\SPSD{m}
\]
is positive semidefinite.
Since $m_{t} - \hat{m}_{t} = \Delta_{K} \, w$ and $\ran \Delta_{K} = \ran \Delta_{\Sigma}$ by \Cref{lemma:EnKF_matrix_identities}\ref{item:covariance_difference}, it follows from \Cref{lem:quatratic_polynomials_differnece} that there exists $u \in \mathbb{R}^d$ such that
\begin{align*}
J_t(x)
&\coloneqq
|x-\hat{m}_{t}|_{\Sigma}^2 - |x-m_{t}|_{\Sigma_{t}}^2 + |w|_S^2
\\
&=
(x - u)^\top (\Sigma^{-1} - \Sigma_{t}^{-1})(x - u) - (m_{t} - \hat{m}_{t})^\top \Delta_{\Sigma}^{\dagger}(m_{t} - \hat{m}_{t}) + |w|_S^2
\\
&=
(x - u)^\top (\Sigma^{-1} - \Sigma_{t}^{-1})(x - u)
+
w^{\top} \big( S^{-1} - \Delta_{K}^{\top} \Delta_{\Sigma}^{\dagger} \Delta_{K} \big) w
\\
&\geq
(x - u)^\top (\Sigma^{-1} - \Sigma_{t}^{-1})(x - u),
\end{align*}
with equality if and only if $w \in S^{1/2} \ran (S^{1/2} \Delta_{K}^{\top}) = \ran (H \Delta_{Q} (I_{d} - K_{t} H)$, where we used \Cref{lemma:EnKF_matrix_identities}\ref{item:Kalman_difference_full_rank}.
$(I_{d} - K_{t} H)$ is invertible by the Weinstein--Aronszajn identity,
\[
\det(I_{d} - K_{t} H)
=
\det(I_{m} - HK)
=
\det(S_{t} - H Q_{t} H^{\top}) \det S_{t}^{-1}
=
\det R \, \det S_{t}^{-1}
>
0,
\]
and therefore the above equality follows from the stated conditions: it holds whenever $H$ has full row rank and $\Cov^{\emp}[(f(x_{t-1}^{(i)})_{i=1}^N])$ is positive definite.

Setting $u_{t}^{(i)} \coloneqq u$, the claim then follows from \Cref{thm:IS_weight_bound}\ref{item:thm_IS_weight_bound_posti_formula} and \Cref{lemma:EnKF_matrix_identities}\ref{item:determinant_P}:
\[
\frac{\posti}{q_{\mathsf{p}, t}^{(i)}}(x)
=
\biggl( \frac{\det R \det \Sigma_{t}}{\det S \det \Sigma}\biggr)^{\! \! 1/2}  \exp \big( - \tfrac{1}{2} J_t(x) \big)
=
\biggl( \frac{\det \Sigma_{t}}{\det Q}\biggr)^{\! \! 1/2}  \exp \big( - \tfrac{1}{2} J_t(x) \big).
\]
\end{proof}

\begin{lemma}[bound for IS weights]
\label{lemma:thm_IS_weight_bound_IS_weight_bound}
Let the assumptions of \Cref{thm:IS_weight_bound} hold and consider the notation therein. 
For each $x\in \mathbb{R}^d$ and $i =1,\dots,N$,
\begin{equation}
	\label{equ:weight_estimation_repeated}
	\frac{\posti}{Z^{\mathsf{mix}} \, \qip}(x)
	\leq
	\biggl(\frac{\det \hat{C}^{x,\mathsf{p}}_{t} }{\det Q }\biggr)^{\! \! 1/2}
	\biggl(\frac{1}{N} \sum_{j=1}^N
	\exp\big( -\tfrac{1}{2} |y_t - Hf(x_{t-1}^{(j)})|_S^2 \big)\biggr)^{-1}.
\end{equation}
Further, for the sampling schemes $\II_{\mathsf{p}}$ and $\MM_{\mathsf{p}}^{\strat}$ defined in \Cref{algo:importance_sampling_strategies}, we have $\sigma_t^2 < \infty$ whenever $f$ is bounded.
In other words, \Cref{thm:IS_weight_bound}\ref{item:thm_IS_weight_bound_IS_weight_bound} holds.
\end{lemma}

\begin{proof}
Since $Q_{t} \succcurlyeq Q$, $S_{t} \succcurlyeq S$ and
\[
\tilde{\Sigma}_{t}
\coloneqq
(I - K_t H) Q_{t} (I - K_t H)^\top + K_t R K_t^\top
=
\Sigma_{t} + (I - K_t H) \Delta_{Q} (I - K_t H)^\top
\succcurlyeq
\Sigma_{t},
\]
where $\succcurlyeq$ denotes the Loewner order,
a proof strategy identical to \Cref{lemma:EnKF_matrix_identities}\ref{item:determinant_P} yields
\begin{equation}
	\label{equ:bound_for_det_product}
	\det \Sigma_{t} \det S
	\leq
	\det \tilde{\Sigma}_{t} \det \tilde{S}
	=
	\det Q_{t} \det R.
\end{equation}
By \Cref{lemma:EnKF_matrix_identities}\ref{item:covariance_difference}, $\Delta_{\Sigma} \in \SPSD{d}$, which implies that $\Sigma^{-1} - \Sigma_{t}^{-1} = \Sigma_{t}^{-1} \Delta_{\Sigma} \Sigma^{-1} \in \SPSD{d}$ holds,
and \eqref{equ:relation_single_posterior_proposal_approx}, 
\[
    \posti
    \le \biggl(\frac{\det \Sigma_{t} }{\det Q }\biggr)^{\!\! 1/2} q_{\mathsf{p}, t}^{(i)}.
\]
Moreover, by \eqref{equ:single_posterior_approx_density},
\[
    \Zmix 
    = \frac{1}{N} \sum_{i=1}^N \int p_t^{(i)}
    = \biggl(\frac{\det R }{\det S }\biggr)^{\!\! 1/2}\frac{1}{N} \sum_{i=1}^N 
    \exp\big( -\tfrac{1}{2}|y_t - Hf(x_{t-1}^{(i)})|_S^2 \big).
\]
Combining these observations, the ratio is bounded as
\[
\frac{\posti}{\Zmix \, \qip}
\leq
\biggl(\frac{\det \Sigma_{t} \det S }{\det Q \det R }\biggr)^{\! 1/2} 
\biggl(\frac{1}{N} \sum_{j=1}^N \exp\big( -\tfrac{1}{2}|y_t - Hf(x_{t-1}^{(j)})|_S^2 \big) \biggr)^{\!-1}
\]
and \eqref{equ:bound_for_det_product} proves \eqref{equ:weight_estimation_repeated}.
Further, if $f$ is bounded, then
\[
c_t' 
\coloneqq \sup_{(x^{(i)})_{i=1}^N \in (\bR^d)^N} \frac{\det \hat{C}^{x,\mathsf{p}}_{t} }{\det Q }
< \infty,
\qquad 
c_t'' \coloneq
\sup_{x\in\bR^d} e^{\, |y_t - Hf(x)|_S^2}
< \infty,
\]
where we used that $\hat{C}^{x,\mathsf{p}}_{t} = Q + \Cov^\emp[(f(x^{(i)}))_{i=1}^N]$ by \eqref{equ:previous_based_gain_linear}.
Therefore, for the scheme $\II_{\mathsf{p}}$, using \Cref{lem:conditional_expectations},
\[
\sigma_t^2
=
\mathbb{E}\biggl[\biggl(\frac{\chposti}{\Zmix\, \chqi}\biggr)^2\biggr] - 1
\le
c_t'\,c_t'' - 1
<
\infty.
\]
For the scheme $\MM_{\mathsf{p}}^\strat$, $\sigma_t < \infty$ follows by \Cref{thm:weight_variance_inequality}.
\end{proof}

\begin{lemma}
	\label{lem:quatratic_polynomials_differnece}
	Let $A, B \in \SPD{d}$ be symmetric and positive definite matrices, $a,b,v \in \bR^{d}$ such that $(B - A) v = b - a$, and let $u = b - Bv$.	
	Then, for all $x \in \mathbb{R}^d$, 
	\[
	|x - a|_A^2 - |x - b|_B^2 = (x - u)^\top (A^{-1} - B^{-1})(x - u) - (b - a)^\top (B - A)^{\dagger}(b - a),
	\]
	where $(B - A)^{\dagger}$ denotes the Moore–-Penrose pseudoinverse.
\end{lemma}

\begin{proof}
Set $M \coloneqq A^{-1} - B^{-1}$ and observe that
\[
A^{-1}(b - a)
=
A^{-1}(B - A) v
=
(A^{-1} - B^{-1}) Bv
=
M B v.
\]
It follows for every \(x \in \mathbb{R}^d\):
\begin{align*}
	|x - a|_A^2 - |x - b|_B^2 	
	&= |x - b + b - a|_A^2 - |x - b|_{B}^2
	\\
	&= (x - b)^\top (A^{-1} - B^{-1})(x - b) + 2(x - b)^\top A^{-1}(b - a) + |b - a|_{A}^2
	\\
	&= (x - b)^\top M(x - b) + 2(x - b)^\top M B v + (b - a)^{\top} A^{-1} (b-a)
	\\
	&= (x - b + B v)^\top M(x - b + B v) - v^\top B M B v + v^\top (B-A) M B v
	\\
	&= (x - u)^\top M(x - u) - v^\top A M Bv.
\end{align*}
This proves the claim since
\[
v^\top A M Bv
=
v^\top (B-A) v
=
v^\top (B-A) (B-A)^{\dagger} (B-A) v
=
(b-a)^{\top} (B-A)^{\dagger} (b-a).
\]
\end{proof}

\section*{Acknowledgments}
\addcontentsline{toc}{section}{Acknowledgements}
The authors acknowledge funding by the Deutsche Forschungsgemeinschaft (DFG, German Research Foundation) --- CRC/TRR 388 ``Rough Analysis, Stochastic Dynamics and Related Fields'' --- Project ID 516748464.
The authors thank Jana de Wiljes and Tim Sullivan for helpful discussions and suggestions.

\bibliographystyle{abbrvnat}
\bibliography{myBibliography}
\addcontentsline{toc}{section}{References}

\end{document}

%% file: Packages.tex
\usepackage{amsfonts,amssymb,amsmath,amsthm,amstext,amssymb,mathtools}
\usepackage{subcaption,float,color,graphicx,enumitem}
\usepackage{dsfont}  
\usepackage[normalem]{ulem} 
\usepackage[authoryear,sort,round]{natbib}  
\usepackage{fullpage} 
\usepackage{hyperref,nicefrac}
\usepackage{algorithm}
\usepackage{algpseudocode}
\usepackage{booktabs,multirow, makecell}
\usepackage{aliascnt}
\usepackage[noabbrev,capitalise,nosort,nameinlink]{cleveref}

\usepackage[table]{xcolor}

\usepackage{adjustbox}
\usepackage{tcolorbox}
\usepackage{tikz}
\usetikzlibrary{bayesnet,er,trees,shapes.symbols,mindmap,arrows,decorations,shapes.misc,shapes.arrows,chains,matrix,positioning,scopes,decorations.pathmorphing,patterns}
\usepackage{tikz-cd} 

%% file: Macros.tex
\newcommand*{\bR}{\mathbb{R}}

\newcommand*{\bN}{\mathbb{N}}

\newcommand*{\bE}{\mathbb{E}}
\newcommand*{\bV}{\mathbb{V}}

\newcommand*{\cN}{\mathcal{N}}

\newcommand*{\cP}{\mathcal{P}}
\newcommand*{\Cov}{\mathbb{C}\mathrm{ov}}

\makeatletter
\newcommand*\quark{\mathpalette\quark@{.5}}
\newcommand*\quark@[2]{\mathbin{\vcenter{\hbox{\scalebox{#2}{$\; \m@th#1\bullet \;$}}}}}
\makeatother

\makeatletter
\newcommand{\mylabel}[2]{#2\def\@currentlabel{#2}\label{#1}}
\makeatother

\newcommand*{\rd}{\mathrm{d}}

\DeclarePairedDelimiterX{\infdivx}[2]{(}{)}{#1\;\delimsize\|\;#2}

\DeclareMathOperator{\Id}{Id}

\definecolor{darkgreen}{rgb}{0,0.4,0}

\newcommand*{\todo}[1]{\bgroup\color{red}TODO:~#1\egroup}
\newcommand*{\pk}[1]{\bgroup\color{orange}HL:~#1\egroup}
\newcommand*{\ik}[1]{\bgroup\color{cyan}IK:~#1\egroup}
\newcommand*{\dw}[1]{\bgroup\color{violet}DW:~#1\egroup}
\newcommand*{\cls}[1]{\bgroup\color{blue}ClS:~#1\egroup}

\DeclarePairedDelimiter\abs{\lvert}{\rvert}%
\DeclarePairedDelimiter\mynorm{\lVert}{\rVert}%
\makeatletter
\let\oldabs\abs
\def\abs{\@ifstar{\oldabs}{\oldabs*}}
\let\oldnorm\mynorm
\def\mynorm{\@ifstar{\oldnorm}{\oldnorm*}}
\makeatother

\newcommand{\SPD}[1]{\mathbf{S}_{++}^{#1}}
\newcommand{\SPSD}[1]{\mathbf{S}_{+}^{#1}}

\newcommand{\mix}{\mathsf{mix}}
\newcommand{\strat}{\mathsf{str}}

\newcommand{\mytt}[1]{\textnormal{\texttt{#1}}}

\newcommand{\II}{\mytt{II}}
\newcommand{\IM}{\mytt{IM}}
\newcommand{\MI}{\mytt{MI}}
\newcommand{\MM}{\mytt{MM}}

\newcommand{\BPF}{\mytt{BPF}}
\newcommand{\EnKF}{\mytt{EnKF}}

\newcommand{\QMC}{\mytt{QMC}\text{-}}
\newcommand{\TQMC}{\mathsf{TQMC}}

\newcommand{\WEnKF}{\textsf{WEnKF}}

\newcommand{\mi}{m^{(i)}_{t}}
\newcommand{\mic}{m^{(i)}_{\mathsf{c},t}}
\newcommand{\mip}{m^{(i)}_{\mathsf{p},t}}
\newcommand{\Qi}{\Sigma_{t}}
\newcommand{\Qic}{\Sigma_{\mathsf{c},t}}
\newcommand{\Qip}{\Sigma_{\mathsf{p},t}}
\newcommand{\qi}{q^{(i)}_{t}}
\newcommand{\qic}{q^{(i)}_{\mathsf{c},t}}
\newcommand{\qip}{q^{(i)}_{\mathsf{p},t}}

\newcommand{\pim}{\pi^{\mathsf{mix}}_{t}}
\newcommand{\posti}{p^{(i)}_{t}}
\newcommand{\postm}{p^{\mathsf{mix}}_{t}}
\newcommand{\qm}{q^{\mathsf{mix}}_{t}}
\newcommand{\qmc}{q^{\mathsf{mix}}_{\mathsf{c},t}}
\newcommand{\qmp}{q^{\mathsf{mix}}_{\mathsf{p},t}}

\newcommand{\genq}{\mathfrak{q}}

\newcommand{\chposti}{\check{p}^{(i)}_{t}}
\newcommand{\chqi}{\check{q}^{(i)}_{t}}

\newcommand{\Zpost}{Z_t^{\mathsf{post}}}
\newcommand{\Zmix}{Z_t^{\mathsf{mix}}}

\newcommand{\ran}{\operatorname{ran}}
\newcommand{\Proj}{\operatorname{Proj}}

\newcommand{\smallgraybox}{\textcolor{lightgray}{\rule[-0.1ex]{0.45em}{1.2ex}}}
\newcommand{\biggraybox}{\textcolor{lightgray}{\rule[0ex]{0.6em}{1.5ex}}}

\newcommand{\emp}{\mathsf{emp}}

\theoremstyle{plain}

\newtheorem{theorem}{Theorem}[section]

\newaliascnt{proposition}{theorem}
\newtheorem{proposition}[proposition]{Proposition}
\aliascntresetthe{proposition}

\newaliascnt{lemma}{theorem}
\newtheorem{lemma}[lemma]{Lemma}
\aliascntresetthe{lemma}

\newaliascnt{corollary}{theorem}

\aliascntresetthe{corollary}

\theoremstyle{definition}

\newaliascnt{definition}{theorem}

\aliascntresetthe{definition}

\newaliascnt{assumption}{theorem}

\aliascntresetthe{assumption}

\newaliascnt{condition}{theorem}

\aliascntresetthe{condition}

\newaliascnt{remark}{theorem}
\newtheorem{remark}[remark]{Remark}
\aliascntresetthe{remark}

\newaliascnt{example}{theorem}
\newtheorem{example}[example]{Example}
\aliascntresetthe{example}

\newaliascnt{notation}{theorem}
\newtheorem{notation}[notation]{Notation}
\aliascntresetthe{notation}

\newaliascnt{conjecture}{theorem}

\aliascntresetthe{conjecture}

\newaliascnt{problem}{theorem}

\aliascntresetthe{problem}

\newaliascnt{openproblem}{theorem}

\aliascntresetthe{openproblem}

\crefname{theorem}{Theorem}{Theorems}
\Crefname{theorem}{Theorem}{Theorems}

\crefname{proposition}{Proposition}{Propositions}
\Crefname{proposition}{Proposition}{Propositions}

\crefname{lemma}{Lemma}{Lemmas}
\Crefname{lemma}{Lemma}{Lemmas}

\crefname{corollary}{Corollary}{Corollaries}
\Crefname{corollary}{Corollary}{Corollaries}

\crefname{definition}{Definition}{Definitions}
\Crefname{definition}{Definition}{Definitions}

\crefname{assumption}{Assumption}{Assumptions}
\Crefname{assumption}{Assumption}{Assumptions}

\crefname{condition}{Condition}{Conditions}
\Crefname{condition}{Condition}{Conditions}

\crefname{remark}{Remark}{Remarks}
\Crefname{remark}{Remark}{Remarks}

\crefname{example}{Example}{Examples}
\Crefname{example}{Example}{Examples}

\crefname{algorithm}{Algorithm}{Algorithms}
\Crefname{algorithm}{Algorithm}{Algorithms}

\crefname{notation}{Notation}{Notations}
\Crefname{notation}{Notation}{Notations}

\crefname{conjecture}{Conjecture}{Conjectures}
\Crefname{conjecture}{Conjecture}{Conjectures}

\crefname{problem}{Problem}{Problems}
\Crefname{problem}{Problem}{Problems}

\crefname{openproblem}{Open Problem}{Open Problems}
\Crefname{openproblem}{Open Problem}{Open Problems}

\newcommand{\norm}[1]{\lVert #1 \rVert}

\newcommand{\Norm}[1]{\left\Vert #1 \right\Vert}

\numberwithin{equation}{section}
\numberwithin{figure}{section}
\numberwithin{table}{section}

%% file: chunk-abstract.tex
The Bootstrap Particle Filter (BPF) and the Ensemble Kalman Filter (EnKF) are two widely used methods for sequential Bayesian filtering: the BPF is asymptotically exact but can suffer from weight degeneracy, while the EnKF scales well in high dimension (typically with localization) yet is exact only in the linear--Gaussian case. We combine these approaches by retaining the EnKF transport step and adding a principled importance-sampling correction.

Our first contribution is a general importance-sampling theory for \emph{mixture} targets and proposals, including variance comparisons between individual- and mixture-based estimators.
We then interpret the stochastic EnKF analysis as sampling from explicit Gaussian-mixture proposals obtained by conditioning on the current or previous ensemble, which leads to six self-normalized IS--EnKF schemes. We embed these updates into a broader class of ensemble-based filters and prove consistency and error bounds, including weight-variance comparisons and sufficient conditions ensuring finite-variance importance weights.

As a second contribution, we construct transported quasi-Monte Carlo (TQMC) point sets for the Gaussian-mixture laws arising in prediction and analysis, yielding TQMC-enhanced variants that can substantially reduce sampling error without changing the filtering pipeline.

Numerical experiments on benchmark models compare the proposed mixture-weighted and TQMC-enhanced filters, showing improved filtering accuracy relative to BPF, EnKF, and the standard weighted EnKF, and that the weighted schemes eliminate the EnKF error plateau often caused by analysis--target mismatch.

%% file: chunk-keywords.tex
{Data assimilation}%
\and%
{Sequential Bayesian filtering}%
\and%
{Ensemble Kalman filter}%
\and%
{Importance sampling}%
\and%
{Gaussian mixture}%
\and%
{Transported quasi-Monte Carlo}%

%% file: chunk-msc.tex
60G35  
\and
62M05  
\and
65C35  
\and
65C05  
\and
62F15  